\documentclass[prx,aps,twocolumn,10pt,superscriptaddress,dvipsnames,nofootinbib,amssymb,floatfix]{revtex4-2}
\usepackage{lipsum, babel}
\usepackage{amssymb}
\usepackage{amsmath,mathtools}
\usepackage{array}
\usepackage{geometry}
\usepackage{amsthm}
\usepackage{hyperref}
\hypersetup{colorlinks=true,citecolor=blue,linkcolor=blue,filecolor=blue,urlcolor=blue,breaklinks=true}
\usepackage[capitalize]{cleveref}
\usepackage{latexsym,bbm,bm}
\usepackage{algorithm,algorithmic}
\usepackage{enumerate}
\usepackage{float}
\usepackage{multirow}
\usepackage{fullpage}
\usepackage{makecell}
\usepackage{tikz}

\usepackage{thmtools, thm-restate}
\usepackage{mleftright}
\usepackage[caption=false]{subfig}
\usepackage{tabularray}
\UseTblrLibrary{booktabs}
\usepackage{xfrac}
\usepackage{siunitx}
\usepackage{xurl}

\usepackage{apptools}
\AtAppendix{
 \setcounter{equation}{0}
 \setcounter{proposition}{0}
 \setcounter{lemma}{0}
 \setcounter{corollary}{0}
 \setcounter{figure}{0}
 \setcounter{remark}{0}
 \setcounter{table}{0}

}

\usepackage{physics2}
\usephysicsmodule{ab}

\mleftright

\allowdisplaybreaks

\newcommand{\nc}{\newcommand}

\DeclarePairedDelimiter{\ket}{\lvert}{\rangle}
\DeclarePairedDelimiter{\bra}{\langle}{\rvert}
\DeclarePairedDelimiter{\norm}{\lVert}{\rVert}
\DeclarePairedDelimiter{\abs}{\lvert}{\rvert}

\DeclareMathOperator{\diag}{diag}

\usetikzlibrary{quantikz2}

\newtheorem{lemma}{Lemma}
\newtheorem{theorem}{Theorem}
\newtheorem{corollary}{Corollary}

\newtheorem{definition}{Definition}
\newtheorem{prob}{Problem}
\newtheorem{prob*}{Problem*}

\newtheorem{remark}{Remark}

\def\be{\begin{eqnarray}}
\def\ee{\end{eqnarray}}

\def\Gsoft{G^{\mathrm{soft}}}

\def\dpoly{\ell_{\mathrm{poly}}}

\def\Ures{U_{\mathrm{res}}}
\def\Uln{U_{\mathrm{LN}}}
\def\relu{\mathrm{ReLU}}
\def\Sinput{(\alpha_s,a_s,\epsilon_s)}

\def\Vinput{(\alpha_v,a_v,\epsilon_v)}
\def\Ginput{(\alpha_g,a_g,\epsilon_g)}
\def\Winput{(\alpha_w,a_w,\epsilon_w)}
\def\QKinput{(\alpha_{0},a_{0},\epsilon_{0})}

\usepackage{color}

\definecolor{Pr}{rgb}{0.4,0.3,0.9}

\definecolor{NG}{rgb}{0.66,0.29,0.48}
\newcommand{\naixu}[1]{{\color{Violet} #1}}

\nc{\tr}{\operatorname{tr}}
\nc{\polylog}{\operatorname{polylog}}
\nc{\ox}{\otimes}
\nc{\dg}{\dagger}
\nc{\dn}{\downarrow}
\nc{\cA}{{\mathcal A}}
\nc{\cB}{{\mathcal B}}
\nc{\cC}{{\mathcal C}}
\nc{\cD}{{\mathcal D}}
\nc{\cE}{{\mathcal E}}
\nc{\cF}{{\mathcal F}}
\nc{\cG}{{\mathcal G}}
\nc{\cH}{{\mathcal H}}
\nc{\cI}{{\mathcal I}}
\nc{\cJ}{{\mathcal J}}
\nc{\cK}{{\mathcal K}}
\nc{\cL}{{\mathcal L}}
\nc{\cM}{{\mathcal M}}
\nc{\cN}{{\mathcal N}}
\nc{\cO}{{\mathcal O}}
\nc{\cP}{{\mathcal P}}
\nc{\cQ}{{\mathcal Q}}
\nc{\cR}{{\mathcal R}}
\nc{\cS}{{\mathcal S}}
\nc{\cT}{{\mathcal T}}
\nc{\cU}{{\mathcal U}}
\nc{\cV}{{\mathcal V}}
\nc{\cX}{{\mathcal X}}
\nc{\cY}{{\mathcal Y}}
\nc{\cZ}{{\mathcal Z}}
\nc{\cW}{{\mathcal W}}

\nc{\RR}{{{\mathbb R}}}
\nc{\CC}{{{\mathbb C}}}
\nc{\FF}{{{\mathbb F}}}
\nc{\NN}{{{\mathbb N}}}
\nc{\ZZ}{{{\mathbb Z}}}
\nc{\PP}{{{\mathbb P}}}
\nc{\QQ}{{{\mathbb Q}}}
\nc{\UU}{{{\mathbb U}}}
\nc{\EE}{{{\mathbb E}}}

\newcommand{\utchem}{Department of Chemistry, University of Toronto, Toronto, Ontario M5G 1Z8, Canada}
\newcommand{\utcomp}{Department of Computer Science, University of Toronto, Toronto, Ontario M5S 2E4, Canada}
\newcommand{\vectorinst}{Vector Institute for Artificial Intelligence, Toronto, Ontario M5S 1M1, Canada}
\newcommand{\cifar}{Lebovic Fellow, Canadian Institute for Advanced Research, Toronto, Ontario M5G 1Z8, Canada}
\newcommand{\material}{Department of Materials Science \& Engineering, University of Toronto, Toronto, Ontario M5S 3E4, Canada}
\newcommand{\chemical}{Department of Chemical Engineering \& Applied Chemistry, University of Toronto, Toronto, Ontario M5S 3E5, Canada}
\newcommand{\aist}{
    Research Center for Emerging Computing Technologies, National Institute of Advanced Industrial Science and Technology (AIST), 1-1-1 Umezono, Tsukuba, Ibaraki 305-8568, Japan}
\newcommand{\keio}{
    Quantum Computing Center, Keio University, 3-14-1 Hiyoshi, Kohoku-ku, Yokohama, Kanagawa, 223-8522, Japan}
    
\date{\today}

\begin{abstract}

Powerful generative artificial intelligence from large language models (LLMs) harnesses extensive computational resources for inference.  
In this work, we investigate the transformer architecture, a key component of these models, under the lens of fault-tolerant quantum computing.
We develop quantum subroutines to construct the building blocks in the transformer, including the self-attention, residual connection with layer normalization, and feed-forward network.
As an important subroutine, we show how to efficiently implement the Hadamard product and element-wise functions of matrices on quantum computers. 
Our algorithm prepares an amplitude encoding of the transformer output, which can be measured for prediction or use in the next layer.
We find that the matrix norm of the input sequence plays a dominant role in the quantum complexity.
With numerical experiments on open-source LLMs, including for bio-informatics applications, we demonstrate the potential of a quantum speedup for transformer inference in practical regimes.
\end{abstract}

\begin{document}
\title{Quantum Transformer:\protect{\,}Accelerating model inference via quantum linear algebra}

\author{Naixu Guo}
\email{naixug@u.nus.edu}
\affiliation{Centre for Quantum Technologies, National University of Singapore, 117543, Singapore}

\author{Zhan Yu} 
\email{yu.zhan@u.nus.edu}
\affiliation{Centre for Quantum Technologies, National University of Singapore, 117543, Singapore}

\author{Matthew Choi}
\affiliation{\utcomp}
\affiliation{\vectorinst}

\author{Yizhan Han} 
\affiliation{School of Computing, National University of Singapore, 117417, Singapore}

\author{Aman Agrawal}
\affiliation{Department of Mathematics, National University of Singapore, 119076, Singapore}

\author{Kouhei Nakaji}
\affiliation{NVIDIA Corporation, 2788 San Tomas Expressway, Santa Clara, 95051, CA, USA}
\affiliation{\utchem}
\affiliation{\aist}
\affiliation{\keio}

\author{Al\'an Aspuru-Guzik}
\affiliation{NVIDIA Corporation, 2788 San Tomas Expressway, Santa Clara, 95051, CA, USA}
\affiliation{\utcomp}
\affiliation{\vectorinst}
\affiliation{\utchem}
\affiliation{\material}
\affiliation{\chemical}
\affiliation{\cifar}

\author{Patrick Rebentrost}
\email{cqtfpr@nus.edu.sg}
\affiliation{Centre for Quantum Technologies, National University of Singapore, 117543, Singapore}
\affiliation{School of Computing, National University of Singapore, 117417, Singapore}

\maketitle

\section{Introduction}

The transformer has emerged as the dominant architecture for large-scale generative artificial intelligence models
\cite{vaswani2017attention, openai2023gpt4}.
Designed to ``learn what to pay attention to'', the transformer employs self-attention mechanisms that effectively capture correlations between different parts of input sequences through dot-product computations \cite{bahdanau2015neural,vaswani2017attention}.
Transformers have been adopted for numerous downstream tasks, including text generation, question answering, and other domains like genomic data analysis \cite{radford2018improve, devlin2019bert, radford2019language, bubeck2023sparks,ji2021dnabert}.
A key challenge lies in the substantial computational resources required by transformer architectures ~\cite{patterson2021carbonemissionslargeneural}.
While training is resource-intensive, the cumulative cost of inference can significantly exceed it, as models trained once undergo extensive deployment~\cite{DESISLAVOV2023100857, McDonald_2022}.
These inference costs, in terms of both time and energy, are becoming increasingly acute, particularly with the rise of large-scale models performing complex reasoning tasks~\cite{openai2024openaio1card, deepseekai2025deepseekr1incentivizingreasoningcapability}.
Therefore, it is crucial to develop methods to enhance the efficiency of transformer inference.

Quantum computing has been investigated for a variety of linear-algebra tasks.
Seminal works are on the solution of linear systems and other matrix operations~\cite{harrow2009quantum,PhysRevLett.109.050505}, which can be applied to traditional machine learning methods such as the support vector machine and recommendation systems~\cite{PhysRevLett.113.130503,kerenidisQuantumRecommendationSystems2016}.
A quantum algorithm for optimizing neural networks by solving differential equations via linearization was shown recently~\cite{liu2024provably}.
Randomized classical algorithms show that using sampling-based input assumptions, quantum speedups are polynomial for many applications \cite{tang2019quantuminspire, Tang_2021, chia2022sampling}. 
Variational quantum circuits, as a quantum analog of neural networks, have been widely explored~\cite{peruzzo2014variational, PhysRevA.98.032309, PhysRevLett.122.040504, cerezo2021variational}, often without provable advantages~\cite{mcclean2018barren, larocca2024review, cerezo2024doesprovableabsencebarren}.
Significant progress in hardware has improved both the quantity and quality of quantum bits (qubits)~\cite{egan2021faulttolerant, acharya2023suppressing}, with recent experiments encoding tens of logical qubits~\cite{Bluvstein2024logical}.
Leveraging these advancements in quantum computation offers a promising pathway to potentially address the significant computational demands of advanced machine learning models.

In this work, we show progress towards an end-to-end transformer architecture implementable on a quantum computer. 
We work in the fault-tolerant model of quantum computation and use the modular framework of block encodings~\cite{low2017optimal, low2019hamiltonian, gilyen2019quantum}.
We assume a classical transformer architecture that has already been trained and focus on the inference process.
We develop efficient quantum subroutines for all the key building blocks of a transformer and combine them into a complete architecture.
Self-attention, residual connection with layer normalization, and feed-forward network are implemented via
the toolbox of quantum linear algebra, including our new method for implementing
element-wise functions of block-encoded matrices.
We analyze the run-time and input assumptions to verify the potential for a quantum speedup for both single- and multilayer structures, combined with performing various numerical experiments on several open-source large language and DNA models with size from millions to billions of parameters.
Hence, our algorithms promise fast inference with fault-tolerant quantum computers and could lead to cost savings in key applications.

\begin{figure*}[htpb!]
\includegraphics[width=0.9\linewidth]{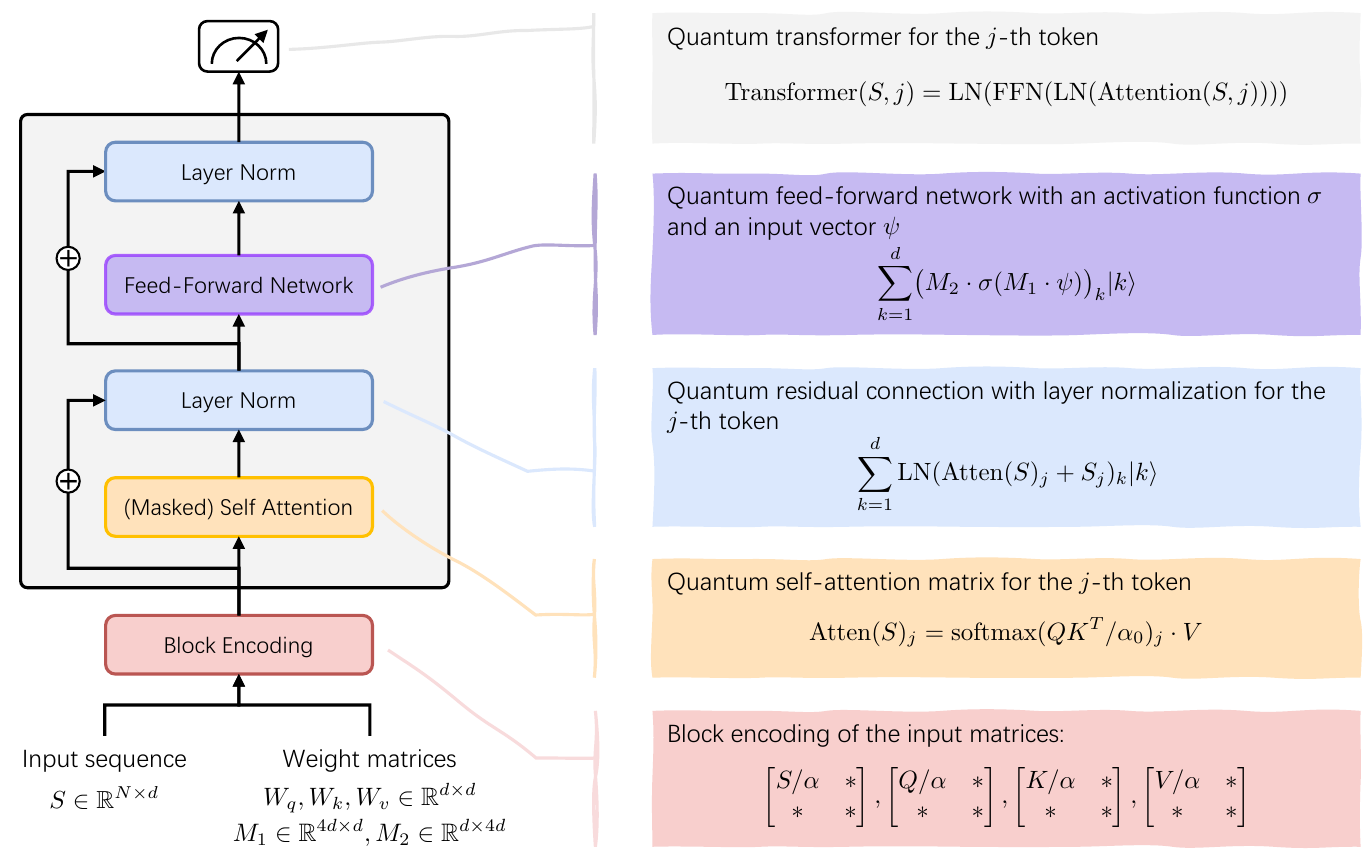}
\caption{\textbf{Overview of the quantum transformer architecture.} 
Same as the original decoder-only transformer architecture, the quantum transformer consists of a self-attention and a feed-forward network sub-layer, incorporating residual connection with layer normalization. The inputs of the quantum transformer are block encodings of the input sequence and pre-trained weight matrices, from which the relevant matrices for the transformer are constructed (query $Q$, key $K$, and value $V$). Given the input block encodings, we construct the corresponding quantum subroutines and combine them to our final result on obtaining the classical output vector corresponding to the $j$-th token. multilayer architecture can be achieved by iterating the procedure for each token $j\in [N]$ and producing a new block encoding of input sequence for the next layer.
}\label{figure_transformer_main}
\end{figure*}

\section{Results}

We first describe the inference of the pre-trained model.
The input to a transformer model typically consists of a sequence of $N$ tokens, each of which is represented by a $d$-dimensional vector via token embeddings~\cite{kudo2018sentencepiece, mielke2021words}, resulting in a matrix of the input sequence $S\in \mathbb{R}^{N\times d}$.
Note that in practice, $N$ is much larger than $d$.
A single-layer transformer consists of a self-attention sub-layer and a feed-forward network (FFN), both of which are followed by a residual connection with layer normalization (LN).
For the multilayer case, the computation is iterated several times to get the final output.
The output of the transformer is a $d$-dimensional vector corresponding to querying the $j$-th input token for $j \in [N]$, which can be further post-processed depending on the task it is applied to.
Formally, one can write the output vector as
\begin{equation*}
        \mathrm{Transformer}(S, j) \coloneqq \mathrm{LN}(\mathrm{FFN}(\mathrm{LN}(\mathrm{Atten}(S))))_j.
\end{equation*}
The subscript $j$ of a matrix denotes the $j$-th row of the matrix, and the subscript $j$ of a vector denotes the $j$-th element in the vector.

We propose the implementation of a single-layer transformer on a quantum computer,
which produces a quantum state corresponding to the output vector of the classical transformer. More details can be seen in \cref{figure_transformer_main}.
\begin{theorem}[Quantum transformer, informal]\label{thmTransformer.informal_main}
For a transformer with embedding dimension $d$ and an input sequence $S$ of length $N$, given access to the sequence matrix and weight matrices via block-encodings, for the index $j\in [N]$, one can construct a quantum circuit that prepares the state
\begin{align}
\sum_{k=1}^d \mathrm{Transformer}(S,j)_k \ket{k}, \label{eq.quantum.transformer}
\end{align}
up to error $\epsilon$ by using ${\mathcal{\widetilde{O}}}(\sqrt{N} d \log^2(1/\epsilon))$ times of the input block encodings. 
\end{theorem}
The classical vector $\mathrm{Transformer}(S,j)$ can be obtained by measuring the state in \cref{eq.quantum.transformer} \cite{Kerenidis2020Quantum}.
One can generalize to the multilayer architecture by iterating the subroutine for every token $j\in [N]$ in each layer.
The complexity of implementing the $k$-layer quantum transformer is then $\mathcal{\widetilde{O}}(kN^{\frac{3}{2}}d)$.
For the informal theorem, we assume that norms of the input sequence and weight matrices scale as $\cO(\sqrt{N})$ and $\cO(1)$ respectively, which will be verified by numerical analysis shown later.

\subsection{Quantum linear algebra}

Here we introduce the quantum linear algebra used to achieve \cref{thmTransformer.informal_main}.
Basic quantum computational steps include unitary multiplication, tensor products, partial measurements, and post-selection, with their associated cost regarding qubits and circuit complexity.
Quantum linear algebra aims to perform general computations including non-linear ones in a subspace via basic quantum operations.
The so-called block encoding is a suitable framework for exploring the power of quantum linear algebra \cite{gilyen2019quantum}.

\textit{Block encoding.} We say a unitary $U_A$ is a block encoding of matrix $A$ if 
\begin{align}
    U_A=
    \begin{bmatrix}
        A/\alpha & \cdot \\
        \cdot & \cdot
    \end{bmatrix},
\end{align}
where $\alpha$ is an encoding factor with $\alpha \geq \Vert A\Vert$.
Given such access,
one can multiply the matrix $A/\alpha$ to a quantum state via post-selection.
Note that unitary is a block encoding of itself by definition.
As a special case when $A$ is a ($L^2$-normalized) vector/state $\psi$, we say that $U_{\psi}$ is a block encoding of $\psi$.

\textit{Quantum singular value transformation (QSVT)} \cite{gilyen2019quantum}. Given a block encoding $U_A$ of Hermitian matrix $A$ with encoding factor $\alpha$ and an $\ell$-degree polynomial function $f$, one can construct a block encoding of $f(A/\alpha)$ 
\begin{align}
    U_A=
    \begin{bmatrix}
        A/\alpha & \cdot \\
        \cdot & \cdot
    \end{bmatrix}
    \Longrightarrow
    U_{f(A)}=
     \begin{bmatrix}
        f(A/\alpha) & \cdot \\
        \cdot & \cdot
    \end{bmatrix},
\end{align}
using $\mathcal{O}(\ell)$ times of  $U_A$.
This method can be used for matrix function based applications like Hamiltonian simulation and linear equation solver \cite{low2017optimal, low2017quantum, Childs_2017}.

Many applications require element-wise operations of matrices, including the self-attention mechanism in the transformer architecture, which cannot be directly achieved via QSVT.
Here, we extend the toolbox of quantum linear algebra
to implement element-wise functions of block-encoded matrices (see Supplementary Material (SM) $\text{III}$.A and Methods for details).

\begin{theorem}[Element-wise function of block encodings, informal\label{thm.element-wise}]
Given access to block encoding of matrix $A$ and an $\ell$-degree polynomial function $f_{\ell}$,
one can construct a block encoding of $f_{\ell}\circ(A/\alpha)$ by using $\mathcal{O}(\ell)$ times the input unitary, where $\circ$ denotes that the function is implemented element-wisely.
\end{theorem}
Note that this query complexity is independent of the dimension of the matrix if the polynomial has no constant term.

\subsection{Quantum transformer architecture}

Here, we describe how to implement blocks of the transformer via quantum circuits for linear algebra.
We assume that the inputs of the quantum transformer are the block encodings of input sequence matrix $S$, weight matrices $W_q,W_k,W_v$ and $M_1, M_2$.
We denote the encoding factor of the input sequence matrix and weight matrices by $\alpha_s, \alpha_w$ and $\alpha_m$, respectively. Given the sentence $S$, 
the convention is to call $Q\coloneqq S W_q$, $K\coloneqq SW_k$, and $V\coloneqq SW_v$ the query, key, and value matrices respectively. 
The target is to prepare a quantum state as \cref{eq.quantum.transformer}.

\textit{Quantum self-attention.} 
The scaled dot-product self-attention \cite{vaswani2017attention} is arguably the transformer's most important block, where correlations among the sequence are estimated.
The self-attention matrix is defined as
\begin{align}
    \mathrm{Atten}(S) \coloneqq \mathrm{softmax}(QK^T/\alpha_0)\cdot V.
\end{align}
The softmax function is implemented row-wise, which converts a real vector into a Gibbs distribution.
Here $\alpha_0$ is a scaling factor, where many works have shown different choices for $\alpha_0$ leading to performance improvements \cite{NEURIPS2021_8df7c2e3, ma2024era}.

Given the block encodings of the input sequence matrix and the weight matrices, for the index $j \in [N]$, we show that one can construct a block encoding of the $j$-th row vector of $ \mathrm{Atten}(S)$, denoted as $\mathrm{Atten}(S)_j$. 
The main challenge is implementing the softmax function.
We achieve this implementation by reducing it to a variant of Gibbs state preparation with element-wise function method described as \cref{thm.element-wise}.
Overall, we achieve the query complexity of this subroutine as $T_\mathrm{atten} = \mathcal{\widetilde{O}}(\alpha_s\alpha_w\log(1/\epsilon))$.
The details of the quantum self-attention and other variants like the masked self-attention can be seen in Methods and SM III.C.

\textit{Quantum residual connection with layer normalization.}
Given an index $j\in [N]$ and block encodings of row vectors $\mathrm{Atten}(S)_j$ and $S_j$, one can construct a block encoding of the (unnormalized) state
\begin{align}
    \sum_{k=1}^d \mathrm{LN}(\mathrm{Atten}(S)_j, S_j)_k \ket{k},
\end{align}
where $\mathrm{LN}(\cdot,\cdot)$ takes vectors as input, and standardizes the summed vector with zero mean and unit variance. Details are provided in SM III.D.
This subroutine uses $T_\mathrm{LN} = \mathcal{\widetilde{O}}(\sqrt{d})$ times of the input unitaries.

\begin{figure*}[t]
\includegraphics[width=0.9\linewidth]{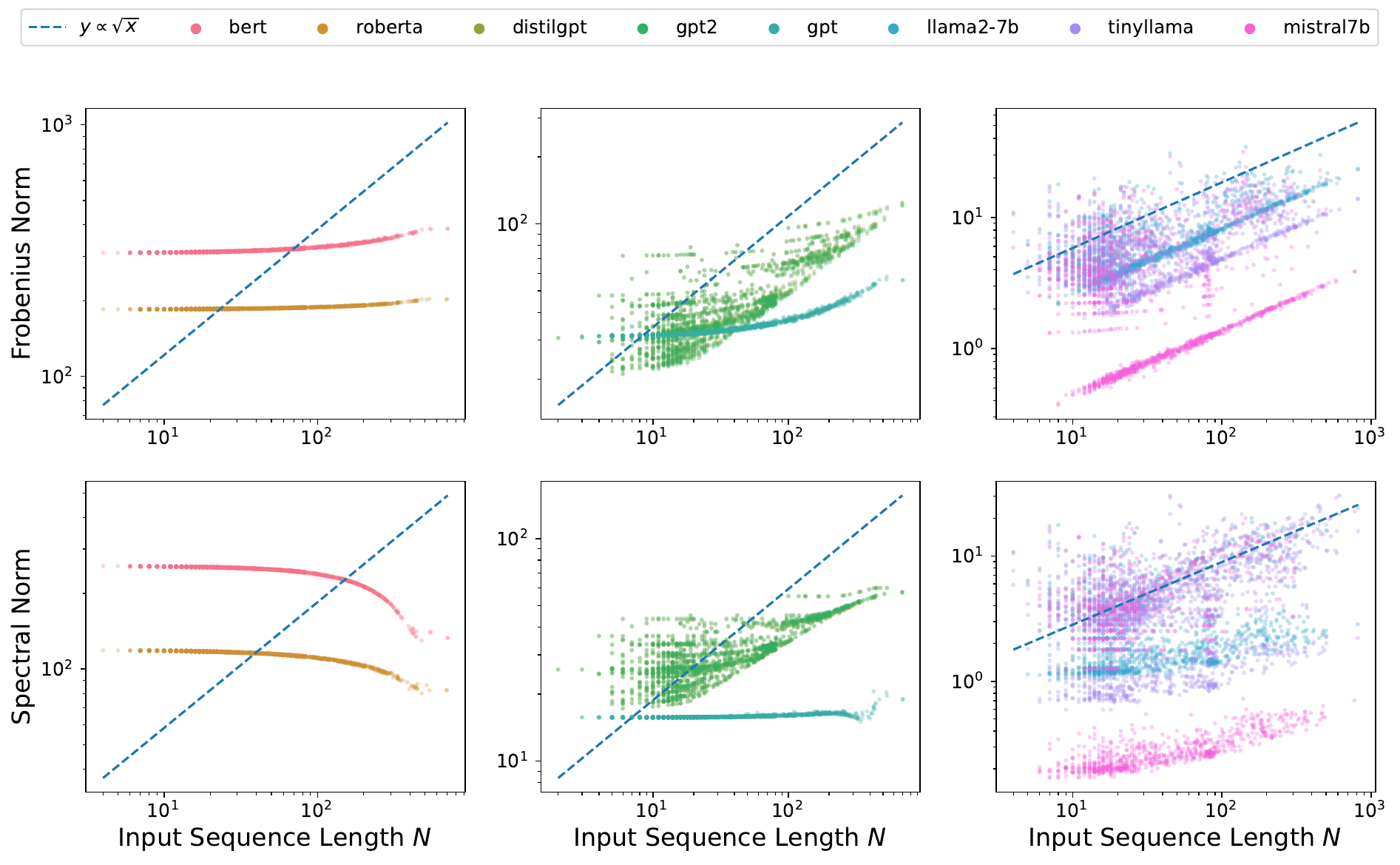}
  \caption{Scaling of the spectral norm $\|S\|$ and the Frobenius norm $\|S\|_{F}$ with $N$ for each model, displayed on logarithmic scales for both axes. For reference, the line $y \propto \sqrt{x}$ is also shown. We use tokens in MMLU dataset and convert them to $S$.
  The embedding dimension $d$ is $768$ for BERT \cite{devlin2019bert}, RoBERTa \cite{liu2019roberta}, GPT \cite{radford2018improve}, DistilGPT \cite{sanh2019distilbert} and GPT2 \cite{radford2019language}; $2048$ for TinyLlama \cite{zhang2024tinyllamaopensourcesmalllanguage}; and $4096$ for both Llama2-7B \cite{touvron2023llama2openfoundation} and Mistral-7B \cite{jiang2023mistral7b}.
  }
  \label{fig:MMLU_dataset_fro_spe_main}
\end{figure*}

\textit{Quantum feed-forward network.}
Given a state encoding $U_{\psi}$ of an $n$-qubit state $\ket{\psi}$ whose amplitudes are proportional to a vector $\psi$, an activation function $\sigma$, and real matrices $M_1$ and $M_2$, one can prepare a state encoding of the quantum state 
\begin{align}
\ket \phi =& \frac{1}{C}\sum_{k=1}^{d} \mathrm{FFN}(M_1,M_2,\psi)_k\ket{k},
\end{align}
where $\mathrm{FFN}(M_1,M_2,\psi) \coloneqq M_2\cdot\sigma(M_1\cdot\psi)$ and $C$ is the normalization factor.
This operation can be achieved by using the nonlinear amplitude transformation \cite{guo2021nonlinear, rattew2023nonlinear}.
Here, we explicitly consider the GELU (Gaussian Error Linear Units) function as the activation function $\sigma$, which has been widely used in practice \cite{hendrycks2017bridging}.
We achieve this step with query complexity independent of the embedding dimension $d$.
As the weight matrix normalization method has been well explored with many benefits \cite{miyato2018spectral, zhang2019selfattentiongenerativeadversarialnetworks}, we consider $\alpha_m=\mathcal{O}(1)$.
The subroutine therefore uses $T_\mathrm{FFN} = \mathcal{\widetilde{O}}(\log(1/\epsilon))$ times of the input block encodings. A proof sketch is provided in Methods and details are presented in SM $\text{III}$.E.

Combining all these blocks together, we can obtain the final target state
\begin{equation*}
    \sum_{k=1}^d \mathrm{Transformer}(S,j)_k \ket{k}
\end{equation*}
up to error $\epsilon$ using the input block encodings for $$T_\mathrm{atten} \cdot T_\mathrm{LN} \cdot T_\mathrm{FFN} \cdot T_\mathrm{LN} =  {\mathcal{\widetilde{O}}}(  \alpha_s\alpha_w d\log^2(1/\epsilon))$$ times in total.
The computation procedures can be mainly divided into two types: row-wise arithmetic, and matrix arithmetic.
The insight of our result is that we may achieve row-wise arithmetic with no dependency on the dimension, including the softmax and nonlinear activation function. 
However, matrix arithmetic like multiplication depends on the matrix norm (encoding factor), which limits the runtime of quantum transformer implementation.

\subsection{Numerical analysis}

In order to provide evidence of our assumptions, we perform numerical experiments on training and benchmarking several transformer-based large machine learning models with size from millions to billions of parameters.
The omitted details are provided in SM IV.

The complexity of our quantum transformer mainly depends on the encoding factors $\alpha_s$ and $\alpha_w$. The result in \cref{thmTransformer.informal_main} holds directly under the assumption of $\alpha_s = \cO(\sqrt{N})$ and $\alpha_w = \cO(1)$. For an arbitrary matrix $A$, one can construct its block encoding with an encoding factor $\alpha=O(\|A\|_F)$ given access to quantum Random Access Memory (QRAM) and a quantum data structure \cite{lloyd2014quantum, kerenidisQuantumRecommendationSystems2016}. Also, recall that $\alpha\geq \|A\|$ by definition. 

\begin{table}[t]
    \centering
    \begin{tblr}{lccc}
    \toprule
    \textbf{Model} &
    \textbf{Dimension} $d$ & \textbf{Mean}  & \textbf{Var.}    \\
    \midrule
    GPT2 & $768$ & $3.6973$ & $1.5615$\\
    GPT2-medium & $1024$ & $3.4570$ & $0.8457$\\
    GPT2-large & $1280$ & $1.7617$ & $0.1929$\\
    GPT2-xl & $1600$ & $1.6289$ & $0.1406$\\
    \midrule
    TinyLlama & $2048$ & $0.6973$ & $0.1692$\\
    Llama2-7b  & $4096$ & $1.3486$ & $0.0901$\\
    Mistral-7b & $4096$ & $0.1576$ & $0.0047$ \\
    \bottomrule
    \end{tblr}
    \caption{The $L^2$-norm of column vectors in weight matrices from different large language models.}
    \label{table.norm}
\end{table}
We first investigate the input sequence matrix $S \in \RR^{N \times d}$, which introduces the dependency on $N$.
We consider input data in real-world applications sampled from the widely-used Massive Multitask Language Understanding (MMLU) dataset~\cite{hendry2021measuring}.
The scaling of the spectral norm and the Frobenius norm of $S$ in the MMLU dataset is demonstrated in \cref{fig:MMLU_dataset_fro_spe_main}. We can find that the matrix norms of the input matrix of all LLMs scale at most as $\mathcal{O}(\sqrt{N})$. 
From \cref{fig:MMLU_dataset_fro_spe_main}, one can observe that the matrix norms do not show a dependence on $d$, since the Llama2-7b and Mistral-7b with larger embedding dimension have smaller matrix norms than other models like BERT and GPT.

Next we consider the norms of weight matrices $W_q,W_k,W_v \in \RR^{d \times d}$ in the large language models. We study the $L^2$-norm of the column vectors in the weight matrices with various embedding dimensions.
For each model, the mean and variance are calculated among the weight matrices and across all layers.
As shown in \cref{table.norm}, there is a trend that as the embedding dimension $d$ increases, both the mean and variance of the $L^2$-norm of column vectors in weight matrices decrease, especially for the GPT2 family \cite{radford2019language}.
Thus one can reasonably assume that the $L^2$-norm of column vectors is upper bounded by a constant that is independent of $d$. By direct calculation, the Frobenius norm of the weight matrices scales as $\cO(\sqrt{d})$, and so does the encoding factor $\alpha_w$.

Furthermore, inspired by the idea of normalizing weight matrices in generative adversarial networks~\cite{miyato2018spectral}, we 
train transformers with sub-normalized weight matrices for both self-attention and feed-forward network.
We perform the promoter detection task on the Genomic Benchmarks dataset~\cite{gresova2023genomic}.
The results can be seen in \cref{table.nontata}.
We observe that the matrix normalization does not affect much the performance in the single-layer case. We further train the multilayer normalized transformer and find its performance is comparable to other advanced multilayer models.
Therefore, it is reasonable to normalize the weight matrices so that $\alpha_w = \alpha_m= \cO(1)$, which enables the quantum transformer to run faster without loss of performance.

\begin{table}[htbp!]
    \centering
    \begin{tblr}{lc}
    \toprule
    \textbf{Model} &
    \textbf{Nontata Accuracy}     \\
    \midrule
      Single-layer transformer & 89.1  \\
    Single-layer SN transformer & 88.4   \\ 
    Single-layer FN  transformer& 87.7   \\
    \midrule
    CNN & 85.1 \cite{nguyen2023hyenadna}\\
    HyenaDNA & 96.6 \cite{nguyen2023hyenadna} \\
    DNABERT & 92.6 \cite{gresova2023genomic}\\
    Multilayer FN transformer & 92.1 \\
    \bottomrule
    \end{tblr}
    \caption{Benchmarks of different large machine learning models on the Genomic Benchmarks (GB) dataset. ``SN'' and ``FN'' stand for spectral-normalized and Frobenius-normalized respectively.
    The multilayer FN transformer has the same size of parameters with DNABERT.
    }
    \label{table.nontata}
\end{table}

\subsection{Runtime and speedup}

Combining the theoretical analysis on the runtime of quantum transformers and numerical observations of $\alpha_s = \cO(\sqrt{N})$ and $\alpha_w = \cO(1)$, we obtain the query complexity of the quantum transformer being $\widetilde{\mathcal{O}}(\sqrt{N} d)$ as in \cref{thmTransformer.informal_main}, where $d$ is the embedding dimension and $N$ is the input sequence length. We continue with a discussion on the time complexity so that a fair comparison with classical models can be made. With the QRAM assumption, the input block encodings can be implemented in $\cO(\polylog N)$ time.  Even without a QRAM assumption there can be cases when the input sequence is generated  efficiently, for example when the sequence is generated from a differential equation, see additional discussions in SM $\text{IV}.$C.
In these cases that the input block encodings are efficiently prepared, the time complexity of quantum transformers is in the nearly same order as the query complexity.

We first consider the simplest task of using a single-layer transformer to produce an output vector $\mathrm{Transformer}(S,j)$ by querying the $j$-th input token. By analyzing the naive matrix multiplication $V=SW_v$, the runtime of classical single-layer transformer inference is $\cO(Nd^2)$.
Note that for the single-layer structure, one only needs to compute a single row vector of $\mathrm{softmax}(QK^{T}/\alpha_0)$. 
One can find that the quantum transformer provides a nearly quadratic speedup over the classical counterpart. 

In more general cases of multilayer architecture, for the input sequence of length $N$, the transformer must then output $N$ vectors for the input of the next layer. The time complexity of quantum transformers is $\mathcal{\widetilde{O}}(N^\frac{3}{2} d)$ using the $L^{\infty}$ tomography method \cite{Kerenidis2020Quantum}, whereas the classical transformer runs in $\cO(N^2 d + N d^2)$ time.
Since $N$ is much larger than $d$ in practical scenarios, one can see that quantum transformer still provides a speedup over the classical counterpart but less than quadratic. 
See Methods for more detailed discussions.

Note that several quantum machine learning algorithms that showed promises of exponential advantages have been dequantized \cite{tang2019quantuminspire, Tang_2021, chia2022sampling}. For the transformer, we rigorously analyze classical randomized algorithms. The analysis indicates that there exists a polynomial separation on the query complexity of quantum and classical algorithms in terms of the dependency on matrix norms, hence our algorithm is robust to dequantization. A sketch of analysis can be found in Methods and detailed proofs are in SM IV.D.

\section{Discussion}

In this work, we show progress towards accelerating the inference of transformer architectures on fault-tolerant quantum computers.
We show how to formulate and achieve each computation block of the transformer as quantum subroutines, which can be further combined in a modular fashion.
The ability to obtain a quantum advantage hinges on how the input is given and the particular machine learning problem.
We have discussed the relevant input quantities for all quantum subroutines and their behavior in real-world large-language models. 
Based on comprehensive theoretical and numerical studies, we demonstrate the potential for a polynomial quantum speedup in the input sequence length.
The main subroutines are efficient such that
in principle these subroutines allow for other, broader regimes of quantum speedups.
We believe that our work shows novel directions on how quantum computing may enhance state-of-the-art machine learning models.

\section{Methods}
\subsection{Proof sketch of \cref{thm.element-wise}}
The intuition is that the element-wise function can be decomposed into the linear combination of matrices: $f_{\ell}(A/\alpha)=\sum_{j=1}^{\ell} c_j (A/\alpha)^{\circ j}$.
Therefore, if we can achieve the Hadamard product of block encodings, combining with linear combination of unitaries (LCU)  \cite{childs2012hamiltonian}, we can achieve the element-wise function.

For the Hadamard product $A_1 \circ A_2$, note that all elements are contained in the tensor product $A_1\otimes A_2$.
The next step is to find unitaries that arrange the elements into a particular block of the matrix, i.e., to find $P$ such that 
\begin{align}
    P (U_{A_1}\otimes U_{A_2})P^\dagger=
    \begin{bmatrix}
        \frac{A_1\circ A_2}{\alpha_1\alpha
        _2} & \cdot \\
        \cdot & \cdot
    \end{bmatrix}.
\end{align}
Inspired by Ref.~\cite{zhao2021compiling}, we find that $P$ can be easily constructed using $n$ CNOT gates.

To combine with the LCU, we note that a similar trick mentioned as Lemma~$8$ in Ref.~\cite{Childs_2017} can also be applied here, which enables us to achieve the linear dependency on degree.
The coefficients $\{c_j\}_{j=1}^{\ell}$ are encoded using the quantum state preparation technique \cite{sun2023asymptotically, zhang2022quantum}, which is efficient in our case as the dimension is the polynomial degree.
Based on these, we can achieve the element-wise function as follows
\begin{align}
    c_1\begin{bmatrix}
        A/\alpha & \cdot \\
        \cdot & \cdot
    \end{bmatrix}+
    c_2\begin{bmatrix}
        \frac{A\circ A}{\alpha^2} & \cdot \\
        \cdot & \cdot
    \end{bmatrix}+\cdots\equiv
    \begin{bmatrix}
        f\circ(A/\alpha) & \cdot \\
        \cdot & \cdot
    \end{bmatrix}.
\end{align}

\subsection{Quantum self attention}
To achieve the quantum self-attention, we divide it into three steps.
First, we implement the element-wise function $e^x$ on $QK^T/\alpha_0$ via \cref{thm.element-wise} and polynomial approximation.
Then for index $j\in [N]$, we construct the state encoding of the quantum state
\begin{align}
    \ket{\mathrm{Atten}^{\mathrm{soft}}(S)_j}\coloneqq \frac{1}{\sqrt{Z_j}}\sum_{k=1}^N \sqrt{\mathrm{softmax}\left(\frac{QK^T}{\alpha_0}\right)_j}\ket{k},
\end{align}
where $Z_j$ is the partition function for the $j$-th row of $\mathrm{softmax}(QK^T/\alpha_0)$.
Like other Gibbs state preparation algorithms \cite{gilyen2019quantum}, the query complexity is $\mathcal{O}(\sqrt{N/Z_j})$.
However, as $\alpha_0$ rescales the size of each element, the elements are lower bounded by a constant and therefore $\mathcal{O}(\sqrt{N/Z_j})=\mathcal{O}(1)$. 
Finally, we take the square of the state encoding by using the Hadamard product and multiply with matrix $V$.
The masked self-attention can be achieved by constructing and multiplying a block encoding of projectors, where details are provided in SM III.C.

\subsection{Quantum feed-forward network with GELU function}

The explicit representation of the GELU function is $\mathrm{GELU}(x)\coloneqq x\cdot \frac{1}{2}(1+\mathrm{erf}(\frac{x}{\sqrt{2}}))$.
We show that the GELU function can be well-approximated by a polynomial without the constant term.
For well-approximate we mean the degree of polynomial scales logarithmically to the precision.
In this case, the importance weighted method can be used to implement the GELU function on quantum states with no dependency on the dimension \cite{rattew2023nonlinear}.
It is also suitable to use the element-wise function as \cref{thm.element-wise}, yet the number of ancilla qubits is worse than the nonlinear amplitude transformation for a single state.

\subsection{Numerical details}
All data for the open-source models are obtained from Hugging Face \cite{devlin2019bert, radford2018improve, radford2019language, liu2019roberta, sanh2019distilbert, jiang2023mistral7b, touvron2023llama2openfoundation, zhang2024tinyllamaopensourcesmalllanguage}.
As another way to verify the matrix norm scaling of the sequence matrix $S$, we also compute the $L^2$-norm of token vectors in different models, which are found to be upper bounded by a constant.
The result can be seen SM IV.A.
Furthermore, for applications like retrieval-augmented generation (RAG) and other similarity estimation based tasks, token embeddings are typically 
$L^2$-normalized to unit length \cite{10.5555/3495724.3496517, karpukhin-etal-2020-dense}.
Since the input sequence contains $N$ tokens, we have $\alpha_s = \mathcal{O}(\sqrt{N})$.

For the training of DNA models, all experiments run on a single NVIDIA A100 SXM4 GPU.
We use the same tokenization as Ref.~\cite{zhou2024dnabert}.
We train the embedding and all following layers based on the training dataset provided in the Genomic Benchmarks (GB) \cite{gresova2023genomic}.
For the benchmarking, we consider the promoter detection task, which can be framed as a binary classification problem to determine whether a given DNA sequence region functions as a promoter.
Note that the GB dataset contains $36131$ sequences for promoter detection, with  $27097$ for training and others for validation and testing.
Explicitly, we sub-normalize the weight matrices in both self-attention and feed-forward network layers.
For the final output, we use a linear mapping and thresholding to achieve the classification.

\subsection{Quantum multilayer transformer}

Our method for the single-layer structure can be directly generalized to a multilayer structure, where the transformer outputs $N$ vectors as the input sequence of the next layer.
For all $j\in [N]$, prepare the corresponding quantum state and measure to obtain the output vector. 
Using the $L^{\infty}$ tomography method \cite{Kerenidis2020Quantum} enables the readout the $d$-dimensional vector with $\mathcal{O}(\log d)$ copies, with a precision bound for the $L^{\infty}$-norm (maximum absolute entry of the vector).
Note that for the classical transformer, the quantization method has been widely used \cite{thiruvathukal_low-power_2022, Ofir2019Q8BERT}, which trains with $32$-bit precision and performs inference with $4$- or $8$-bit precision.
Therefore, one can consider that the classical quantization method uses a constant precision in $L^{\infty}$-norm, which can be used for the quantum case as well.
Repeating the algorithm for $N$ times leads to the  complexity in $\mathcal{\widetilde{O}}(N^{\frac{3}{2}}d)$.
After obtaining $d$-dimensional vectors for all $N$ tokens, we construct the new block encoding for the next layer.
Since there are at most $Nd$ elements, this construction takes complexity $\mathcal{\widetilde{O}}(Nd)$.
For the $k$-layer architecture, the complexity is $\mathcal{\widetilde{O}}(k N^{\frac{3}{2}}d)$ in total.
If one considers implementing the multilayer structure fully coherently, the complexity will scale exponentially with $k$, which is much worse than the incoherent method presented in this work.

Analogous to the quantum linear equation solver \cite{harrow2009quantum} and quantum data fitting \cite{Wiebe_2012}, there could be an ideal regime $\|S\|_F=\mathcal{O}(\mathrm{polylog}(N))$ where we can achieve exponential speedup for single-layer and quadratic speedup for multilayer architecture compared to the classical standard algorithm. We leave further exploration of this regime for future work.

\subsection{Robustness to dequantization}
Here we briefly describe how we show a separation between the quantum and the classical randomized algorithm.
Similar to the QRAM assumption, the classical algorithm assumes the so-called sample and query (SQ) access \cite{tang2019quantuminspire}.
This input assumption assumes one can efficiently query each element of given vector and matrix, and can sample based on the $L^2$-norm and the Frobenius norm of the vector and the matrix, respectively.
To show the separation, we focus on the self-attention computation for comparison.
For the single-layer, the computation of self-attention can be decomposed as a matrix-vector multiplication $\mathrm{softmax}(QK/\alpha_0)_{j}\cdot V$, since we only focus on the $j$-th token.
Even if we assume the classical randomized algorithm can easily construct the SQ access of $\mathrm{softmax}(QK/\alpha_0)_{j\cdot}$,
by Ref.~\cite{Tang2023thesis}, it takes query complexity $\Theta(\|S\|_F^2\|W_v\|_F^2/\epsilon^2)$ to achieve the matrix-vector multiplication for the classical randomized algorithm.
The dependency on $\|S\|_F$ and $\|W_v\|_F$ is because $V$ is computed from $S$ and $W_v$.
Note that the complexity of our quantum single-layer transformer is $\mathcal{\widetilde{O}}(\alpha_s)=\mathcal{\widetilde{O}}(\|S\|_F)$, where we neglect the dependency on $d$ for simplicity.
Therefore, there is at least a quadratic separation on the matrix norm $\|S\|_F$, and our quantum algorithm cannot be effectively dequantized.

\section*{Data availability}
The full data of this work is available at Ref.~\cite{Guo2025}.

\section*{Code availability}
The full code for this work is available at Ref.~\cite{Guo2025}.

\bibliography{ref}

\section*{Acknowledgement}
This research is supported by the National Research Foundation, Singapore, and A*STAR under its CQT Bridging Grant and its Quantum Engineering Programme under grant NRF2021-QEP2-02-P05.
KN acknowledges the support of Grant-in-Aid for JSPS Research Fellow 22J01501.

\section*{Author contributions}
This project was conceived by N.G., Z.Y., and P.R.
Theoretical results were proved by N.G., Z.Y., and P.R.
The numerical experiments were conducted by M.C., Y.H., A.A., and K.N.
All authors contributed to the technical discussions and writing of this manuscript.

\appendix
\newpage
\onecolumngrid

\begin{center}
\textbf{
{\Large{Supplementary Material}}}
\end{center}

\tableofcontents

\section{Preliminary\label{sec.preliminary}}
\subsection{Notation}
We use the Dirac notation $\ket{\psi}$ to represent a vector with $\norm{\psi}_2=1$ (pure quantum state).
Denote by $\mathbb N$ the natural numbers $\{1,2, \cdots \}$.
For $N \in \mathbb N$, we use the notation $[N]$ to represent the set $\{1,\dots,N\}$.
For an $n$-qubit state $\ket{0}^{\otimes n}$, we write $\ket{0^n}$ for simplicity.
When there is no ambiguity, we may further ignore the superscript $n$ of $\ket{0^n}$.
For a matrix or an operator $A$, we use $A_{jk}\coloneqq \langle j|A|k\rangle$ to represent its $(j,k)$-th element, where $\{\ket{k}\}$ are the standard basis.
We use $A_{j\star}$ to represent its $j$-th row and $A_{\star k}$ to represent its $k$-th column.
The spectral norm, i.e., the largest singular value, is denoted by $\norm{A}$.
We write $\|A\|_F$ to represent the Frobenius norm.
For a normal matrix $A\coloneqq \sum_k \lambda_k(A)|\psi_k\rangle\langle\psi_k|$, with eigensystem $\{\lambda_k(A),|\psi_k\rangle\}$, and a function $f$, we write $f(A)\coloneqq \sum_k f(\lambda_k(A))|\psi_k\rangle\langle\psi_k|$ to represent the eigenvalue transformation of $A$ with $f$.
For a matrix $A$ and a function $f$, we use $f\circ(A)$ to represent the element-wise application of the function to the matrix, i.e., $\left(f\circ(A)\right)_{jk} = f(A_{jk})$. 

\subsection{Brief description about transformer\label{intro.transformer}}

The transformer is a key component of pretrained foundation models.
It has many applications and one of the main ones is the next token prediction, which has achieved great success in natural language processing.
Given a part of a sequence, the transformer aims to predict the next object of the sequence.  
The transformer is constructed by three main building blocks: self-attention, residual connection with layer normalization, and feed-forward networks (FFN). These building blocks will be described in this section.
The original paper \cite{vaswani2017attention} contains both the encoder and decoder parts.
Later many practically significant models only use one part, especially the decoder-only structure, which is shown in Fig.~\ref{figure_transformer}.

A key aspect of large-language models is \textit{tokenization}. The token is the basic unit of the transformer process.
Concepts like words, codes, and images can be converted to tokens with the so-called tokenization method \cite{sennrich2016neural, kudo2018sentencepiece, mielke2021words}. 
For the transformer, tokens are further mapped to real vectors via \textit{embedding} \cite{vaswani2017attention}. Let $d_{\mathrm{token}}$ be the number of tokens in the dictionary of the machine learning model and $d_{\mathrm{model}}$ be the dimension of the vectors of the embedding.
Let $\mathcal{W}\coloneqq \{\omega_j \in \mathbb{R}^{d_{\mathrm{model}}} : \omega_j \text{ is the embedding of token } j\in [d_{\mathrm{token}}] \}$ be the set of the embedding vectors of all tokens.
For simplicity, when we mention tokens in this paper, we directly mean their vector representations.
An $N$-length sentence is a sequence of vectors $\{S_j\}_{j=1}^N$, where $S_j\in \mathcal W$.
Due to the vector embeddings of the tokens, a sentence can also be understood as a real matrix $S\in \mathbb{R}^{N\times d_{\mathrm{model}}}$.

\textit{Self-attention\ ---} The correlations of the original concepts, such as words in natural languages, imply correlations of the corresponding tokens in the set of tokens.
Self-attention is the building block to encode such correlation information among tokens (vectors) into a new vector, which is the input vector for the next block.
The correlation is computed via estimating inner products. The block is also called the ``scaled dot-product attention".

There are three real parameterized (weight) matrices $W_q, W_k\in \mathbb{R}^{d_{\mathrm{model}}\times d_k}$ and $W_v\in \mathbb{R}^{d_{\mathrm{model}}\times d_v}$ arising in the self-attention block.
In practical cases, $d_{\mathrm{model}}=d_k=d_v$ is widely used, e.g., in the original paper~\cite{vaswani2017attention}.
In our discussion, we will keep this condition and write $d:=d_{\mathrm{model}}$ for simplicity.
Given the sentence $S$, 
the convention is to call $Q\coloneqq S W_q$, $K\coloneqq SW_k$, and $V\coloneqq SW_v$ the query, key, and value matrices respectively.
The attention block computes the matrix
$\Gsoft \in \mathbb R^{N \times d}$ such that 
\begin{align}
    \mathrm{Attention}(Q,K,V)=\mathrm{Attention}(S) =\mathrm{softmax}(Q K^{\top}/\alpha_0)V\eqqcolon \Gsoft,\label{matrix.self-attention}
\end{align}
where $\alpha_0>0$ is a scaling factor, and $\mathrm{softmax}(z)_j\coloneqq e^{z_j}/(\sum_{k\in[N]} e^{z_k})$ for $z\in\mathbbm R^N$ and $j\in [N]$.

\begin{figure}[htpb!]
\centering
\includegraphics[width=1\linewidth]{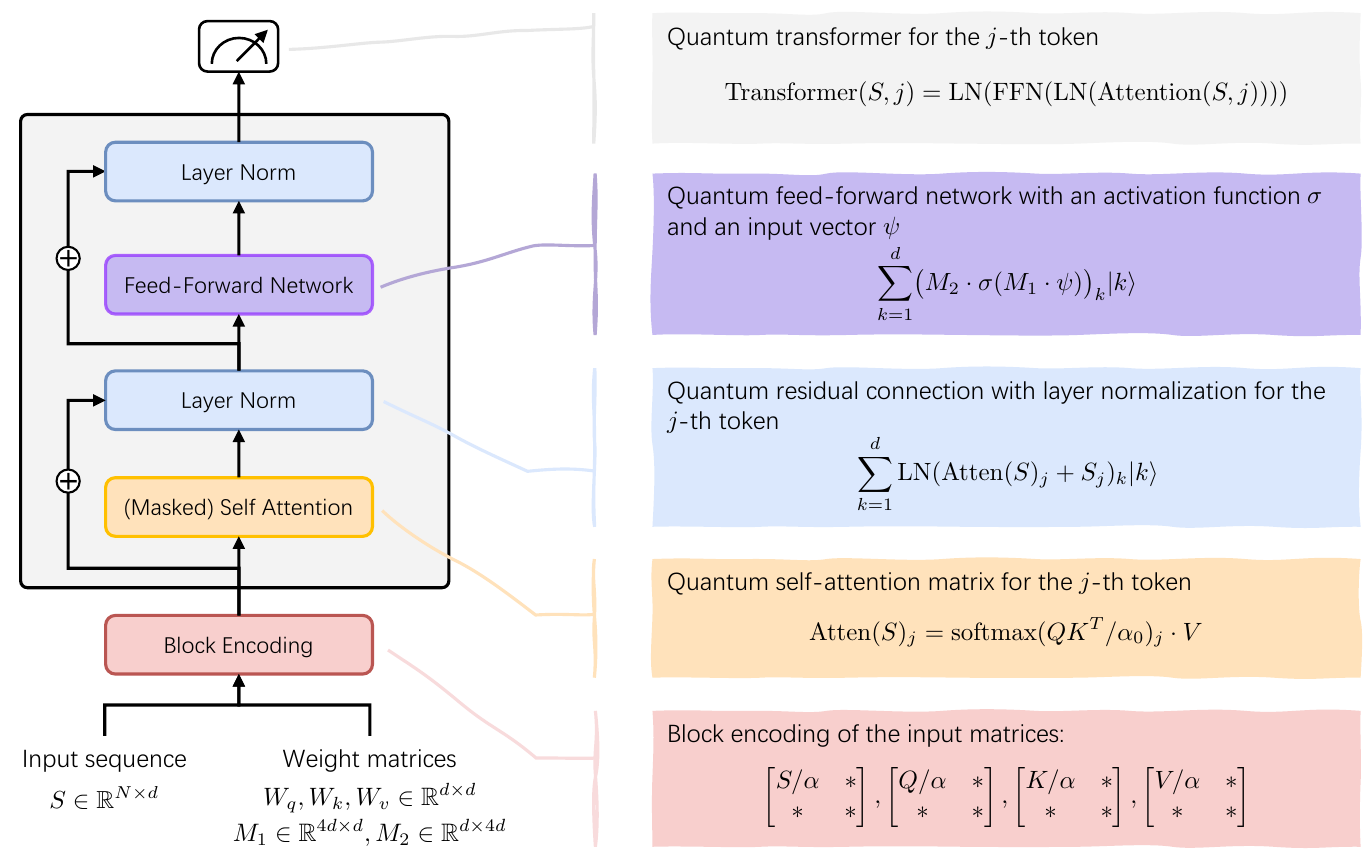}
\caption{\textbf{Overview of the quantum transformer architecture.} 
}\label{figure_transformer}
\end{figure}

In the attention block, the $\mathrm{softmax}$ is implemented for each row of the matrix $Q K^{\top}/\alpha_0$.
Formally, for a matrix $M\in\mathbbm R^{N\times N}$, it is defined as a row-wise application of the softmax function, i.e., $\mathrm{softmax}(M)_{ij}\coloneqq e^{M_{ij}}/(\sum_{k\in[N]} e^{M_{ik}})$ for $i,j\in [N]$.
The factor $\alpha_0$ controls that the exponentiated values are not too large.
The value $\alpha_0 = \sqrt{d}$ has been discovered to be a good choice in practice. To see this, assume that
each row of $Q$ and $K$ has zero mean and unit standard deviation.
Then for each element of $(QK^{\top})_{jk}=\sum_{m=1}^{d} Q_{jm}K_{km}$, the standard deviation will be bounded by $\sqrt{d}$. The coefficient rescales the standard deviation to $1$. Depending on the architecture and embeddings other scaling factors may also be employed \cite{yang2022tensor, ma2024era}. Inspired from the block-encoding discussion in this work, there is a natural choice for this scaling as we discuss in Section \ref{sec.attention}.

For $j \in [N]$, if the current query token is the $j$-th token $S_j$, the corresponding output vector is the $j$-th row of the self-attention matrix in \cref{matrix.self-attention}, denoted by $\Gsoft_j$.
More explicitly, the output vector of the self-attention layer for the $j$-th token is
\begin{align}
    \Gsoft_j= \sum_{k=1}^{d} \Gsoft_{jk}\hat e_k \equiv (\Gsoft)^{\top} \hat e_j 
    ,\label{vector.self-attention}
\end{align}
where $\{\hat e_j\}_{j=1}^N$ is the standard basis.
For the decoder-only structure which achieves the best practical performance, the so-called \textit{masked} self-attention is used,  which has the effect to mask or hide the tokens after the current query token.
This is achieved by adding a masked matrix $QK^{\top}\rightarrow QK^{\top} +M$, where 
\begin{align}
M_{jk}=\begin{cases}
0 & \quad k\leq j, \\
-\infty & \quad k>j.
\end{cases}\label{eq.mask}
\end{align}
Since $\exp(-\infty)=0$, tokens with index larger than $j$ receive no attention.
A further generalization called the \textit{multi-head} self-attention is based on computing several smaller attention matrices and concatenating them together.
The $h$-head self attention can be achieved with linear transformations $W_i^Q, W_i^K, W_i^V\in \mathbb{R}^{d\times\lceil\frac{d}{h}\rceil}$, and $W^{O}\in \mathbb{R}^{d\times d}$ for $i\in [h]$:
\begin{align}
  \mathrm{Multihead}(Q,K,V)=[\mathrm{head}_1,\dots, \mathrm{head}_h]W^{O} \in \mathbb R ^{N\times d},\notag
\end{align} 
where $\mathrm{head}_i=\mathrm{Attention}(QW_i^Q,KW_i^K,VW_i^V) \in \mathbb R^{N \times \lceil\frac{d}{h}\rceil}$.

\textit{Residual connection\ ---}
For a computation block like the self-attention, a residual connection with subsequent layer normalization is employed. This layer provides the ability to skip the computation block.
We take the self-attention as an example.
Note that if we focus on the $j$-th token, $S_j$ can be understood as the input and $\Gsoft_j\equiv \mathrm{Attention}(S,j)$ is the output vector of the self-attention block. The residual connection gives the output vector
$\Gsoft_j + S_j$\footnote{We note that this output vector can also be written as 
$\Gsoft_j(S) + \mathrm{Attention}(0, 0, S)^{\top}  \hat e_j$.}.
The next step is the layer normalization, which standardizes the vector.
Let $\Bar{s}_j :=\frac{1}{d}\sum_{k=1}^d(\Gsoft_{jk}+S_{jk})\cdot \vec 1$, where $\vec 1 
=(1,\dots,1)^T \in \mathbb R^d$ and $\varsigma := \sqrt{\frac{1}{d}\sum_{k=1}^d ((\Gsoft_{j}+S_{j}-\Bar{s}_j\cdot \vec 1)_k)^2}$.
The complete residual connection with the normalization layer can be expressed as
\begin{align}
    \mathrm{LN}_{\gamma,\beta}(\Gsoft_j,S_j)=\gamma\frac{\Gsoft_j+S_j-\Bar{s}_j\cdot \vec 1}{\varsigma}+\beta, \label{eq.Resnet}
\end{align}
where $\gamma$ is the scaling factor and $\beta \in \mathbb R^d$ is the bias vector.
For simplicity, we may not write these factors explicitly when there is no confusion.
We write $\mathrm{LN}_{\gamma,\beta}(\Gsoft_j,S_j)_k$ to represent the $k$-th element, i.e., $(\mathrm{LN}_{\gamma,\beta}(\Gsoft_j,S_j))_k$.
The role of layer normalization is to improve the trainability, which has been found essential for training deep neural networks in practice \cite{he2015deep, ba2016layer}.

\textit{Feed-forward network\ ---}
Finally, a two-layer fully-connected feed-forward network is implemented, i.e., 
\begin{align}
    \mathrm{FFN}(\mathrm{LN}(z_j,S_j))=\sigma(\mathrm{LN}(\Gsoft_j,S_j)M_1+b_1)M_2+b_2, \label{eq.neuralnetwork}
\end{align}
where $\sigma$ is an activation function, such as $\tanh(x)$ and $\relu(x)=\max(0,x)$.
Another activation function that may not be widely known, yet has been widely used in LLMs, is the Gaussian Error Linear Units function \cite{hendrycks2017bridging}.
Formally, we have $\mathrm{GELU}(x)\coloneqq x\cdot \frac{1}{2}(1+\mathrm{erf}(\frac{x}{\sqrt{2}}))$, where $\mathrm{erf}(x)\coloneqq \frac{2}{\sqrt{\pi}}\int_{0}^{x} e^{-t^2} \,dt$ is the error function. 
The function can be understood as a smoother $\relu$ activation function and will be our focus in the paper.
In addition, $M_1\in \mathbb{R}^{d\times d_{\mathrm{ff}}}, M_2\in \mathbb{R}^{d_{\mathrm{ff}}\times d}$ are linear transformation matrices, and $b_1,b_2$ are vectors.
In most practical cases, $d_{\mathrm{ff}}=4d$.

Combining these blocks together, we define the function
\begin{align}
    \mathrm{Transformer}(S, j) := \mathrm{LN}(\mathrm{FNN}(\mathrm{LN}(\mathrm{Attention}(S,j)))).
\end{align}
Note that inputs for each function can be recovered from matrix $S$, index $j$, and outputs from the previous layer functions.
In currently employed transformer architectures, several of these building blocks are iterated for a constant number of times.
The output, i.e., the next predicted token, is sampled from the distribution by further linear mapping the output vector to dimension $d_{\mathrm{model}}$ and implementing the softmax function.
Considering the run time,
recall that the length of the input sentence is $N$ and the dimension of the embedded vectors is $d$.
We summarize the time complexity as \cref{table1}.

\begin{table}[ht]
    \centering
    \begin{tblr}{lc}
    \toprule
    \textbf{Block}    &  \textbf{Time complexity} \\
    \midrule
    Preparation of $Q,K,V$ & $\mathcal{O}(N d^2)$\\
    Preparation of $QK^{\top}$ & $\mathcal{O}(N^2 d)$\\
    Preparation of $\mathrm{softmax}(Q K^{\top}/\sqrt{d})V\eqqcolon\Gsoft$ &$\mathcal{O}(N^2 + N^2d)$\\
    Residual connection $\mathrm{LN}(\Gsoft_j,S_j)$ & $\mathcal{O}(d)$\\
    Feed-forward NN ${\rm FFN}(\mathrm{LN}(\Gsoft_j,S_j))$ & $\mathcal{O}(Nd^2)$\\
    \bottomrule
    \end{tblr}
    \caption{Time complexity of transformer steps.}
    \label{table1}
\end{table}

The time complexity of a constant number of iterations of the three main blocks is $\mathcal{O}(N^2d+Nd^2)$, which mainly comes from the self-attention matrix computation.
If we only consider the $1$-layer transformer, the time complexity is $\mathcal{O}(Nd^2)$, as we do not need to compute all $N$ vectors that are needed for the second layer self-attention block.
This complexity comes from the matrix multiplication $V=SW_v$, as shown in \cref{table1}.

\subsection{Quantum procedures}

To encode the classical information into the quantum device, we use a standard input assumption in quantum algorithms literature, called the \textit{block encoding}. 
Note that the encoding can be generalized to non-square matrix cases of arbitrary size by padding the matrix with zeros.
Further, when we say we can construct or are given a block encoding unitary, it means we have access to the corresponding quantum circuit, i.e., we can also implement the controlled, self-adjoint, and controlled self-adjoint of the circuit.

\begin{definition}[Block encoding~\cite{chakraborty2019power, gilyen2019quantum}\label{def.blockencoding}]
We say a unitary $U_A$ is an $(\alpha,a,\epsilon)$-encoding of matrix $A\in \mathbb{C}^{2^n\times 2^n}$ if
\begin{align}
    \|A-\alpha (\bra{0^a}\otimes I_n)U_A(\ket{0^a}\otimes I_n)\|\leq \epsilon.
\end{align}
\end{definition}

By definition, one can see that $\alpha\geq \|A\|$, i.e., $\alpha$ is at least the spectral norm of the block-encoded matrix.
In \cref{sec.block-encoding}, we describe some methods to construct the block encoding for certain kinds of matrices, e.g., sparse.
Assuming the quantum random access memory \cite{giovannetti2009quantum} and quantum data structure \cite{kerenidisQuantumRecommendationSystems2016}, one can construct the block-encoding unitary for arbitrary matrix, paying the price of $\alpha=\|A\|_F$, i.e., $\alpha$ will be the Frobenius norm instead.
Note that the Frobenius norm is strictly larger than the spectral norm.

Since the outputs from each block of the transformer are vectors, we construct quantum circuits that generate quantum states corresponding to these vectors.
We use the natural format of state preparation encoding also defined in Ref.~\cite{rattew2023nonlinear}, and change the definition  from $L_2$ norm to $L_{\infty}$ norm.
\begin{definition}[State preparation encoding\label{def.stateencoding}]
We say a unitary $U_{\psi}$ is an $(\alpha,a,\epsilon)$-state-encoding of an $n$-qubit quantum state $\ket{\psi}$ if
\begin{align}
    \norm{\ket{\psi}-\alpha (\bra{0^{a}}\otimes I)U_{\psi}\ket{0^{a+n}}}_{\infty} \leq \epsilon.
\end{align} 
\end{definition}

More straightforwardly, the $(\alpha,a,\epsilon)$-state-encoding $U_\psi$ prepares the state 
\begin{align}
    U_{\psi}\ket{0}\ket{0}= \frac{1}{\alpha}\ket{0}\ket{\psi'}+\sqrt{1-\alpha^2}\ket{1}\ket{\mathrm{bad}},\notag
\end{align}
where $\norm{\ket{\psi'}-\ket{\psi}}_{\infty}\leq \epsilon$ and $\ket{\mathrm{bad}}$ is an arbitrary quantum state.
One can further prepare the state $\ket{\psi'}$ by using $\mathcal{O}(\alpha)$ times of amplitude amplification \cite{brassard2002quantum}.
The state preparation encoding may also be understood as a block encoding of a $\mathbb{C}^{2^n \times 1}$ matrix.

To encode the classical coefficients into quantum states which will be used multiple times, we follow the results in Ref.~\cite{sun2023asymptotically, zhang2022quantum}.
\begin{theorem}[Quantum state preparation \cite{sun2023asymptotically}]\label{lemma.statepreparation}
    For a given vector $v\in \mathbb{C}^N$ with $\Vert v \Vert_2=1$, one can prepare a $(1,0,0)$-state-encoding $U_v$ of state $\ket{v}=\sum_{i=1}^{N}v_i\ket{i}$ with quantum circuit depth $\mathcal{O}(N/\log N)$ without using ancilla qubits.
    One can also achieve this with depth $\mathcal{O}(\log N)$ with $\mathcal{O}(N)$ ancilla qubits.
\end{theorem}

In the following, we introduce some results on ``linear algebra" of block-encoded matrices such as addition and multiplication.
The first result is to achieve a linear combination of block-encoded matrices, which requires the so-called state preparation pair.

\begin{definition}[State preparation pair \cite{chakraborty2019power, gilyen2019quantum}]\label{LCUcoefficient.blockencoding}
Let $y\in \mathbb{C}^m$ and $\|y\|=1\leq \beta$, the pair of unitaries $(P_L,P_R)$ is called a $(\beta,b,\epsilon)$-state-preparation-pair if $P_L\ket{0^{b}}=\sum_{k=1}^{2^b}c_k\ket{k}$ and $P_R\ket{0^{b}}=\sum_{k=1}^{2^b}d_k\ket{k}$ such that $\sum_{k=1}^{m}|\beta(c^{*}_kd_k)-y_k|\leq \epsilon$ and for all $k\in m+1,\dots,2^{b}$ we have $c^*_k d_k=0$.
\end{definition}
This pair of circuits allows one to create a linear combination of matrices with given coefficients as the next lemma shows.
We notice a typo in the original Lemma~52 in Ref.~\cite{gilyen2019quantum}, and fix it as follows.
\begin{lemma}[Linear combination of block-encoded matrices \cite{chakraborty2019power, gilyen2019quantum}]\label{LCU.blockencoding}
    Let $A=\sum_{k=1}^m y_kA_k$ be an $s$-qubit operator and $\epsilon>0$. Suppose that $(P_L,P_R)$ is a $(\beta,b,\epsilon_1)$-state-preparation-pair for y, and that $W=\sum_{k=1}^{m}|k\rangle\langle k|\otimes U_k+((I-\sum_{k=1}^{m}|k\rangle\langle k|)\otimes I_a\otimes I_s)$ is an $s+a+b$ qubit unitary such that for all $k\in [m]$, the unitary $U_k$ is an $(\alpha,a,\epsilon_2)$-encoding of $A_k$. Then we can implement an $(\alpha\beta,a+b,\alpha\epsilon_1+\beta \epsilon_2)$-encoding of $A$, with a single use of $W,P_R$ and $P_L^\dagger$.
\end{lemma}

The second result is to achieve a multiplication of block-encoded matrices.

\begin{lemma}[Product of block-encoded matrices \cite{chakraborty2019power, gilyen2019quantum}]\label{product.blockencoding}
    If $U$ is an $(\alpha,a,\delta)$-encoding of an $s$-qubit operator $A$, and $V$ is a $(\beta,b,\epsilon)$-encoding of an $s$-qubit operator $B$, then $(I_b\otimes U)(I_a\otimes V)$ is an $(\alpha\beta,a+b,\alpha\epsilon+\beta\delta)$-encoding of $AB$.
\end{lemma}

Given the block-encoding, one can implement polynomial functions on singular values of block-encoded matrices (or eigenvalues for blocked Hermitian matrices) using the quantum singular value transformation (QSVT) method.

\begin{theorem}[Polynomial eigenvalue transformation \cite{gilyen2019quantum}\label{theorem.qsvt}]
    Let $\delta>0$.
    Given $U$ that is an $(\alpha,a,\epsilon)$-encoding of a Hermitian matrix $A$, and a real $\ell$-degree function $f(x)$ with $|f(x)|\leq \frac{1}{2}$ for $x\in [-1,1]$, one can prepare a $(1,a+n+4,4\ell\sqrt{\epsilon/\alpha}+\delta)$-encoding of $f(A/\alpha)$ by using $\mathcal{O}(\ell)$ queries to $U$ and $\mathcal{O}(\ell(a+1))$ one- and two-qubit quantum gates.
    The description of the quantum circuit can be computed classically in time $\mathcal{O}(\mathrm{poly}(\ell,\log(1/\delta)))$.
\end{theorem}



An additional point to note is that for the classical case, they consider the row vector as described previously.
However, for the quantum case, we consider the column vector, i.e., the quantum state.
This small difference can be handled by implementing the self-adjoint of the unitary.

\section{Problem formulations \label{sec.problem}}

Here, we describe our assumptions and the problem statements that are considered for the implementation of the transformer on quantum computers.
Recall that in this paper, we focus on the inference and assume the training process has already been achieved.
The classical problems assume memory access to the inputs such as the sentence and the query, key, and value matrices. The quantum algorithms change this input assumption to a block encoding input assumption. 
The dimensions of $N$ and $d$ can be achieved by padding with zeros. 

\begin{definition}[Input assumption\label{input.assumption}]
We assume $N=2^n$ and $d=2^{\log d}$ for $n,\log d\in \mathbb{N^+}$.
For the input sequence $S\in \mathbb{R}^{N\times d}$, we assume given access to a quantum circuit $U_S$ which is an $\Sinput$-encoding of $S$.
For matrices $W_q, W_k,W_v \in \mathbb{R}^{d\times d}$, assume given access to quantum circuits $U_{W_q}$, $U_{W_k}$, and $U_{W_v}$ that are $\Winput$-encodings of $W_q,W_k$ and $W_v$ respectively.
For the feed-forward neural network, we assume $(\alpha_{m},a_{m},\epsilon_{m})$-encodings $U_{M_1}$ and $U_{M_2}$ of two weight matrices $M_1\in \mathbbm R^{N_{1} \times N}$ and $M_2 \in \mathbbm R^{N_2 \times N_{1}}$.
\end{definition}

We reformulate the classical problems to the quantum version based on this input assumption.

\begin{prob}[Quantum self-attention\label{Attention.output}]
Assume the input assumption as in \cref{input.assumption}.
Define $Q\coloneqq SW_q$, $K\coloneqq SW_k$, and $V\coloneqq SW_v$.
Let the current focused token be $j\in[N]$,
the task is to construct a block-encoding of the matrix $G$ such that
\begin{align}
G_{j\star}=G^{\mathrm{soft}}_{j} \coloneqq \ab\big(\mathrm{softmax}(QK^{\top}/\alpha_{0})V)_{j\star},\label{eq.attention.quantum}
\end{align}
where $\alpha_{0}=\alpha_s^2\alpha_w^2$.
For the masked self-attention, change $G^{\mathrm{soft}}$ to $\mathrm{softmax}(QK^{\top}/\alpha_{0}+M)V$, where $M$ is the masked matrix as \cref{eq.mask}.
\end{prob}

Note that we change the scaling coefficient $\alpha_0$ for the quantum case.
Details of the explanation can be found in \cref{sec.attention}.

\begin{prob}[Quantum residual connection with layer normalization\label{prob.residualwithlayer}]
Assume the input assumption as in \cref{input.assumption}.
Assume given access to an $\Ginput$-encoding of the self-attention $G^{\mathrm{soft}}$ as Eq.~(\ref{eq.attention.quantum}).
Let the current query token be the $j$-th token. 
Construct a state preparation encoding of the state
\begin{align}
\sum_{k=1}^d \mathrm{LN}_{\gamma,\beta}(G^{\mathrm{soft}}_{j},S_{j})_k\ket{k}, \label{eq.layernormalization}
\end{align}
where $\mathrm{LN}_{\gamma, \beta}$ is as \cref{eq.Resnet}.
Here, $\gamma=1/\sqrt{d}$ and $\beta=\vec 0$.
\end{prob}

Note that standardization rescales the $L^2$-norm of the vector to be $\sqrt{d}$.
By taking $\gamma=1/\sqrt{d}$ and $\beta=\vec 0$, the $L^2$-norm will be $1$.
We consider this case to simplify our discussion, yet we also provide a general discussion in \cref{appen.residual}.

\begin{prob}[Quantum two-layer feedforward network\label{prob.ffn}]
Assume the input assumption as in \cref{input.assumption}. Given an $(\alpha,a,\epsilon)$-state-encoding $U_{\psi}$ of an $n$-qubit state $\ket{\psi}=\sum_{k=1}^N \psi_k \ket{k}$, where $\{\psi_k\}$ are real and $\norm{\psi}_2=1$,
and an activation function $\sigma$, prepare a state encoding of the quantum state $\ket \phi$ 
\begin{align}
\ket \phi = \frac{1}{C} \sum_{k=1}^{N_{2}} \bigl(M_2\cdot\sigma(M_1\cdot\psi)\bigr)_k\ket{k},
\end{align}
where $C$ is the normalization factor.
\end{prob}

\section{Main results\label{sec.main}}

In this section, we present our main technical contributions.
The first contribution is to show how to implement element-wise functions  applied to a block-encoded matrix, which plays an essential role in the quantum self-attention block.
To achieve this, we also show how to perform the Hadamard product of block-encoded matrices.
The second contribution is to clearly state the conversion between state preparation encoding and matrix block encoding, based on previous works about nonlinear amplitude transformation \cite{guo2021nonlinear,rattew2023nonlinear}.
This ensures we can implement the complex transformer architecture coherently on the quantum computer.
Based on these methods and further tricks, we describe a complete implementation of the quantum self-attention, residual connection and layer normalization, and the FNN blocks on a quantum computer.

\subsection{Element-wise function of block-encoded matrices}

In this section, we show an essential building block for our algorithm.
For a function $f:\mathbb R \to \mathbb R$ and a matrix $A \in \mathbb C^{2^n\times 2^n}$, the task is to apply the element-wise operation $f\circ(A)$. In a classical or quantum query model for the matrix elements, the solution is to apply the function after each particular element is queried. However, here we do not work in such a query model. The matrix $A$ is accessed via querying the circuit that constructs a block encoding.
This access model includes the element query model, but also includes the use of other input models such as input from a preceding subroutine. 

The key idea of our subroutines is a rather surprising concatenation of simple tricks as follows, see below for the formal results. Assume that $f$ in some range admits a polynomial approximation $g$ with some degree $\dpoly$ and some point-wise error, i.e, $f(x)\approx g(x)=\sum_{k=0}^{\dpoly} c_k x^k$. 
For each entry of the matrix inside the range, it holds that $f(A_{ij})\approx g(A_{ij})$ and thus
$[f\circ(A)]_{ij} \approx  [g\circ(A)]_{ij}$.
By definition of $g$, the entry can be expressed as
$[g\circ(A)]_{ij} = \sum_{k=0}^{\dpoly}c_k A_{ij}^k $. The next step is that, for $k > 0$, the $k$-degree monomial can be expressed using the $k$-th Hadamard product of the matrix $A^{\circ k}$, i.e., $A_{ij}^k = (A^{\circ k})_{ij}$.
Furthermore, we can relate the Hadamard product to the tensor product as follows. There exists a matrix $P$ such that
$A^{\circ k} = [P A^{\otimes_k} P^{\top}]_{\rm block}$,
where the subscript ``block" indicates that we choose the correct block of the matrix $P A^{\otimes_k} P^{\top}$. Hence, we can implement the element-wise polynomial by applying linear combination of unitaries on monomial blocks and a constant matrix such that 
\begin{equation}
[f\circ(A)]_{ij} \approx \sum_{k=0}^{\dpoly}c_k \ab[[P A^{\otimes_k} P^{\top}]_{\rm block}]_{ij}.
\end{equation}
In summary, the quantum algorithm uses a tensor-product of matrices, permutation matrices, linear combination of matrices,  and polynomial approximation to construct an elementwise application of a function to the matrix entries. 
We start with a lemma about the max-norm of a block-encoding. 

\begin{lemma}\label{lem:element-wise error bound}
    If $U$ is an $(\alpha,a,\epsilon)$-encoding of matrix $A \in \mathbb C^{2^n\times 2^n}$, we have 
    \begin{align}
        \max_{i,j \in [2^n]} \left |\alpha (\bra{0^{a}}\otimes\bra{i})U(\ket{0^{a}}\otimes \ket{j})- A_{ij}\right|\leq \epsilon.
    \end{align}
\end{lemma}
\begin{proof}
    Let $B=A- \alpha(\bra{0^a}\otimes I)U(\ket{0^a}\otimes I)$, which is a complex matrix.
    By definition, 
    \begin{align}
        \|B\|=\|A- \alpha (\bra{0^a}\otimes I)U(\ket{0^a}\otimes I)\|\leq \epsilon.\notag
    \end{align}
    By the standard \cref{lem:bound.maxnorm}, we have $\max_{i,j}|B_{ij}|\leq \Vert B\Vert \leq \epsilon$.
\end{proof}
As seen from the qualitative discussion above, we have to be able to construct the Hadamard product between matrices. Here, we consider the general case of two different matrices. 
\begin{theorem}[Hadamard product of block-encoded matrices\label{Hadamard.blockencoding}]
With $n \in \mathbbm N$ and $N=2^n$, consider matrices $A_1,A_2\in \mathbb{C}^{N\times N}$, and assume that we have an $(\alpha,a,\delta)$-encoding of matrix $A_1$ and $(\beta,b,\epsilon)$-encoding of matrix $A_2$.
We can construct an $(\alpha\beta,a+b+n,\alpha\epsilon+\beta\delta)$-encoding of matrix $A_1\circ A_2$.    
\end{theorem}

\begin{proof}
For simplicity, we first consider the perfect case without input block-encoding errors.
Let $U_{A_1}$ and $U_{A_2}$ be the $(\alpha,a,0)$- or $(\beta,b,0)$-encoding unitary of $A_1$ and $A_2$, respectively.
Note that 
\begin{align}
    (\bra{0^{a+b}}\otimes I_{2n})(I_b\otimes U_{A_1}\otimes I_n)(I_a\otimes U_{A_2}\otimes I_n)(\ket{0^{a+b}}\otimes I_{2n})=\frac{1}{\alpha \beta}A_1\otimes A_2.
\end{align}
Let $P'=\sum_{i=0}^{N-1} |i\rangle\langle i|\otimes |0\rangle\langle i|$.
As shown in Ref.~\cite{zhao2021compiling}, $P'(A_1\otimes A_2)P'^\dagger= (A_1 \circ A_2)\otimes |0\rangle \langle 0|$.
However, note that $P'$ is not a unitary.
Instead, we consider $P=\sum_{i,j=0}^{N-1} |i\rangle\langle i|\otimes |i\oplus j\rangle\langle j| $, which can be easily constructed by using $n$ CNOT gates, i.e., one CNOT gate between each pair of qubits consisting of one qubit from the first register and the corresponding qubit from the second register.
By direct computation, we have 
\begin{align}
    (I_n\otimes \bra{0^n}) P(A_1\otimes A_2)P^\dagger (I_n\otimes \ket{0^n} )=A_1\circ A_2.
\end{align}
Therefore, 
\begin{align}
(I_n \otimes \bra{0^{n+a+b}})\left((
P\otimes I_{a+b}) (I_b\otimes U_{A_1}\otimes I_n)(I_a\otimes U_{A_2}\otimes I_n)(P^\dagger\otimes I_{a+b})\right)(I_n\otimes\ket{0^{n+a+b}})=\frac{1}{\alpha\beta}A_1\circ A_2.
\end{align}
Now we consider the error from the input block encodings.
Write $\Bar{A}_1\coloneqq \alpha\bra{0^a}U_{A_1}\ket{0^a}$ and $\Bar{A}_2\coloneqq \beta\bra{0^b}U_{A_2}\ket{0^b}$.
Let $B_1= A_1-\Bar{A}_1$ and $B_2=A_2-\Bar{A}_2$.
By definition, $\|B_1\|\leq \delta$, $\|B_2\|\leq \epsilon$.
The error can be bounded by
\begin{align}
    &\norm*{A_1\circ A_2 - \alpha\beta (\bra{0^{n+a+b}})\left((P\otimes I_{a+b}) (I_b\otimes U_{A_1}\otimes I_n)(I_a\otimes U_{A_2}\otimes I_n)(P^\dagger\otimes I_{a+b})\right)(\ket{0^{n+a+b}})}\notag\\
    \leq{}& \norm*{A_1\circ A_2- \alpha\beta \bra{0^n}\left(P(\bra{0^{a+b}}(I_b\otimes U_{A_1}\otimes I_n)(I_a\otimes U_{A_2}\otimes I_n)\ket{0^{a+b}})P^\dagger\right)\ket{0^n}}\notag\\
    \leq{}& \norm*{A_1\circ A_2- \bra{0^n}\left(P \Bar{A}_1 \otimes \Bar{A}_2 P^\dagger\right)\ket{0^n}}\notag\\
    \leq{}& \norm*{A_1\circ A_2 +\bra{0^n}\left(P A_1 \otimes \Bar{A}_2 P^\dagger\right)\ket{0^n} - \bra{0^n}\left(P A_1 \otimes \Bar{A}_2 P^\dagger\right)\ket{0^n}-\bra{0^n}\left(P\Bar{A}_1 \otimes \Bar{A}_2 P^\dagger\right)\ket{0^n}}\notag\\
    \leq{}& \norm*{A_1\circ A_2- \bra{0^n}\left(P A_1 \otimes \Bar{A}_2 P^\dagger\right)\ket{0^n}}+ \norm*{\bra{0^n}\left(P A_1 \otimes \Bar{A}_2 P^\dagger\right)\ket{0^n}-\bra{0^n}\left(P\Bar{A}_1 \otimes \Bar{A}_2 P^\dagger\right)\ket{0^n}}\notag\\
    \leq{}& \norm*{\bra{0^n}\left(P A_1 \otimes B_2 P^\dagger\right)\ket{0^n}}+ \norm*{\bra{0^n}\left(P B_1 \otimes \Bar{A}_2 P^\dagger\right)\ket{0^n}}\notag\\
    \leq{}& \alpha\epsilon+\beta\delta.
\end{align}
\end{proof}

The previous lemma can be implemented iteratively.
Given an $(\alpha,a,\epsilon)$-encoding of matrix $A$, for $j \in \mathbb{N}>0$,
one can construct an $(1,j a+(j-1)n,j\epsilon/\alpha)$-encoding of matrix $(A/\alpha)^{\circ j}\coloneqq (A/\alpha)\circ (A/\alpha)\circ \cdots \circ (A/\alpha)$ containing $j-1$ Hadamard products among $j$ copies of matrix $A/\alpha$.
In the following, we describe how to implement the polynomials element-wisely onto the block encoded matrix by combining the Hadamard product with linear combination of unitaries \cite{childs2012hamiltonian}.

\begin{theorem}[Element-wise polynomial function of block-encoded matrix\label{elementfunction.blockencoding}]
Let $n,k \in \mathbbm N$.
Given access to an $(\alpha, a,\epsilon)$-encoding $U_A$ of a matrix $A\in \mathbb{C}^{2^n\times 2^n}$ and an $\ell$-degree polynomial function $f_{\ell}(x)=\sum_{j=1}^{\ell} c_j x^j$, $c_j \in \mathbb C$ for $j \in [l]$, one can construct a $(C,b,\gamma)$-encoding of $f_{\ell}\circ(A/\alpha)$ by using \naixu{$\mathcal{O}(\ell)$} times the input unitary, where $C\coloneqq \sum_{j=1}^{\ell} |c_j| $, $b\coloneqq \ell a+(\ell-1)n+2\log \ell$, and $\gamma\coloneqq \frac{\epsilon}{\alpha}\cdot (\sum_{j=1}^\ell |c_j|j)$.
For polynomial function $g_\ell(x)=\sum_{j=0}^{\ell} c_j x^j$ with constant term $c_0$, one can construct a $(C',b,\gamma)$-encoding of $g_{\ell}\circ (A/\alpha)$, where $C'=N c_0+C$.
\end{theorem}
    
\begin{proof}
We first consider the perfect case, i.e., $\epsilon=0$.
To achieve this implementation, we construct two state-preparation unitaries, which act on $\lceil \log(\ell+1)\rceil$ qubits such that
\begin{align}
P_L:\ \ket{0^{\lceil \log(\ell+1)\rceil}}&\to\frac{1}{\sqrt{C}}\sum_{j=1}^{\ell} \sqrt{|c_j|} \ket{j},\\
P_R:\ \ket{0^{\lceil \log(\ell+1)\rceil}}& \to\frac{1}{\sqrt{C}}\sum_{j=1}^{\ell} \sqrt{|c_j|}e^{i\theta_j} \ket{j},
\end{align}
where $C=\sum_{j=1}^{\ell} |c_j|$ and $|c_j| e^{i\theta_j}=c_j$.
By \cref{lemma.statepreparation}, $P_L$ and $P_R$ can be prepared with depth $\mathcal{O}(\ell)$ using only elementary quantum gates.
Therefore, by \cref{LCUcoefficient.blockencoding}, $(P_L,P_R)$ is a $(C,2\log \ell,0)$ state-preparation pair of $(c_1,\dots, c_\ell)$.

Now, we describe how to construct the unitary $W=\sum_{j=1}^{\ell}|j\rangle \langle j|\otimes U_{A^j} +(I_{2\log \ell}-\sum_{j=1}^{\ell}|j\rangle \langle j|)\otimes I_{\ell a+\ell n}$, where $U_{A^j}$ is a block encoding of $A^{\circ j}$.
Similar to Lemma~$8$ in \cite{Childs_2017}, instead of preparing block encodings of $A^{\circ j}$ for all $j\in [\ell]$, it suffices to prepare block encodings of $A^{\circ 2^j}$ for $j\in \lfloor \log N \rfloor$.
For $j>0$, we can construct a $(1,ja+(j-1)n,0)$-encoding $U_{A^j}$ of $(A/\alpha)^{\circ j}$ by iteratively applying Theorem~\ref{Hadamard.blockencoding}.
Combining these together, we need to use $\mathcal{O}( \sum_{j=1}^{\lfloor \log \ell \rfloor} 2^j)=\mathcal{O}(\ell)$ times of $U_A$ to construct $(\ell a+(\ell-1)n+2\log \ell)$-qubit unitary $W$.
By \cref{LCU.blockencoding}, we can implement a $(C,\ell a+(\ell-1)n+2\log \ell, 0)$-encoding of $f_{\ell}\circ(A/\alpha)$.

To implement element-wise functions including constant term, we also need access to the block encoding of a matrix whose elements are all $1$.
Notice that this matrix can be written as the linear combination of the identity matrix and the reflection operator, i.e.,
  \begin{align}
    \sum_{k,k'}|k\rangle\langle k|= \frac{N}{2}\left(I_n-(I_n- \frac{2}{N} \sum_{k k'}\ket {k}\bra{k'})\right)=\frac{N}{2}(I_n-H^{\otimes n}\left(I_n-2 \ket { 0^n}\bra{ 0^n} \right)H^{\otimes n}).\label{eq.sum_all_one}
  \end{align}
Define $U_{\mathrm{ref}}= |0\rangle\langle 0|\otimes I_n+|1\rangle\langle 1| \otimes (H^{\otimes n} (I_n-2|0\rangle\langle 0|)H^{\otimes n})$.
By direct computation, one can show that $U_0=(XH\otimes I_n)U_{\mathrm{ref}}(H\otimes I_n)$ is an $(N,1,0)$-encoding of $\sum_{k,k'}|k\rangle\langle k|$.
One can achieve the element-wise function by following the same steps as above.
One point to notice is that we can only construct $(N,1,0)$-encoding of the matrix whose elements are all $1$ since the spectral norm of this matrix is $N$. We encode $N c_0$ into the state instead of $c_0$. 

Now we perform the error analysis.
As mentioned, for each $(A/\alpha)^{\circ j}$, the error is bounded by $j\epsilon/\alpha$.
Summing up these errors, the error of $f_\ell\circ(A/\alpha)$ can be bounded by $\frac{\epsilon}{\alpha}\cdot (\sum_{j=0}^\ell |c_j|j)\eqqcolon\gamma$.
\end{proof}

How to use polynomial functions to approximate many useful functions has been well studied in the field of approximation theory.
Those results have also been utilized in the quantum computing field for QSVT-based quantum algorithms via \textit{quantum signal processing} \cite{low2017optimal}.
Note that here, we only consider functions with no constant term.


\subsection{Conversion between state preparation encoding and matrix block encoding}

Typically for each block in the transformer, the input is a vector $\psi$ and the output is another vector $f(\psi)$ in the same dimension with some nonlinear transformations.
As the quantum analog, the question becomes given a state-encoding unitary of some input state $\ket{\psi}$, output a state-encoding unitary of the state $\ket{f(\psi)}$.

To achieve this, we use the diagonal block encoding developed in the context of the nonlinear amplitude transformation method, which has been introduced in Ref.~\cite{guo2021nonlinear, rattew2023nonlinear}.
The key insight 
of the nonlinear amplitude transformation 
is that it can convert a state preparation encoding as in \cref{def.stateencoding} to a matrix block encoding as \cref{def.blockencoding}.
Then, by \cref{theorem.qsvt} one can implement polynomial functions onto these amplitudes.
For our discussion, we directly describe the robust version, which is a straightforward generalization of previous works.
The proof is provided in \cref{appen.nonlinearamplitude}.

\begin{theorem}[Robust amplitude encoding \cite{guo2021nonlinear, rattew2023nonlinear}\label{block encoding.amplitudes}]
    Given an $(\alpha,a,\epsilon)$-state-encoding $U_{\psi}$ of an $n$-qubit state $\ket{\psi}=\sum_{j=1}^{N} \psi_j \ket{j}$, where $\{\psi_j\}$ are real and $\norm{\psi}_2=1$,
    one can construct an $(\alpha,2a+n+2,\epsilon)$-encoding of the diagonal matrix $A=\diag(\psi_1, \dots, \psi_{N})$ with $\mathcal{O}(n)$ circuit depth and $\mathcal{O}(1)$ queries to controlled-$U$ and controlled-$U^\dagger$.
    One can also construct an $(\alpha^2, 3a+2n+2, 3\epsilon)$-encoding of diagonal matrix $A_{abs}=\diag(\psi_1^2,\dots, \psi_{N}^2)$.
\end{theorem}

The reason why we slightly changed the definition of state preparation encoding compared to Ref.~\cite{rattew2023nonlinear}, i.e., from $L_2$ norm to $L_{\infty}$ norm, is that after robust amplitude encoding, the $L_{\infty}$ distance between the target state $\ket{\psi}$ and exact preparable state $\ket{\psi'}$ is directly the upper bound of $\norm{\diag(\psi_1,\dots,\psi_N)-\diag(\psi_1',\dots,\psi_N')}$.

After implementing functions with QSVT, one needs to convert the block-encoding back to the state-encoding. 
This can be achieved by either the uniform-weighted \cite{guo2021nonlinear} or the importance-weighted \cite{rattew2023nonlinear} method.
The first one is more general, yet the latter one can achieve a much better, i.e., up to exponentially better, dependency on the state dimension.
A point to note is about the error analysis.
We have the error bound in matrix norm for block-encoding, which is also an upper bound for each matrix element difference, as \cref{lem:element-wise error bound}.
However, in general, the column/row of the block-encoded matrix is not normalized in the $L_2$ norm, so we also need to consider the influence of the normalization factor.
We prove the following lemma, where the proof is provided in \cref{appen.ineqnormalizeerror}.

\begin{lemma}\label{lem:inf_norm_bound_for_normalized_vectors}
    For two $d$-dimensional vectors $\psi=(\psi_1,\dots,\psi_d)$ and $\psi'=(\psi'_1,\dots,\psi'_d)$, if $|\psi_j-\psi'_j|\leq \epsilon$ for each $j\in [d]$, we have 
    \begin{align}
        \norm[\Big]{\frac{1}{C}\psi - \frac{1}{C'}\psi'}_\infty \leq \frac{(\sqrt{d}+1)\epsilon}{C} +\sqrt{\frac{2\epsilon\sqrt{d}}{C}}=\mathcal{O}\left(\sqrt{\frac{\epsilon\sqrt{d}}{C}}\right),
    \end{align}
    where $C=\norm{\psi}_2$ and $C'=\norm{\psi'}_2$.
\end{lemma}

As an example, one can easily see the following stands using lemma \cref{lem:inf_norm_bound_for_normalized_vectors}.

\begin{remark}
Given an $(\alpha,a,\epsilon)$-encoding $U_A$ of a matrix $A\in \mathbb{C}^{d \times d}$, for $U_i:\ket{0}\rightarrow\ket{i}$ where $i\in [d]$, $U_A(U_i\otimes I_a)$ is a $(\mathcal{O}(\alpha/C),a,\mathcal{O}((\epsilon \sqrt{d}/C)^{\frac{1}{2}}))$-state-encoding of $\frac{1}{C}\sum_{j=1}^d A_{ji}\ket{j}$, where $C=\|A_{\star i}\|_2$.
\end{remark}

\subsection{Quantum self-attention\label{sec.attention}}
In this section, we describe how to achieve the quantum self-attention block.
Given the block encoding of matrices as input and let $j$-th token be the current query vector, the output is a block encoding unitary of a matrix whose $j$-th row is the same as the output of the classical transformer.
We divide the task into two parts: the first part is to achieve the softmax function;
the second part is to achieve the remaining procedures. 

We provide two methods to implement the softmax function: one is based on the element-wise function as \cref{elementfunction.blockencoding}, and the other one is based on the nonlinear amplitude transformation as \cref{block encoding.amplitudes}.
In the main part, we follow the results based on the element-wise function.
The key insight for achieving the softmax function via this method is that it can also be understood that we first implement $\exp\circ(Q K^{\top}/\alpha_{0})$, then multiply with different coefficients (normalization) for each row.
Detailed analysis for the nonlinear amplitude transformation based method and comparisons are provided in \cref{attention.nonlinear}.

For quantum self-attention, we set the scaling factor $\alpha_0=\alpha_s^2\alpha_w^2$ for the following reasons.
The first is that the $1/\sqrt{d}$ is chosen somehow in a heuristic sense, and there are already some classical works considering different scaling coefficients which may even achieve better performance \cite{yang2022tensor, ma2024era}.
The second, which is more important, is that the quantum input assumption using the block encoding format naturally contains the normalization factor $\alpha$ which plays a similar role to the scaling factor.
Therefore, for the quantum case in the context of our work, it suffices to use $\alpha$ directly.

\begin{theorem}[Quantum softmax for self-attention]\label{attention.softmax}
    Given an $(\alpha,a,\epsilon)$-encoding $U_A$ of a matrix $A\in \mathbb{R}^{N\times N}$, a positive integer $d\in \NN^+$, and an index $j\in [N]$,
    one can prepare a $\ab\big(1, \mathcal{O}(\ell(a+n)),\cO\ab\big(\sqrt[^4]{\frac{N}{Z_j}}\sqrt{\epsilon}))$-state-encoding of the state
    \[
    \ket{A_j}\coloneqq \sum_{k=1}^N \sqrt{\mathrm{softmax}\left(A/\alpha\right)_{jk}} \ket{k}=
    \frac{1}{\sqrt{Z_j}}\sum_{k=1}^N \exp\circ\ab\Big(\frac{A}{2\alpha})_{jk}\ket{k},\] 
    by using $U_A$ for $\cO\ab\big(\sqrt{\frac{N}{Z_j}}\ell)$ times, where $Z_j=\sum_{k=1}^N \exp\circ(A/\alpha)_{jk}$, and $\ell=\mathcal{O}\ab\big(n\log(\frac{1}{\epsilon}))$.
\end{theorem}

\begin{proof}
    We first construct the block encoding of $\exp\circ(\frac{A}{2\alpha})$.
    Note that Taylor expansion of $\exp(x)$ contains a constant term $1$.
    This can be achieved with \cref{elementfunction.blockencoding} and \cref{approximation.exp}.
    Here, since we are only focusing on the $j$-th row, instead of taking linear combination with the matrix whose elements are all $1$, we take sum with the matrix whose $j$-th row elements are all $1$ and else are $0$.
    This enables us to have a better dependency on $N$, i.e., from $N$ to $\sqrt{N}$.
    For index $j\in [N]$, let $U_j:\ket{0}\rightarrow \ket{j}$.
    One can achieve this by changing \cref{eq.sum_all_one} to the following,
    \begin{align}
        \sum_{k}|j\rangle\langle k|=\frac{\sqrt{N}}{2}(U_j H^{\otimes n}-U_j\left(I_n -2 \ket { 0^n}\bra{0^n} \right)H^{\otimes n}).
    \end{align}
    Following the same steps in \cref{elementfunction.blockencoding}, one can achieve the construction.
    There are two error terms in this step.
    Note that by \cref{def.blockencoding}, $|A/\alpha|_{jk}\leq 1$ for $j,k\in [N]$.
    The first term comes from the intrinsic error of block encodings, and the second is from the polynomial approximation.
    Denote $U_{f\circ(A)}$ as the constructed block encoding unitary.
    By \cref{elementfunction.blockencoding} and some additional calculation, one can show that $U_{f\circ(A)}$ is a $(C_{f}, b_{f}, \gamma_{f})$-encoding of $f_\ell\circ(A)$, where $C_{f}=\sqrt{N}+\sum_{j=1}^\ell 1/j!=\mathcal{O}(\sqrt{N})$, $b_{f}=\ell a+(\ell-1)n+2\log \ell$, and $\gamma_{f}=\frac{\epsilon}{\alpha} \cdot \sum_{j=1}^\ell 1/(j-1)!=\mathcal{O}(\epsilon/\alpha)$. By triangle inequality, we have 
    \begin{align}
        &\ab\|\exp\circ\ab\Big(\frac{A}{2\alpha})_{j\star}-C_{f}\bra{0^{b_{f}}}U_{f\circ(A)}\ket{0^{ b_{f}}}\|\notag\\
        ={}& \ab\|\exp\circ\ab\Big(\frac{A}{2\alpha})-f_\ell\circ(A)+f_\ell\circ(A)-C_{f}\bra{0^{b_{f}}}U_{f\circ(A)}\ket{0^{ b_{f}}}\| \notag\\
        \leq{}& \ab\|\exp\circ\ab\Big(\frac{A}{2\alpha})-f_\ell\circ(A)\|+\ab\|f_\ell\circ(A)-C_{f}\bra{0^{b_{f}}}U_{f\circ(A)}\ket{0^{ b_{f}}}\|\notag \\
        \leq{}& \ab\|\exp\circ\ab\Big(\frac{A}{2\alpha})-f_\ell\circ(A)\| + \gamma_{f}.
    \end{align}

    Note that we can bound for each element between $\exp\circ(\frac{A}{2\alpha})$ and $f_\ell\circ(A)$ with error $\delta$, which comes from the polynomial approximation.
    By the norm inequality between spectral and Frobenius norm, we have 
    \begin{align}
        \ab\|\exp\circ\ab\Big(\frac{A}{2\alpha})-f_k\circ(A)\|&\leq \ab\|\exp\circ\ab\Big(\frac{A}{2\alpha})-f_k\circ(A)\|_F\notag \\
        &= \ab\bigg(\sum_{j,k} \ab\Big|\exp\circ\ab\Big(\frac{A}{2\alpha})_{jk}-f_\ell\circ(A)_{jk}|^2)^{\frac{1}{2}}\notag\\
        &\leq \ab(N^2 \delta^2)^{\frac{1}{2}}\leq N\delta.
    \end{align}
    To ensure the error bounded by $\epsilon$, we set $\ell=\cO\ab\big(\log(\frac{N}{\epsilon}))=\mathcal{O}\ab\big(n\log(\frac{1}{\epsilon}))$. 
    By \cref{lem:bound.maxnorm}, we have
    \begin{align}
        \max_{j,k \in [N]} \ab\Big|\exp\circ\ab\Big(\frac{A}{2\alpha})_{jk} - C_{f}(\bra{0^{b_{f}}} \bra{i})U_{f\circ(A)}(\ket{0^{ b_{f}}} \ket{j})| &\leq \ab\|\exp\circ\ab\Big(\frac{A}{2\alpha})-C_{f}\bra{0^{b_{f}}}U_{f\circ(A)}\ket{0^{ b_{f}}}\|\notag\\
        &\leq \epsilon + \gamma_f=\mathcal{O}(\epsilon).
    \end{align}
Note that $\exp\circ(\frac{A}{2\alpha})_{jk}=\exp\circ(\frac{A}{2\alpha})^{\top}_{kj}$.    
With unitary $U_{f\circ(A)}^\dagger (I\otimes U_j)$ and amplitude amplification, one can prepare a state that is close to the target state
\begin{align}
    \ket{A_j}\coloneqq\frac{1}{\sqrt{Z_j}}\sum_{k=1}^N \exp\circ\ab\Big(\frac{A}{2\alpha})_{jk}\ket{k},
\end{align}
where $Z_j=\sum_{k=1}^{N} \exp\circ(A/\alpha)_{jk}$ is the normalization factor of softmax function for the $j$-th row.
By \cref{lem:inf_norm_bound_for_normalized_vectors}, the $L_\infty$ distance between the prepared and the target state is $\cO\ab\big((\epsilon\sqrt{N/Z_j})^\frac{1}{2})$.
Therefore, $U_{f\circ(A)}^\dagger (I\otimes U_j)$ is an $\ab\big(\cO(\sqrt{N/Z_j}), b_f, \cO\ab\big((\epsilon\sqrt{N/Z_j})^\frac{1}{2}))$-state-encoding of state $\ket{A_j}$.
By using amplitude amplification~\cite{brassard2002quantum} $\mathcal{O}(\sqrt{N/Z_j})$ times, one can prepare a $(1,b_f,\cO\ab\big((\epsilon\sqrt{N/Z_j})^{\frac{1}{2}})$-state-encoding of state $\ket{A_j}$.
\end{proof}

Then we use the quantum softmax function to implement the block encoding of the self-attention matrix, as shown in the following theorem.

\begin{theorem}[Quantum self-attention\label{attention.attention}]
    Consider the setting as in \cref{Attention.output}.
    Let $\alpha_0=\alpha_s^2\alpha_w^2$.
    For the index $j\in [N]$, one can construct an $\ab\big(\alpha_s\alpha_w, \mathcal{O}(\ell(n+a_s+a_w)), \cO\ab\big(\alpha_s\alpha_w\sqrt[^4]{\frac{N}{Z_j}}\sqrt{\epsilon_s+\epsilon_w}))$-encoding of a matrix $G$ such that $G_{j\star}=G^{\mathrm{soft}}_{j}\coloneqq(\mathrm{softmax}\left(\frac{QK^{\top}}{\alpha_{0}}\right)V)_{j\star}$,
    by using $\mathcal{O}(\sqrt{\frac{N}{Z_j}}\ell)$ times of $U_S,U_{W_q},U_{W_k}$ and $U_{W_v}$, where $Z_j=\sum_{k=1}^N \exp\circ(QK^{\top}/\alpha_{0})_{jk}$, and $\ell=\mathcal{O}(n\log(\frac{1}{\epsilon_s + \epsilon_w}))$.
\end{theorem}

\begin{proof}
    In the first step, we construct the block encoding of matrix $QK^{\top}$ and $V$.
    Note that for a real matrix $M$ and its block encoding unitary $U_M$, $U_M^\dagger$ is the block encoding of $M^{\top}$.
    By \cref{product.blockencoding}, one can construct an $\QKinput$-encoding $U_{QK^{\top}}$ of $QK^{\top}$, where $\alpha_{0}\coloneqq \alpha_s^2\alpha_w^2$, $a_{0}=2a_s+2a_w$, and $\epsilon_{0}=2\alpha_s\alpha_w^2\epsilon_s+2\alpha_s^2\alpha_w\epsilon_w$.
    One can also construct an $\Vinput$-encoding $U_V$ of $V$, where $\alpha_v=\alpha_s\alpha_w$, $a_v = a_s + a_w$, and $\epsilon_v = \alpha_s \epsilon_w + \alpha_w \epsilon_s$.

    By \cref{attention.softmax}, using $U_{QK^{\top}}$ for $\cO\ab\big(\sqrt{\frac{N}{Z_j}}\ell)$ times, one can prepare a $(1, 2n+3b_f+2,\cO\ab\big(((\epsilon_s+\epsilon_w)\sqrt{N/Z_j})^{\frac{1}{2}})$-state-encoding of the state 
    \[\sum_{k=1}^N \sqrt{\mathrm{softmax}(QK^{\top}/\alpha_0)_{jk}}\ket{k},\]
    where $Z_j=\sum_{k=1}^N \exp\circ(QK^{\top}/\alpha_{0})_{jk}$, $\ell=\mathcal{O}\ab\big(n\log(\frac{1}{\epsilon_s+\epsilon_w}))$, $b_{f}=\ell a_{0}+(\ell-1)n+2\log \ell$, and $\gamma_{f}=\frac{\epsilon_0}{\alpha_0}\cdot \sum_{j=1}^\ell \frac{1}{(j-1)!}=\mathcal{O}(\epsilon_s+\epsilon_w)$.
    Recall that state encoding is also a block encoding.
    By \cref{Hadamard.blockencoding}, one can construct a $(1,\mathcal{O}(\ell(n+a_s+a_w)), \cO\ab\big(((\epsilon_s+\epsilon_w)\sqrt{N/Z_j})^{\frac{1}{2}})$-encoding of a matrix whose $j$-th column is $(\mathrm{softmax}(QK^{\top}/\alpha_0)_{j1},\dots, \mathrm{softmax}(QK^{\top}/\alpha_0)_{jN})$ ignoring other columns.
    By \cref{lem:element-wise error bound}, the absolute difference for each element is also bounded by $\cO\ab\big(((\epsilon_s+\epsilon_w)\sqrt{N/Z_j})^{\frac{1}{2}})$.
    Let this block-encoding unitary be $U_{f(QK^{\top})}$.
    
    Finally, we implement the matrix multiplication with $V$.
    This is easily achieved by \cref{product.blockencoding}, with $U_{f(QK^{\top})}^\dagger$ and $U_V$, and the error will be $\cO\ab\big(\alpha_s\alpha_w\sqrt[^4]{\frac{N}{Z_j}}\sqrt{\epsilon_s+\epsilon_w})$.
    In total, this needs $\mathcal{O}(\sqrt{\frac{N}{Z_j}}\ell)$ times of $U_{S}, U_{W_q}, U_{W_k}$ and $U_{W_v}$.
\end{proof}

Now we consider how to implement the \textit{masked} self-attention, which is essential for the decoder-only structure.
This can be achieved by slightly changing some steps as introduced in previous theorems.

\begin{corollary}
[Quantum masked self-attention]\label{cor:masked_attention.prehalf}
    Consider the same as \cref{Attention.output}.
    Let $\alpha_0=\alpha_s^2\alpha_w^2$.
    For the index $j\in [N]$, one can construct an $\ab\big(\alpha_s\alpha_w, \mathcal{O}(\ell(n+a_s+a_w)), \cO\ab\big(\alpha_s\alpha_w\sqrt[^4]{\frac{2^{\lceil \log j\rceil}}{Z_j}}\sqrt{\epsilon_s+\epsilon_w}))$-encoding of a matrix $G^{\mathrm{mask}}$ such that $G^{\mathrm{mask}}_{j\star}=(\mathrm{softmax}(\frac{QK^{\top}}{\alpha_{0}}+M)V)_{j\star}$,
    by using $\mathcal{O}(\sqrt{\frac{N}{Z_j}}\ell)$ times of $U_S,U_{W_q},U_{W_k}$ and $U_{W_v}$, where $M$ is the masked matrix as \cref{eq.mask}, $Z_j=\sum_{k=1}^N \exp\circ(\frac{QK^{\top}}{\alpha_{0}}+M)_{jk}$, and $\ell=\mathcal{O}(n\log(\frac{1}{\epsilon_s+\epsilon_w}))$.
\end{corollary}
\begin{proof}
    To achieve the masked self-attention, we slightly change the steps mentioned in \cref{attention.softmax}.
    First, about approximating the exponential function,
    instead of taking linear combination with the matrix whose $j$-th row elements are all $1$ and others are $0$, we further consider only the first $2^{\lceil \log j\rceil}$ elements in $j$-th row are $1$.
    Note that this matrix can be achieved similarly as the original one.
    The encoding factor of this matrix is $2^{\lceil \log j\rceil/2}$.
    Second, after approximating the function,
    for index $j\in [N]$, we multiply the block encoding with a projector $\sum_{k:k\leq j}\ket{k}\bra{k}$ to mask the elements.
    Though the projector $\sum_{k\in \mathcal{S}} \ket{k}\bra{k}$ for $S\subseteq [N]$ is not unitary in general, one can construct a block encoding of the projector by noticing that it can be written by the linear combination of two unitaries:
    \begin{align}
        \sum_{k\in \mathcal S} \ket{k}\bra{k}=\frac{1}{2}I+\frac{1}{2}\ab\Big(2\sum_{k\in \mathcal S} \ket{k}\bra{k}-I).
    \end{align}
    Define $U_{\rm proj}\coloneqq\ket{0}\bra{0}\otimes I+|1\rangle\langle 1|\otimes (2\sum_{k\in \mathcal S} |k\rangle\langle k|-I)$.
    One can easily verify that $(H\otimes I)U_{\rm proj}(H\otimes I)$ is a $(1,1,0)$-encoding of $\sum_{k\in \mathcal S} |k\rangle\langle k|$, where $H$ is the Hadamard gate.
    The following steps follow the same with \cref{attention.softmax} and \cref{attention.attention}.
    Complexity analysis can be derived by direct computation.
\end{proof}

One may further achieve the multi-head self-attention case by using the linear combination of unitaries.
We do not describe further details on multi-head attention in this work.
For simplicity, in the following, we will directly say we have a $(\alpha_g,a_g,\epsilon_g)$-encoding of $G$, e.g., $\alpha_g=\alpha_s\alpha_w\sqrt{N}$, $a_g=\mathcal{O}(\ell(n+a_s+a_w))$ and $\epsilon_g=\mathcal{O}\left(\alpha_s \alpha_w\sqrt[4]{\frac{N}{Z_j}}\sqrt{\epsilon_s+\epsilon_w}\right)$.

\subsection{Quantum residual connection and layer normalization}

Here, we discuss how to implement the residual connection with layer normalization as \cref{prob.residualwithlayer}.

\begin{theorem}[Quantum residual connection with layer normalization\label{theorem.residual}]
Consider the setting of \cref{prob.residualwithlayer}.
One is able to construct an $(\mathcal{O}(\sqrt{d}(\alpha_g+\alpha_s)/\varsigma), 2a_g+n+4, \mathcal{O}((\epsilon_g+\epsilon_s)/\varsigma))$-state-encoding of the state 
    \[
    \sum_{k=1}^d \mathrm{LN}(G^{\mathrm{soft}}_{j},S_{j})_k\ket{k}=\frac{1}{\varsigma}\sum_{k=1}^d (G^{\mathrm{soft}}_{jk}+S_{jk}-\Bar{s}_j)\ket{k},\]
    where $\Bar{s}_j\coloneqq \frac{1}{d}\sum_{k=1}^d (G^{\mathrm{soft}}_{jk}+S_{jk})$ and $\varsigma\coloneqq \sqrt{\sum_{k=1}^d (G^{\mathrm{soft}}_{jk}+S_{jk}-\Bar{s}_j)^2}$.
\end{theorem}

\begin{proof}
    As shown in \cref{attention.attention}, we can construct an $(\alpha_g,a_g,\epsilon_g)$-encoding of a matrix whose $j$-th row is the same row as that of $G^{\mathrm{soft}}$.
    By assumption, we are given $U_s$ which is an $\Sinput$-encoding of $S$.
    By Lemma~\ref{LCU.blockencoding} with state preparation pair $(P,P)$ such that 
    \begin{equation}
        P\ket{0} = \frac{1}{\sqrt{\alpha_g+\alpha_s}}(\sqrt{\alpha_g}\ket{0}+\sqrt{\alpha_s}\ket{1}),
    \end{equation}
    one can construct a quantum circuit $\Ures$ which is an $(\alpha_g+\alpha_s, a_g+1, \epsilon_g+\epsilon_s)$-encoding of an $N \times d$ matrix whose $j$-th row is the same as that of $G^{\mathrm{soft}}+S$.
    
    Now we consider how to create a block encoding of a diagonal matrix $\Bar{s}_j\cdot I$, where $\Bar{s}_j\coloneqq \frac{1}{d}\sum_{k=1}^d (G^{\mathrm{soft}}_{jk}+S_{jk})$. Let us define a unitary $H_{\log d} \coloneqq H^{\otimes \log d}$.
    Note that $H_{\log d}$ is a $(1,0,0)$-encoding of itself, and the first column of $H_{\log d}$ is $\frac{1}{\sqrt{d}}(1,\dots,1)^{\top}$.
    By \cref{product.blockencoding}, one can multiply $\Gsoft+S$ with $H_{\log d}$ to construct an $(\alpha_g+\alpha_s,a_g + 1,\epsilon_g+\epsilon_s)$-encoding of an $N \times d$ matrix, whose $(j,1)$-element is $\sqrt{d}\Bar{s}_i$.
    One can further move this element to $(1,1)$ by switching the first row with the $j$-th row.
    By tensoring with the identity $I$ of $\log d$ qubits, one can construct an $(\alpha_g+\alpha_s,a_g+n+1,\epsilon_g+\epsilon_s)$-encoding of $\sqrt{d}\Bar{s}_i\cdot I$.

    With $U_j:\ket{0}\rightarrow \ket{j}$, one can prepare the state 
    \begin{align}
        \Ures^\dagger (I\otimes U_j)\ket{0}\ket{0}=\frac{1}{\alpha_g+\alpha_s}\ket{0}\sum_{k=1}^d \psi'_k\ket{k}+\sqrt{1-\frac{\sum_k \psi_k'^2}{(\alpha_g+\alpha_s)^2}}\ket{1}\ket{\mathrm{bad}},
    \end{align}
    where $|\psi_k'-(G^{\mathrm{soft}}_{jk}+S_{jk})|\leq \epsilon_g+\epsilon_s$ for $k\in [d]$.
    By \cref{block encoding.amplitudes}, this can be converted to an $(\alpha_g+\alpha_s, 2a_g+n+3,\epsilon_g+\epsilon_s)$-encoding of the diagonal matrix $\diag(G_{j1}+S_{j1},\dots,G_{jd}+S_{jd})$.
    
    By \cref{LCU.blockencoding} with state preparation pair $(P_1,P_2)$, where
    \begin{equation}
        P_1\ket{0}=\frac{1}{\sqrt{1+1/\sqrt{d}}}(\ket{0}+\frac{1}{\sqrt{d}}\ket{1})
    \end{equation}
    and 
    \begin{equation}
        P_2\ket{0}=\frac{1}{\sqrt{1+1/\sqrt{d}}}(\ket{0}-\frac{1}{\sqrt{d}}\ket{1}),
    \end{equation}
    one can construct an $((\alpha_g+\alpha_s)(1+1/\sqrt{d}), 2a_g+n+4, (\epsilon_g+\epsilon_s)(1+1/\sqrt{d}))$-encoding of $\mathrm{diag}(G_{j1}+S_{j1}-\Bar{s}_j,\dots,G_{jd}+S_{jd}-\Bar{s}_j)$.
    
    Let this unitary be $\Uln$.
    Then the unitary $\Uln(I\otimes H_{\log d})$ is an $(\mathcal{O}(\sqrt{d}(\alpha_g+\alpha_s)/\varsigma), 2a_g+n+4, \mathcal{O}((\epsilon_g+\epsilon_s)/\varsigma))$-state-encoding of the state 
    \[\frac{1}{\varsigma}\sum_{k=1}^d (G^{\mathrm{soft}}_{jk}+S_{jk}-\Bar{s}_j)\ket{k},\]
    where $\varsigma\coloneqq \sqrt{\sum_{k=1}^d (G^{\mathrm{soft}}_{jk}+S_{jk}-\Bar{s}_j)^2}$.
\end{proof}

\subsection{Quantum feedforward network}

We turn our attention to the third main building block of the transformer architecture, the feed-forward neural network. This block often is a relatively shallow neural network with linear transformations and ReLU activation functions \cite{vaswani2017attention}. More recently, activation functions such as the GELU have become popular, being continuously differentiable. 
We highlight that they are ideal for quantum Transformers, since the QSVT framework requires functions that are well approximated by polynomial functions.
Functions like $\relu(x)=\max (0,x)$ can not be efficiently approximated.
The GELU is constructed from the error function, which is efficiently approximated as follows.


\begin{lemma}[Polynomial approximation of error function \cite{low2017quantum}\label{polyappro.error}]
    Let $\epsilon>0$.
    For every $k>0$, the error function $\mathrm{erf}(kx)\coloneqq \frac{2}{\sqrt{\pi}}\int_{0}^{kx} e^{-t^2} \,dt$ can be approximated with error up to $\epsilon$ by a polynomial function with degree $\mathcal{O}(k\log(\frac{1}{\epsilon}))$.
\end{lemma}
This lemma, implies the following efficient approximation of the GELU function with polynomials.
\begin{corollary}[Polynomial approximation of GELU function\label{thm.gelu}]
Let $\epsilon>0$ and $\lambda \in \mathcal{O}(1)$.
For every $k>0$ and $x\in [-\lambda, \lambda]$, the $\mathrm{GELU}$ function $\mathrm{GELU}(kx)\coloneqq kx\cdot \frac{1}{2}(1+\mathrm{erf}(\frac{kx}{\sqrt{2}}))$ can be approximated with error up to $\epsilon$ by a polynomial function with degree $\mathcal{O}(k\log(\frac{k\lambda}{\epsilon}))$.
\end{corollary}
\begin{proof}
    It suffices to approximate the error function with precision $\frac{\epsilon}{k\lambda}$ by \cref{polyappro.error}.
\end{proof}

In the following theorem, we consider how to implement the two-layer feedforward network.
As mentioned, the GELU function is widely used in transformer-based models and we explicitly consider it as the activation function in the theorem.
Cases for other activation functions like sigmoid follow the same analysis.
An example is the $\mathrm{tanh}(x)$ function, which can be well approximated by a polynomial for $x\in [-\pi/2,\pi/2]$ \cite{guo2021nonlinear}.

\begin{theorem}[Two-layer feedforward network with GELU function\label{theorem.ffn}]
Consider the setting as in \cref{prob.ffn}.
Let the activation function be $\mathrm{GELU}(x)\coloneqq x\cdot \frac{1}{2}(1+\mathrm{erf}(\frac{x}{\sqrt{2}}))$.
One can prepare an $(\mathcal{O}(\alpha\alpha_{m}^2/C), 2a+n+2 a_{m}+4, \mathcal{O}((\frac{\sqrt{N_2}}{C}\alpha^2_{m}\ell'\sqrt{\alpha_{m}\epsilon+\epsilon_{m}})^{\frac{1}{2}}))$-state-encoding of the state
\begin{align}
\ket \phi = \frac{1}{C} \sum_{k=1}^{N_{2}} \ab\Big(M_2\cdot\mathrm{GELU}(M_1\cdot\psi))_k\ket{k},
\end{align}
by using $\mathcal{\ell'}$ times of $U_{\psi}$ and $U_{\psi}^{\dagger}$,
where $C$ is the normalization factor and $\ell'= \mathcal{\tilde{O}}(\alpha \alpha_{m}\log(1/\epsilon_m))$.
\end{theorem}
\begin{proof}
Let the erroneous block-encoded matrices be $M_1'$ and $M_2'$.
We have 
\begin{align}
(I_{a}\otimes U_{M_1})(I_{a_{m}} \otimes U_{\psi}) \ket{0^{a+a_{m}+n}}=\frac{1}{\alpha \alpha_{m}}\ket{0^{a+a_{m}}}M_1'\ket{\psi'}+\ket{\widetilde{\perp}},
\end{align}
where $\ket{\widetilde{\perp}}$ is an unnormalized orthogonal state.
For the case $N_1\geq N$, this can be achieved by padding ancilla qubits to the initial state.
By direct computation, we have 
\begin{align}
&\norm{M_1\ket{\psi}-M_1'\ket{\psi'}}_{\infty}\notag\\
\leq& \norm{M_1\ket{\psi}-M_1\ket{\psi'}+M_1\ket{\psi'}-M_1'\ket{\psi'}}_{\infty} \notag\\
\leq& \norm{M_1\ket{\psi}-M_1\ket{\psi'}}_{\infty}+\norm{M_1\ket{\psi'}-M_1'\ket{\psi'}}_{\infty} \notag \\
\leq& \norm{M_1}\norm{\ket{\psi}-\ket{\psi'}}_{\infty} + \norm{M_1-M_1'}\norm{\ket{\psi'}}_{\infty} \notag \\
\leq& \alpha_{m}\epsilon+ \epsilon_{m}.
\end{align}
By \cref{block encoding.amplitudes}, one can construct an $(\alpha \alpha_{m}, a+n+2, \alpha_{m}\epsilon+ \epsilon_{m})$-encoding of matrix $\diag((M_1\psi)_1,\dots, (M_1\psi)_{N_1})$.
Note that the $\mathrm{GELU}$ function does not have a constant term, and is suitable to use the importance-weighted amplitude transformation as in Ref.~\cite{rattew2023nonlinear}.
Instead of directly implementing the GELU function, we first implement the function $f(x)=\frac{1}{2}(1+\mathrm{erf}(\frac{x}{\sqrt{2}}))$.
Note that the value of $|\mathrm{erf}(x)|$ is upper bounded by $1$.
By \cref{block encoding.amplitudes} with function $\frac{1}{4}(1+\mathrm{erf}(\alpha \alpha_{m} \frac{x}{\sqrt{2}}))$, one can construct a $(2, a+n+4, 4\ell\sqrt{\alpha_{m}\epsilon+\epsilon_{m}}+\gamma+\delta)$-encoding of matrix $\diag(f(M_1\psi)_1,\dots, f(M_1\psi)_{N_1})$, where $\ell=\mathcal{\tilde{O}}(\alpha \alpha_{m}\log(1/\gamma))$.

Let the previously constructed block-encoding unitary be $U_{f(x)}$.
We have
\begin{align}
U_{f(x)}(I\otimes U_{M_1})(I\otimes U_{\psi})\ket{0}\ket{0}=\frac{1}{2\alpha\alpha_{m}}\ket{0}\sum_k \mathrm{GELU}'(M_1'\psi')_k\ket{k}+\ket{\widetilde{\perp'}},
\end{align}
where $\ket{\widetilde{\perp'}}$ is an unnormalized orthogonal state.
Setting $\gamma,\delta=\mathcal{O}(\epsilon_m)$,
by direct computation, we have 
\begin{align}
&\| \mathrm{GELU}'(M_1'\psi')-\mathrm{GELU}(M_1\psi) \|_{\infty}\notag\\
=& \| M_1'\psi'f'(M_1'\psi')-M_1\psi f(M_1\psi)\|_{\infty} \notag\\
\leq& \|M_1'\psi'f'(M_1'\psi')-M_1'\psi' f(M_1\psi)\|_{\infty}+\|M_1'\psi'f(M_1\psi)-M_1\psi f(M_1\psi)\|_{\infty}  \notag\\
\leq& \alpha_{m}(4\ell\sqrt{\alpha_{m}\epsilon+\epsilon_{m}}+\gamma+\delta)+\alpha_{m}\epsilon+\epsilon_{m}=\mathcal{O}(\alpha_{m}\ell\sqrt{\alpha_{m}\epsilon+\epsilon_{m}}).
\end{align}
Finally, by implementing the block-encoding unitary $U_{M_2}$, we have 
\begin{align}
&(I\otimes U_{M_2})(I\otimes U_{f(x)})(I\otimes U_{M_1})(I\otimes U_{\psi})\ket{0}\ket{0}\notag\\
=&\frac{C'}{2\alpha\alpha^2_{m}}\ket{0}\frac{1}{C'}\sum_j \psi_{\mathrm{fin}}\ket{j}+\ket{\widetilde{\perp''}},
\end{align}
where $C'$ is the exact normalization factor, $\norm{\psi_{\mathrm{inf}}-M_2\mathrm{GELU}(M_1\psi)}_{\infty}=\mathcal{O}(\alpha^2_{m}\ell'\sqrt{\alpha_{m}\epsilon+\epsilon_{m}}+\epsilon_{m})=\mathcal{O}(\alpha^2_{m}\ell'\sqrt{\alpha_{m}\epsilon+\epsilon_{m}})$,
and $\ket{\widetilde{\perp}''}$ is an unnormalized orthogonal state.
By \cref{lem:inf_norm_bound_for_normalized_vectors}, we have 
\begin{align}
\Bigl\|\frac{1}{C'}\psi_{\mathrm{inf}} -\frac{1}{C} M_2\mathrm{GELU}(M_1\psi) \Bigl\|_{\infty}=\mathcal{O}\left(\left(\frac{\sqrt{N_2}}{C}\alpha^2_{m}\ell'\sqrt{\alpha_{m}\epsilon+\epsilon_{m}}\right)^{\frac{1}{2}}\right).
\end{align}
\end{proof}

\subsection{Quantum single-layer transformer}

Combining the previous results, one can obtain the following result.
Note that for a single-layer transformer, we mean the same as \cref{figure_transformer}, i.e., combined with a self-attention block, a two-layer feedforward network, and two residual connection with layer normalization blocks.

\begin{theorem} 
[Quantum single-layer Transformer]\label{thmTransformer}
Let the input assumptions be as in \cref{input.assumption}.
If $\epsilon_s,\epsilon_w,\epsilon_m=\mathcal{O}(\epsilon^8d^{-4}\alpha_m^{-14}\alpha_s^{-6}\alpha_w^{-6}\varsigma^{2}\varsigma'^{8} \sqrt{\frac{Z_j}{N}})$, then for the index $j\in [N]$,
one can construct a $(1, \mathcal{O}(\ell(n+a_s+a_w)+a_M), \epsilon)$-state-encoding of a quantum state proportional to
\begin{align}
\sum_{k=1}^d \mathrm{Transformer}(S,j)_k \ket{k},
\end{align}
by using $\mathcal{O}(d \alpha_s\alpha_w\alpha_m^3\ell\sqrt{\frac{N}{Z_j}}\frac{1}{\varsigma \varsigma'}\log(\frac{1}{\epsilon_m}))$ times of $U_S, U_{W_q}, U_{W_k}$, $U_{W_v}$ and $U_{M}$, where $\ell=\mathcal{O}(n\log(\frac{1}{\epsilon_s+\epsilon_w}))$,
$Z_j=\sum_{k=1}^N \exp\circ(QK^{\top}/\alpha_s^2\alpha_w^2)_{jk}$, and $\varsigma,\varsigma'$ are standard deviations from two layer normalization blocks. 
\end{theorem}

\begin{proof}
    As shown in \cref{figure_transformer}, a single-layer transformer contains the self-attention, residual connection and layer normalization, and the feedforward network.
    In Theorem~\ref{attention.attention}, \ref{theorem.residual} and \ref{theorem.ffn}, we have considered each block in detail.
    Here, we complete the analysis for the second residual connection after the feedforward network.

    As described in \cref{prob.ffn} and \cref{theorem.ffn}, we have access to $(\alpha, a, \epsilon)$-state-encoding of $\ket{\psi}$ and $(2\alpha\alpha_m^2, \mathcal{O}(a+n+a_m), \mathcal{O}((\sqrt{d}\alpha\alpha^3_{m}\log(\frac{1}{\epsilon_m})\sqrt{\alpha_{m}\epsilon+\epsilon_{m}})^{\frac{1}{2}}))$-encoding of matrix $B$ such that $B_{\star 1}=(\tilde{\phi}_1,\cdots,\tilde{\phi}_d)$, where $\tilde{\phi}\coloneqq M_2\cdot\mathrm{GELU}(M_1\cdot\psi)$.
    Here, the dimension of vector $\psi$ is $d$ and $N_2=d$.
    The target is to construct a state encoding of
    \begin{align}
        \sum_{k=1}^d \mathrm{LN}(\tilde{\phi}_k+\psi_k)\ket{k}.
    \end{align}
    The state encoding can be understood as a block encoding of a matrix whose first column corresponds to the quantum state.
    By \cref{LCU.blockencoding} and taking the self-adjoint, one can construct a $(2\alpha\alpha_m^2+\alpha, \mathcal{O}(a+n+a_m), \mathcal{O}((\sqrt{d}\alpha\alpha^3_{m}\log(\frac{1}{\epsilon_m})\sqrt{\alpha_{m}\epsilon+\epsilon_{m}})^{\frac{1}{2}}))$-encoding of a matrix whose first row is $(\psi_1+\tilde{\phi}_1,\dots,\psi_d+\tilde{\phi}_d)$.
    
    The following steps are the same as in \cref{theorem.residual}.
    One can construct an $(\mathcal{O}((\sqrt{d}+1)\alpha\alpha_m^2/\varsigma',\mathcal{O}(a+n+a_m),\mathcal{O}((\sqrt{d}\alpha\alpha^3_{m}\log(\frac{1}{\epsilon_m})\sqrt{\alpha_{m}\epsilon+\epsilon_{m}})^{\frac{1}{2}}/\varsigma') )$-state-encoding of the state
    \begin{align}
        \sum_{k=1}^d \mathrm{LN}(\tilde{\phi}_k+\psi_k)\ket{k},
    \end{align}
    where $\varsigma'\coloneqq\sqrt{\sum_{k=1}^d (\tilde{\phi}_k+\psi_k-\Bar{\psi})^2}$ and $\Bar{\psi}\coloneqq \frac{1}{d}\sum_{k=1}^d (\tilde{\phi}_k+\psi_k)$.

    The final result can be achieved by combining the results in Theorem~\ref{attention.attention}, \ref{theorem.residual}, and \ref{theorem.ffn}.
    Let the initial encoding error be $\epsilon_s,\epsilon_w, \epsilon_m=\mathcal{O}(\epsilon_{\mathrm{block}})$.
    In the quantum self-attention block, we output the state with error $\mathcal{O}(\alpha_s\alpha_w\sqrt[^4]{\frac{N}{Z_j}}\sqrt{\epsilon_{\mathrm{block}}})$, which is stated in \cref{attention.attention}.
    After the self-attention, we implement the quantum residual connection and layer normalization. The accumulated error is $\mathcal{O}(\alpha_s\alpha_w\sqrt[^4]{\frac{N}{Z_j}}\sqrt{\epsilon_i}/\varsigma)$, where $\varsigma$ is the standardization factor. Note that here, the normalization factor is $\mathcal{O}(\sqrt{d}\alpha_s\alpha_w/\varsigma)$. This can be seen from \cref{theorem.residual} and \cref{attention.attention}.
    Continuing to the quantum feedforward network and residual connection, described as \cref{theorem.ffn} and above, the error is 
    \begin{align}
    &\mathcal{O}\left(\left(\sqrt{d}\alpha\alpha^3_{m}\log\left(\frac{1}{\epsilon_m}\right)\sqrt{\alpha_{m}\epsilon+\epsilon_{m}}\right)^{\frac{1}{2}}/\varsigma'\right)\\
    =&\mathcal{O}\left(\left(\sqrt{d}\sqrt{d}\alpha_s\alpha_w/\varsigma\alpha^3_{m}\log\left(\frac{1}{\epsilon_m}\right)\sqrt{\alpha_{m}\alpha_s\alpha_w\sqrt[^4]{\frac{N}{Z_j}}\sqrt{\epsilon_{\mathrm{block}}}/\varsigma}\right)^{\frac{1}{2}}/\varsigma'\right)\\
    =&\mathcal{O}\left(\left(d\alpha_m^{\frac{7}{2}} \alpha_s^{\frac{3}{2}}\alpha_w^{\frac{3}{2}}\sqrt[^8]{\frac{N}{Z_j}}\sqrt[^4]{\epsilon_{\mathrm{block}}}/\sqrt{\varsigma}\right)^{\frac{1}{2}}/\varsigma'\right).
    \end{align}
    To make this bounded by $\mathcal{O}(\epsilon)$, we need to set $\epsilon_{\mathrm{block}}=\mathcal{O}(\epsilon^8d^{-4}\alpha_m^{-14}\alpha_s^{-6}\alpha_w^{-6}\varsigma^{2}\varsigma'^{8} \sqrt{\frac{Z_j}{N}})$.
\end{proof}

One can arrive the informal theorem shown in the main part by assuming $\alpha_s=\mathcal{O}(\sqrt{N}), \alpha_w=\mathcal{O}(1), \alpha_m=\mathcal{O}(1), Z_j=\Omega(N)$, and $\varsigma,\varsigma'=\Omega(1)$.

\begin{theorem}[Quantum single-layer transformer, informal\label{thmTransformer.informal_appen}]
For a transformer with embedding dimension $d$ and an input sequence $S$ of length $N$, assume that block-encoded inputs of sequence matrix and weight matrices has embedding factors $\alpha_s = \cO(\sqrt{N})$ and $\alpha_w = \cO(1)$ respectively.
For the index $j\in [N]$, one can construct a quantum circuit that prepares the state
\begin{align}
\sum_{k=1}^d \mathrm{Transformer}(S,j)_k \ket{k},
\end{align}
up to error $\epsilon$ by using ${\mathcal{\widetilde{O}}}(\sqrt{N} d \log^2(1/\epsilon))$ times of the input block encodings.  
\end{theorem}

To validate the assumptions, we provide numerical experiments in later sections.

\subsection{Output of quantum transformer}

Notice that the quantum single-layer transformer prepares a quantum state proportional to the corresponding classical vectors. For data post-processing and related applications like classification and next token prediction, we need to first translate the quantum state to a classical vector. Here we provide a detailed discussion about the output procedure.

To obtain the classical output, one can perform the quantum state tomography. 
Here, we use the $\ell_{\infty}$-norm tomography for the analysis.
Note that we change from the time complexity to the query complexity to match the analysis in this paper.
\begin{theorem}[$L^{\infty}$ state tomography \cite{Kerenidis2020Quantum}\label{thm.tomography}]
    Given access to a quantum circuit $U:\ket{0}\rightarrow \ket{\psi}$, there is a tomography algorithm that produces unit vector $\psi'\in\mathbb{R}^d$ such that $\|\psi'-\psi\|_{\infty}\leq \epsilon$ with probability at least $1-1/\mathrm{poly}(d)$ by using $\mathcal{O}(\log d/\epsilon^2)$ times of controlled-$U$.
\end{theorem}

The theorem implies that we can obtain the classical vector $\mathrm{Transformer}(S,j)$ with precision $\epsilon$ by using the quantum transformer circuit in \cref{thmTransformer.informal_appen} for $\cO(\log d/\epsilon^2)$ times. After getting the classical vector $\mathrm{Transformer}(S,j)$, one could directly follow the procedure of applying classical transformer in various tasks such as sequence classification and next token prediction.
Here, we can take $\epsilon$ as a constant, similar to the classical quantization method \cite{gholami2021surveyquantizationmethodsefficient, Ofir2019Q8BERT}, where they train the model with $16$ or $32$ bit precision, and implement the inference with $4$ or $8$ bit precision.
They show that this works well in practice and can save computational cost from low precision computation.

For the task of $k$-category sequence classification, a pre-trained linear map is applied on the classical vector $\mathrm{Transformer}(S,j)$ and give a $k$-dimension vector that indicates the classification result. The computational cost is $\cO(dk)$ from the matrix multiplication, which is negligible as the number of categories $k$ and the embedding dimension $d$ are typically much smaller than the sequence length $N$.

As for next token prediction, the predicted token is obtained by first linearly transforming the vector $\mathrm{Transformer}(S,j)$ to dimension $d_{\mathrm{token}}$ (the number of different tokens), then implementing a softmax function and sampling from the distribution. The runtime of such a procedure is $\cO(d \cdot d_{\mathrm{token}})$. Note that $d_{\mathrm{token}}$ is comparatively small to $N$, thus it does not slow down the quadratic speedup on $N$ brought by the quantum subroutine. It could be easily extended to predicting next $k$ tokens, by adding the previously predicted token to the input sequence and repeating the procedure for $k$ times.

One can implement the process for all focused tokens $j\in [N]$ to obtain the information required by the next layer's self-attention block.
Since there are $N$ tokens, one needs to repeat the algorithm $N$ times.
After reading out the classical vectors, one can reload the $Nd$ data back to qRAM and the quantum data structure.
After reloading, one can continue the computation for the next layer.
In this way, one can directly generalize to the multi-layer transformer architecture.
However, in this case the quantum complexity is $\widetilde{\mathcal{O}}(N^{\frac{3}{2}}d)$, while the classical is $\widetilde{\mathcal{O}}(N^2d+Nd^2)$.
Whether there exists a more efficient method to generalize to the multi-layer reamins as an open problem.

\subsection{Possible generalizations}

We briefly describe an extension of our work.

\textit{Trainable architecture\ ---}
For the trainability of the architecture, we require trainable parameters and a loss function. So far, we have assumed that the weights are pre-trained and made available via block-encodings. The modularity of the block-encoding framework allows to swap the assumed block encodings for parameterized block encodings, that contain trainable parameters. We provide a formal definition for a trainable block encoding here and note that the definition contains the usual variational circuits and allows for more general circuits. 

\begin{definition}[Parameterized block encoding (PBE)] \label{defPBE}
Let $\theta \in \mathbbm R^M$ where $M$ is the number of parameters, $A(\theta) \in \mathbbm C^{2^n\times 2^n}$ and $\alpha(\theta)>0$ such that $\|A(\theta)\|/\alpha(\theta) \leq 1$. 
We say a unitary $U$ is a $(\alpha(\theta),a,\epsilon)$ parameterized block encoding if $U$ is a $(\alpha(\theta),a,\epsilon)$ block encoding of $A(\theta)$.
\end{definition}
For training, the main strategy is to use the loss functions from the classical architectures \cite{vaswani2017attention} and results from tomography \cite{Kerenidis2020Quantum, huang2020predict}. While we expect that issues such as barren plateaus \cite{mcclean2018barren, larocca2024review} will appear, especially for variational PBEs, there could be room for efficient training arising from the discussed possible quantum advantages of the inference step.  We leave a discussion of PBEs and transformer architecture training for future work. It would also be interesting to consider a comparison of the more general definition of PBEs and variational circuits in light of the barren plateau issue.

\section{Discussion for quantum advantages}

\subsection{Numerical studies of quantum-relevant properties of real-world LLMs}
In this section, we provide numerical investigations of popular open-source LLMs in terms of their connection to our quantum implementation of the transformer. In particular, we focus on the key quantities that determine the run time of the quantum transformer, which arise from the given input.
There are multiple ways to construct the block encoding as given in the input assumption, which we describe in \cref{input.assumption}.
The embedding dimension $d$ is $768$ for BERT \cite{devlin2019bert}, RoBERTa \cite{liu2019roberta}, GPT \cite{radford2018improve}, DistilGPT \cite{sanh2019distilbert} and GPT2 \cite{radford2019language}; $2048$ for TinyLlama \cite{zhang2024tinyllamaopensourcesmalllanguage}; and $4096$ for both Llama2-7B \cite{touvron2023llama2openfoundation} and Mistral-7B \cite{jiang2023mistral7b}.

If it is possible to have access to the qRAM and quantum data structure, one can construct a block encoding for an arbitrary matrix, paying the price that the normalization factor will be the Frobenius norm of block encoded matrix.
\cref{sec.block-encoding}.
\begin{figure}[htbp]
  \centering
  \includegraphics[width=\columnwidth]{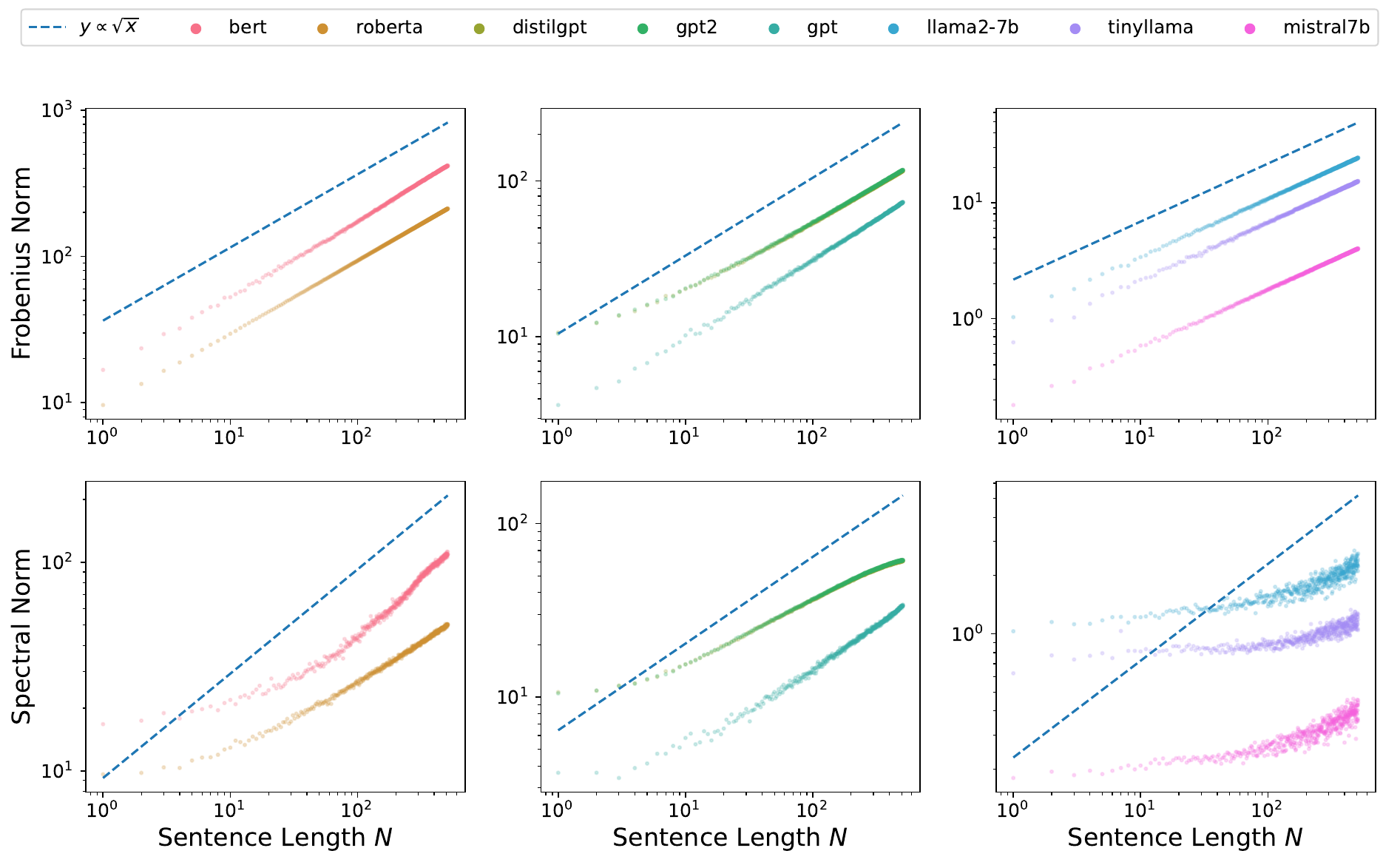}
  \caption{Scaling of the spectral norm $\|S\|$ and the Frobenius norm $\|S\|_{F}$ with $N$ for each model, displayed on logarithmic scales for both axes. For reference, the line $y \propto \sqrt{x}$ is also shown. We randomly generate tokens and convert them to $S$.}
  \label{fig:random_dataset_fro_spe}
\end{figure}
Based on this consideration and to obtain a better intuition, we numerically study several open-source large language models\footnote{Parameters are obtained from the \href{https://huggingface.co/}{Hugging Face} website, which is an open-source platform for machine learning models.}. We first investigate the spectral and Frobenius norm of the input sequence matrix $S$. To demonstrate how the norms of $S$ scale with the length $N$, we randomly sample tokens from the tokenizer that each pretrained model uses and then perform inference on the model with the generated dataset. The results are shown in \cref{fig:random_dataset_fro_spe}. The norms seen in \cref{fig:random_dataset_fro_spe} are calculated by summing the input embedding with the positional embedding, and lastly computing the respective norms on the resulting vector. We observe that the spectral norm scales almost sublinearly with $\cO(\sqrt{N})$ and the Frobenius norm scales as $\cO(\sqrt{N})$.

We also consider data in real-world applications, such as samples from the widely-used Massive Multitask Language Understanding (MMLU) dataset~\cite{hendry2021measuring} covering 57 subjects across STEM, the humanities, the social sciences, and more. The scaling of the spectral norm and the Frobenius norm of $S$ on the MMLU dataset is demonstrated in \cref{fig:MMLU_dataset_fro_spe}.
Again, the {\it DistilGPT} results almost overlap with those of {\it GPT2}. We see that in some of the models, the variances of the Frobenius norm and the spectral norm for a given $N$ are large compared to those of the random dataset. The large variances are arguably the consequence of the training in those models; the embeddings that frequently appear in the real-world dataset are actively updated at the pre-training stage, and therefore, are more broadly distributed as a result of the pre-training. In models with relatively small variance, e.g., {\it BERT}, {\it GPT}, and {\it Llama2-7b}, the spectral norm and the Frobenius norm sublinearly scale as $\mathcal{O}(\sqrt{N})$.

It is notable that the spectral norms in {\it BERT} and {\it Roberta} even decrease with the value of $N$. This can be caused by the correlations between the embeddings; the embeddings that appear in the longer sentences may be correlated with each other in those models, resulting in a smaller spectral norm.

\begin{figure}[htbp]
  \centering
  \includegraphics[width=\columnwidth]{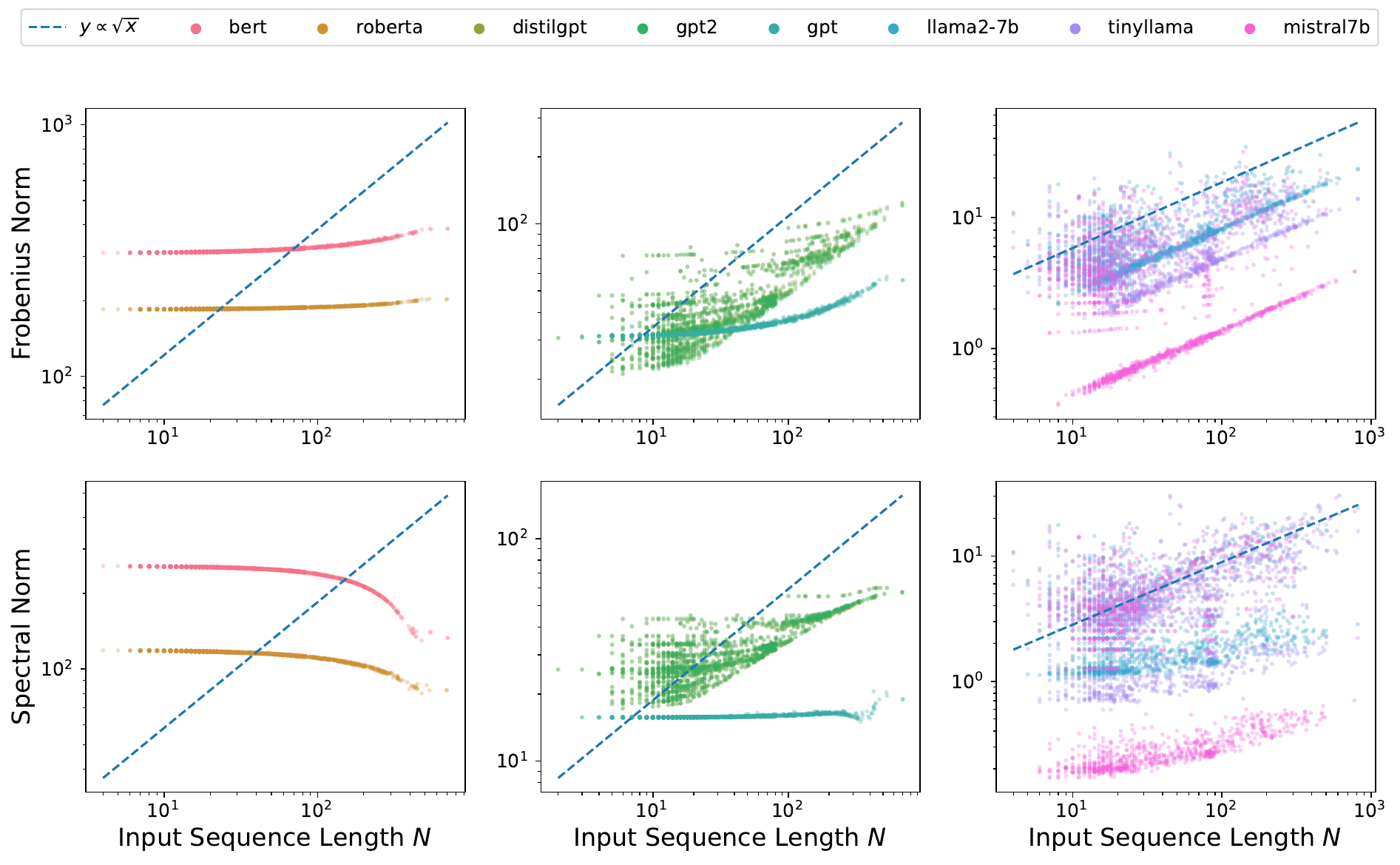}
  \caption{Scaling of the spectral norm $\|S\|$ and the Frobenius norm $\|S\|_{F}$ with $N$ for each model, displayed on logarithmic scales for both axes. For reference, the line $y \propto \sqrt{x}$ is also shown. We use tokens in the MMLU dataset and convert them to $S$.}
  \label{fig:MMLU_dataset_fro_spe}
\end{figure}

\begin{figure}[htbp]
  \centering
  \includegraphics[width=0.8\columnwidth]{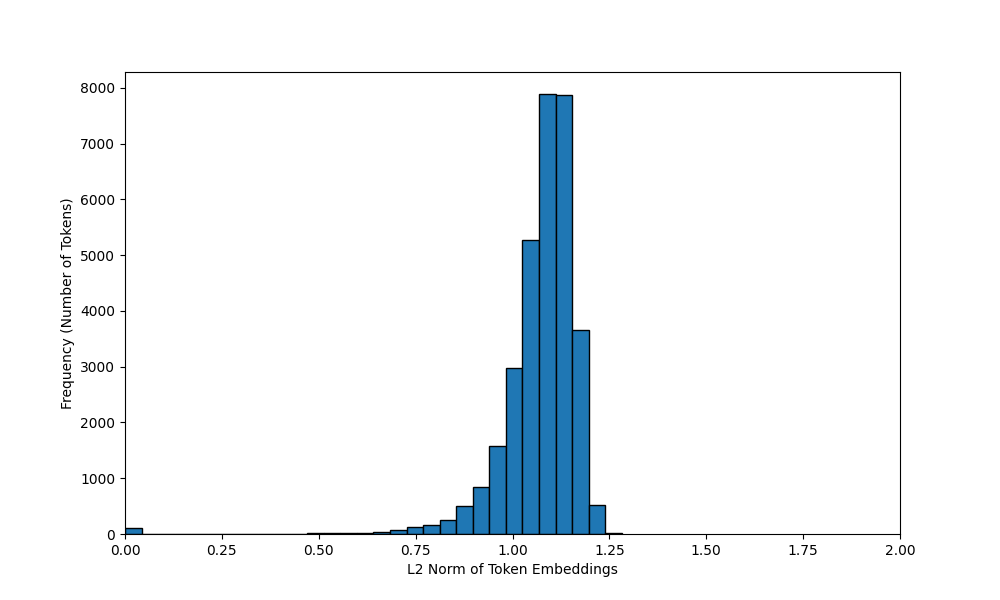}
  \caption{$L^2$ norm of token vectors in Llama2-$7b$. We compute the vector norm for all token vectors in the vocabulary of Llama2-$7b$.}
  \label{fig:llama_vector_norm}
\end{figure}

For a random matrix $S\in \mathbb{R}^{N\times d}$, the Frobenius norm in general scale as $\mathcal{O}(\sqrt{Nd})$.
From this mathematical aspect, one may wonder whether there is an additional dependency on the embedding dimension $d$ for the transformer architecture.
However, it is not the case as after training, the $L^2$-norm of each token vector is upper bounded by a constant, and independent of the dimension.
This can be observed in \cref{fig:MMLU_dataset_fro_spe}, as the Llama2-7b and Mistral-7b with $d=4096$ have smaller Frobenius norm than the other models like BERT, ROBERTA and GPT with $d=768$.
To verify this even further, we have computed the vector norm for all tokens in the vocabulary of the \textit{Llama2}-7b, shown as \cref{fig:llama_vector_norm}.
One can clearly see that the $L^2$-norm is centered around $1.1$, and upper bounded by $1.5$.
As a comparison, the embedding dimension of \textit{Llama2-7b} is $4096$.

Furthermore, for applications like retrieval-augmented generation (RAG) and other similarity estimation based tasks, token embeddings are typically 
$L^2$-normalized to unit length \cite{10.5555/3495724.3496517, karpukhin-etal-2020-dense}.
Based on these, we see that whether explicitly or implicitly, the spectral and Frobenius norm will not have dependency on the embedding dimension.

We then compute the spectral and Frobenius norms of weight matrices ($W_q,W_k,W_v$) for the large language models.
The result can be seen in \cref{Fig.LLMs}. Many of the LLMs below a dimension $d$ of $10^3$ that we have checked have substantially different norms.
We observe that for larger models such as \textit{Llama2-7b} and
\textit{Mistral-7b}, which are also current state-of-the-art open-source models, the norms do not change dramatically.
To better present the result, we compute the $L^2$-norm of column vectors inside weight matrices in various models. As shown in \cref{table.norm.appen}, there is a clear trend that as the embedding dimension $d$ increases, both the mean and variance of the $L^2$-norm of column vectors in weight matrices decrease. This trend is most apparent within the same model family of GPT2. Thus one can reasonably assume that the $L^2$-norm of column vectors is upper bounded by a constant that is independent of $d$. By direct calculation, the Frobenius norm of weight matrices scales as $\cO(\sqrt{d})$, and so does the encoding factor $\alpha_w$. Therefore, by direct computation one can  conclude that the Frobenius norm is $\mathcal{O}(\sqrt{d})$ since $W_q,W_k,W_v\in \mathbb{R}^{d\times d}$.

\begin{figure}[htbp]
\centering
\includegraphics[width=\linewidth]{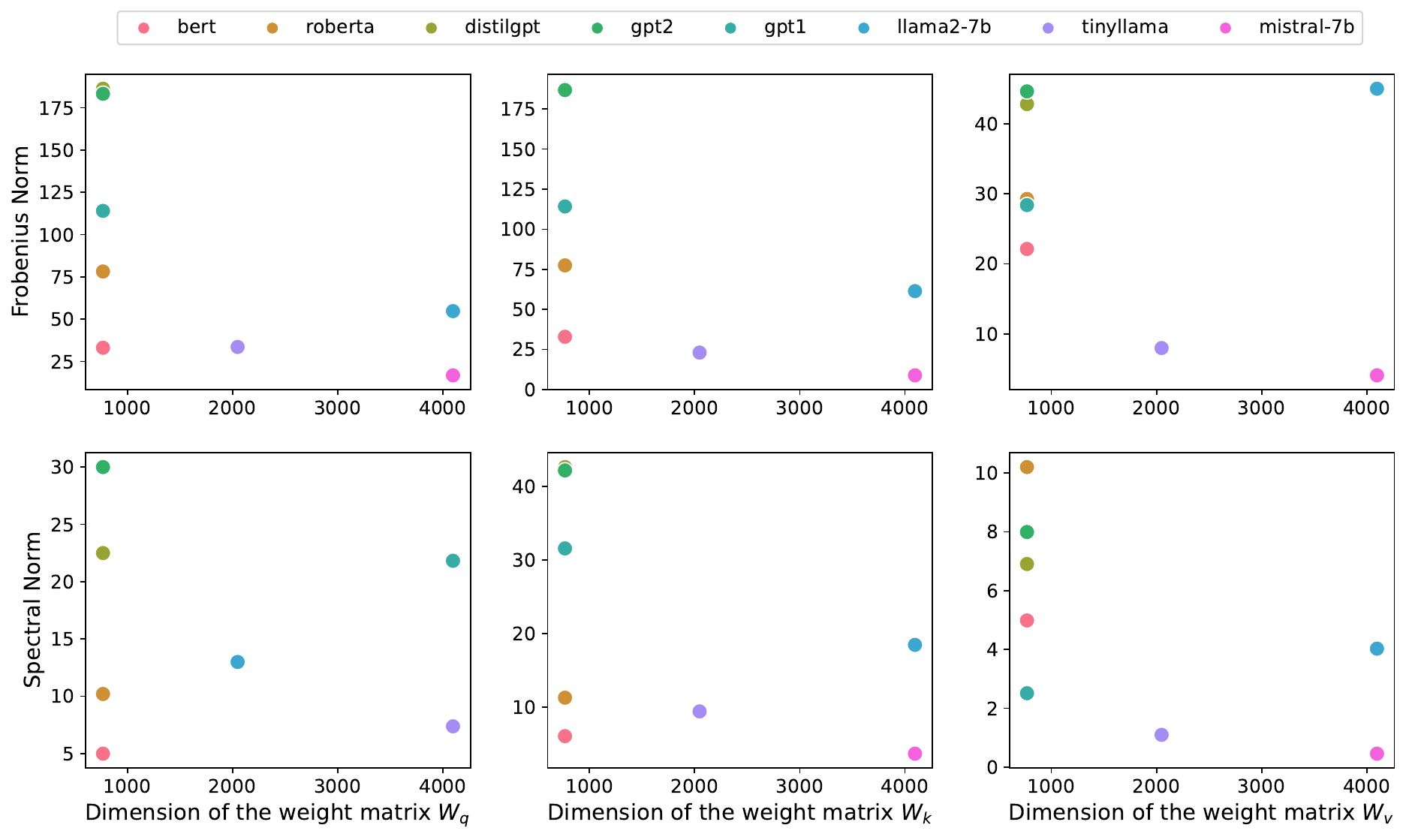}
\caption{Norms of weight matrices in popular open-source LLMs. We compute the spectral and Frobenius norms of the weight matrices $W_q, W_k$, and $W_v$ in the first layer. Note that for the multi-head self-attention, matrices have been concatenated to achieve the square matrix.}\label{Fig.LLMs}
\end{figure}

\begin{table}[htbp!]
    \centering
    \begin{tblr}{lccc}
    \toprule
    \textbf{Model} &
    \textbf{Dimension} $d$ & \textbf{Mean of $L^2$ norm}  & \textbf{Variance of $L^2$ norm}    \\
    \midrule
    GPT2 & $768$ & $3.6973$ & $1.5615$\\
    GPT2-medium & $1024$ & $3.4570$ & $0.8457$\\
    GPT2-large & $1280$ & $1.7617$ & $0.1929$\\
    GPT2-xl & $1600$ & $1.6289$ & $0.1406$\\
    \midrule
    TinyLlama & $2048$ & $0.6973$ & $0.1692$\\
    Llama2-7b  & $4096$ & $1.3486$ & $0.0901$\\
    Mistral-7b & $4096$ & $0.1576$ & $0.0047$ \\
    \bottomrule
    \end{tblr}
    \caption{The $L^2$-norm of column vectors in weight matrices from different large language models.}
    \label{table.norm.appen}
\end{table}

The ability to obtain a quantum advantage hinges on how the input is given and the particular problem. We do not provide a provable end-to-end advantage here, but rather develop the pertinent quantum subroutines and combine them into a transformer architecture. 
Given the input, our subroutines are efficient in several aspects. They use a number of working plus ancilla qubits that is logarithmic in the problem size specified by the sequence length $N$ and the embedding size $d$. The use of  amplification and its cost depends on the final task at hand. 
A regime for a possible quantum advantage is summarized in the \cref{table2}. According to our numerical observations on the spectral norm and Frobenius norm of matrices $S$, $W_q$, $W_k$, and $W_v$, the regime for the normalization factors in the table is reasonable and can be broader in possible real-world scenarios.
Based on these assumption, we obtain a number of queries to the input 
of $\widetilde{\mathcal{O}}(d^{\frac{3}{2}} \sqrt N)$. The classical run time is $\mathcal{O}(Nd +d^2)$. 
We note that the efficiency of the subroutines allows for the potential for larger speedups in other regimes.
\begin{table}[htbp!]
    \centering
    \begin{tblr}{lcc}
    \toprule
    \textbf{Quantity} &
    \textbf{Symbol}    &  \textbf{Regime} \\
    \midrule
    Softmax normalization factor & $Z_j$ & $\Omega(N)$\\
    Sequence matrix normalization &
    $\alpha_s$ & $\cO(\sqrt N)$   \\
    Attention weight matrix normalization&
$\alpha_w$ & $\cO(\sqrt d)$  \\ 
    Layer normalization factors &
    $\varsigma,\varsigma'$ & $\Omega(1)$
    \\
    Self-attention weight matrix normalization
    &
    $\alpha_w$ & $\cO(1)$\\
    FNN matrix normalization
    &
    $\alpha_m$ & $\cO(1)$\\
    Final output error & $\epsilon$ & $\Omega (1/N)$\\
    \bottomrule
    \end{tblr}
    \caption{A possible regime for the transformer where a quantum advantage could be exhibited, based on our result in \cref{thmTransformer}.}
    \label{table2}
\end{table}
With the QRAM assumption, the input block encodings can be implemented in a polylog time of $N$. In the next section, we provide detailed discussions of possible quantum advantages without QRAM.
In these cases, we obtain a quadratic speedup compared to the runtime of classical transformers.

\subsection{Training quantum-friendly transformer}

In the previous section, we have estimated the important properties for quantum algorithms directly from the classical transformer architectures.
However, a possibility remains that one may train a classical transformer that is ``friendly" to the quantum setting and has similar performance with standard architecture.
Here, for ``quantum friendly", we mean transformer architectures whose weight matrices are normalized based on spectral and Frobenius norm.

To verify this consideration, we have tested the performance for the genomic task.
We do the test on the dataset called the GenomicBenchmarks \cite{gresova2023genomic}.
We consider the promoter detection task, which can be framed as a binary classification problem to determine whether a given DNA sequence region functions as a promoter, i.e., the site where RNA polymerase and other factors bind to initiate transcription—or not.
This dataset includes $36131$ sequences, and we have used $27097$ sequences for training and $9034$ sequences for validation and testing.

We trained the standard, spectral-normalized, and Frobenius-normalized transformer model, which are all single-layer and have $10$M parameters.
All experiments run on a single NVIDIA A100 SXM4 GPU paired with an AMD EPYC 7713 processor.
We also trained a multi-layer Frobenius normalized transformer model containing $110$M parameters.
For the tokenization, we use the same as Ref.~\cite{zhou2024dnabert}.The embedding dimension is $768$, and the total vocabulary size is $4096$, including combinations of DNA letters $A,C,G,T$ and other special tokens.
We implement a linear mapping on the output of transformer to achieve the classification.
The results can be seen in \cref{table.nontata.appen}.
The performance of other models are mentioned in Ref.~\cite{gresova2023genomic, nguyen2023hyenadna}, and we list here to make comparisons.
We find that the performance of multilayer normalized transformer architecture is comparable to other advanced multilayer models.
It is therefore reasonable to normalize the weight matrices so that the encoding factor is $\alpha_w = \cO(1)$, which could make the quantum transformer run faster without much loss of performance.

\begin{table}[htbp!]
    \centering
    \begin{tblr}{lc}
    \toprule
    \textbf{Model} &
    \textbf{Nontata Accuracy}     \\
    \midrule
      Single-layer transformer & 89.1  \\
    Single-layer SN transformer & 88.4   \\ 
    Single-layer FN  transformer& 87.7   \\
    \midrule
    CNN & 85.1 \cite{nguyen2023hyenadna}\\
    HyenaDNA & 96.6 \cite{nguyen2023hyenadna} \\
    DNABERT & 92.6 \cite{gresova2023genomic}\\
    Multilayer FN transformer & 92.1 \\
    \bottomrule
    \end{tblr}
    \caption{Benchmarks of different large machine learning models on the Genomic Benchmarks (GB) dataset. ``SN'' and ``FN'' stand for spectral-normalized and Frobenius-normalized respectively.
    The multilayer FN transformer has the same size of parameters with DNABERT.
    }
    \label{table.nontata.appen}
\end{table}

\subsection{Quantum advantage without QRAM assumption\label{advantage.noqram}}

In this section, we thoroughly discuss in which cases we can achieve the quantum advantage without QRAM.
Note that $N$ often plays a dominant role in applications of the transformer. The input sequence length can keep increasing, while the dimension of token $d$ is fixed once the model has been trained.
We first consider how to implement the block encoding of $W_q, W_k, W_v$ matrices.
We numerically verified in the state-of-the-art open source models to see whether these matrices are (approximate) sparse, i.e., elements are smaller than a certain threshold value like $0.01$.
From the results, we see that these parameterized matrices are in general dense matrices.
Therefore, we consider the direct \cref{be.dense} for implementing these block encodings.
Note that the upper bound of each element of these matrices is a constant.
First, we need to store the matrix elements in the quantum registers.
As there are $\mathcal{O}(d^2)$ elements, it takes time $\mathcal{O}(d^2)$ to achieve this.
Then, we use the bucket-brigade method \cite{giovannetti2009quantum,  PhysRevA.78.052310, Olivia2020QRAM} to construct the quantum circuit for the oracle required by \cref{be.dense}. 
The quantum circuit construction requires $\mathcal{O}(d^2)$ ancilla qubits and $\mathcal{O}(\log d)$ circuit depth, i.e., after storing elements into the quantum registers, it takes $\mathcal{O}(\log d)$ time to implement the oracle.
By using \cref{be.dense}, we have $\alpha_w=\mathcal{O}(d)$ in this case. 

Next, we discuss how to implement the block encoding of the input sequence matrix $S$, which contains the $N$ dependency and is hence the more dominant part.
Similar to the parameter matrices $W_q, W_k, W_v$, we notice that open-source Large Language Models use dense encoding for the tokens, i.e., the token vector is not sparse in general. 
However, there is an alternative method called the sparse embedding \cite{formal2021spladesparselexicalexpansion, formal2021spladev2sparselexical}, which maps tokens into $k$-sparse vectors.
This method is now widely used in the Retrieval-Augmented Generation (RAG) and vector database \cite{10.5555/3495724.3496517}, which are closely related to the LLMs. 
Also, since the model size of state-of-the-art LLMs like GPT-4 \cite{openai2023gpt4} are much larger than the open-source LLMs, their dense embedding may behave in an approximate sparse way.
Therefore, we believe it is reasonable and practical to consider the case when the embedding is sparse.
Under this condition, the input sequence matrix $S$ is row-sparse, but note that may not be column-sparse.

We discuss scenarios where the column-density does not pose a problem. In particular,  when the cost of preparing a quantum state of a dense column of $N$ amplitudes is at most $O({\rm poly} \log N)$. Efficient state preparation can be attained in several special cases. Some examples are when subnormalizations are efficiently computable \cite{grover2002creatingsuperpositionscorrespondefficiently} or when the amplitudes are efficiently computable and only a small number of amplification steps are needed \cite{rattew2022preparingarbitrarycontinuousfunctions}. In the second case, the filling ratio determines the number of steps until successful state preparation. Preparation of Gaussian distributions has been discussed extensively \cite{kitaev2009wavefunctionpreparationresamplingusing, Rattew2021efficient, iaconis2024quantum}. If the sequence is generated from a linear system or an ordinary, partial, or stochastic differential equation (e.g., driven by Gaussians), there are scenarios when an efficient state preparation is possible as well \cite{harrow2009quantum, doi:10.1073/pnas.2026805118, Childs2021highprecision, An2021quantumaccelerated}.
These efficient state preparation results imply that the columns of the sequence matrix could be efficiently constructed in special cases without the use of QRAM. 

More specifically, in these cases, \cref{be.rowsparse} allows to construct the block encoding of $S$ in $\mathcal{O}(\mathrm{polylog}(N))$ time and $\alpha_s=\mathcal{O}(\sqrt{N}k)=\mathcal{O}(\sqrt{N})$.
Classical computation can not utilize these properties to improve the dependency on $N$ as in general the multiplication between a row-sparse matrix and a dense matrix is not sparse, and even if there is a computable function to generate the sequence, for the inner product and softmax function, and multiplication with $V$ there is still a linear dependency on $N$ for single-layer Transformer.

Based on the above discussion, under certain conditions that we believe are still practical for certain applications, we can obtain the quantum advantage without QRAM assumption.

\subsection{Classical randomized algorithm}

Here, we provide a discussion about the classical randomized algorithm.
Similar to the QRAM input assumption in quantum algorithms, the classical randomized algorithms assume the sample and query (SQ) access.
In the dequantization literature, people find that for some quantum machine learning algorithms like quantum recommendation system \cite{kerenidisQuantumRecommendationSystems2016} and quantum principal component analysis \cite{lloyd2014quantum}, they can not achieve exponential quantum advantage when compared with classical randomized algorithms \cite{tang2019quantuminspire, Tang_2021} in the low rank regime.
The regime has been further generalized to the extreme sparse case \cite{Gharibian_2023}, when the sparsity is constant.
However, there remains a large polynomial gap between the quantum and classical algorithms.

In the following, we analyze the classical randomized algorithm for the self-attention block based on dequantization techniques.
We follow the useful subroutines in the dequantized algorithm introduced in~\cite{tang2019quantuminspire, gilyen2018quantuminspiredlowrankstochasticregression,chia2022sampling, Gilyen2022improvedquantum}. The query access to a matrix or vector is denoted as $Q(\cdot)$. The sample and query access is denoted as $SQ(\cdot)$. Then we have the following lemma on the matrix-vector multiplication via sample and query access.

\begin{lemma}[Matrix-vector multiplication via SQ access]\label{deq.matrixvector}
Let $A \in \mathbb{C}^{M \times N}$ and $x \in \mathbb{C}^N$.
Given $SQ(A)$ and $Q(x)$, one can output a sample from the vector $Ax$ with at least $1-\delta$ probability with $\cO(N^2 C(A,x)\log\frac{1}{\delta})$ query and time complexity, where
\begin{align}
    C(A,x) := \frac{\sum_{j=1}^{N} \| x_j A(\cdot,j) \|^2}{\| Ax \|^2}.
\end{align}
Further, we can compute $(Ax)_i$ with $\mathcal{O}(N)$ queries.
\end{lemma}

When $N$ is large, one may consider the following importance sampling based method.

\begin{lemma}[Approximate matrix-vector multiplication via SQ access \cite{Tang2023thesis}\label{deq.matrixvector.approx}]
Let $A \in \mathbb{R}^{M \times N}$ be a matrix and $x \in \mathbb{R}^N$ be a vector.
Given $SQ(A)$ and $Q(x)$, with probability at least $1-\delta$, we can output a vector that is 
$\epsilon$-close to $Ax$ with 
\begin{align}
    \tau = \Theta \left(\frac{\|A\|_F^2\|x\|^2}{\epsilon^2} \right)
\end{align}
query complexity.
\end{lemma}

Now we discuss how to construct the classical randomized algorithm for the self-attention. For simplicity, we assume the sample and query access to $ \mathbb{R}^{N\times d}$ matrices $Q=SW_q,K=SW_k,V=SW_v$, the Frobenius norm of $S$ being $\cO(\sqrt{N})$, and the Frobenius norms of weight matrices $Q,K,V$ being $\cO(\sqrt{d})$. For the self-attention block, we focus on the $j$-th token, the same with the quantum setting, then matrix multiplication $QK^T$ can be simplified as the matrix vector multiplication.
By \cref{deq.matrixvector}, one can sample from $(QK^T)_{j\star}$ with $\mathcal{O}(d^2 C(K^T,Q_{j\star})\log\frac{1}{\delta})$ queries, and compute $(QK^T)_{j\star}$ with $\mathcal{O}(d)$ queries.
Now suppose one could efficiently implement the row-wise softmax function $\mathrm{softmax}(QK^\top/\alpha_0)_j$ based on sample and query access to $(QK^T)_{j\star}$, considering the matrix-vector multiplication between $V\in \mathbb{R}^{N \times d}$ and $\mathrm{softmax}(QK^\top/\alpha_0)_j \in \mathbb{R}^d$, the query complexity is $\mathcal{O}(N^2)$ by \cref{deq.matrixvector}. Even if one uses \cref{deq.matrixvector.approx}, the query complexity depends on the Frobenius norm of $V$. Since we construct $V$ from $S$ and $W_v$, the query complexity can be written as $\Theta(\|S\|_F^2\|W_v\|_F^2)$. Follow the assumptions of $\|S\|_F = \cO(\sqrt{N})$, the query complexity of the classical randomized algorithm is $\widetilde{\cO}(N)$. From this one can still see an at least quadratic separation on the matrix norm  between the quantum algorithm that we proposed and the classical randomized algorithm (thus at least quadratic quantum speedup).

The softmax function may also be a challenge for the classical randomized algorithm.
To implement function onto amplitudes/vectors, the classical randomized algorithms basically use the rejection sampling method.
Note that the rejection sampling method requires the knowledge of each probability.
However, for the softmax function, one needs to estimate the partition function to get the probability.
We know that in the worst case the partition function estimation is $\#P$-hard.
Though one can use Metropolis-Hastings method, which allows us to sample without knowing the partition function, in general there is no theoretical guarantee about how many iterations are needed.
The quantum algorithm we provide in this work does not need to estimate the partition function and has no such problem.
Therefore, this may enable our quantum algorithm to be not dequantized even if in the ideal reagime, i.e, when the matrix norm is $\mathcal{O}(\mathrm{polylog}(N))$.
However, to demonstrate whether this can really enable our quantum algorithm to be not dequantized in the ideal regime requires further study and remains as an open problem.

\section{Technical tools}

\subsection{Construction of block encoding unitaries\label{sec.block-encoding}}

In this section, we summarize some methods to construct a block-encoding unitary.
The first method is applicable to sparse matrices.
As mentioned in \cite{tay2022efficient}, there are many works considering the sparsification of attention matrices. Quantum may also benefit from these results.

\begin{lemma}[Block-encoding of sparse-access matrices \cite{gilyen2019quantum}]\label{be.sparse}
Let $A \in \mathbb{C}^{N \times N}$ ($N=2^n$)  be a matrix that is \(s_{r}\)-row-
sparse and \(s_{c}\)-column-sparse, and each element of \(A\) has absolute value at most $1$.
Suppose that we have access to the following sparse-access oracles acting on two $(n+1)$ qubit registers
\begin{align*}
    &\mathrm{O}_{r}:|i\rangle|k\rangle \rightarrow|i\rangle\left|r_{i k}\right\rangle \quad \forall i \in\left[2^{w}\right]-1, k \in\left[s_{r}\right] \text {, and }\\
    &\mathrm{O}_{c}:|\ell\rangle|j\rangle \rightarrow\left|c_{\ell j}\right\rangle|j\rangle \quad \forall \ell \in\left[s_{c}\right], j \in\left[2^{n}\right]-1 \text {, where}
\end{align*}
$r_{i j}$ is the index for the $j$-th non-zero entry of the $i$-th row of $A$, or if there are less than $i$ non-zero
entries, then it is $j+2^{n}$, and similarly $c_{i j}$ is the index for the \(i\)-th non-zero entry of the \(j\)-th column
of \(A\), or if there are less than \(j\) non-zero entries, then it is \(i+2^{n}\). Additionally assume that we have
access to an oracle \(\mathrm{O}_{A}\) that returns the entries of \(A\) in a binary description
$$
\mathrm{O}_{A}:|i\rangle|j\rangle|0\rangle^{\otimes b} \rightarrow|i\rangle|j\rangle\left|a_{i j}\right\rangle \quad \forall i, j \in\left[2^{w}\right]-1 \text {, where }
$$
\(a_{i j}\) is a b-bit binary description of the $A_{ij}$.
Then we can implement a
\(\left(\sqrt{s_{r} s_{c}}, n+3, \varepsilon\right)-\) block-encoding of \(A\) with a single use of \(\mathrm{O}_{r}, \mathrm{O}_{c}\), two uses of \(\mathrm{O}_{A}\) and additionally using \(\mathcal{O}\left(n+\log ^{2.5}\left(\frac{s_{r} s_{c}}{\varepsilon}\right)\right)\) one and two qubit gates while using \(\mathcal{O}\left(b, \log ^{2.5}\left(\frac{s_{r} s_{c}}{\varepsilon}\right)\right)\) ancilla qubits.    
\end{lemma}

Based on this lemma, one can see the following statements.

\begin{lemma}[Naive block encoding of dense matrices \cite{Nguyen2022blockencodingdense}]\label{be.dense}
    Let $A \in \mathbb{C}^{N \times N}$ ($N=2^n$) and let $\hat{a}=\max_{ij}|a_{ij}|$. Suppose we are given the oracle acting on two $(n+1)$ qubit registers
    \begin{align}
        \mathcal{O_A}:\ket{i}\ket{j}\ket{0}\rightarrow \ket{i}\ket{j}\ket{\tilde{a}_{ij}},
    \end{align}
    where $\tilde{a}_{ij}=a_{ij}/\hat{a}$. One can implement a $(N\hat{a}, n+1, \epsilon)$-encoding of $A$ with two uses of $\mathcal{O}_A$ with $\mathcal{O}(\mathrm{polylog}(\frac{N \hat{a}}{\epsilon}))$ one- and two-qubit gates, and ancilla qubits.
\end{lemma}

\begin{lemma}[Block encoding of row sparse matrices]\label{be.rowsparse}
    Let $A\in \mathbb{C}^{2^n\times 2^n}$ be a matrix that is $s_r$-row sparse, and let $\hat{a}=\max_{ij}|a_{ij}|$.
    Suppose we are given the oracles acting on two $(n+1)$ qubit registers
    \begin{align}
        \mathcal{O}_r: \ket{i}\ket{k}&\rightarrow \ket{i}\ket{r_{ik}}\\
        \mathcal{O_A}: \ket{i}\ket{j}\ket{0}&\rightarrow \ket{i}\ket{j}\ket{\tilde{a}_{ij}},
    \end{align}
    where $r_{ik}$ is the index for the $k$-th non-zero entry of the $i$-th row of $A$, and $\tilde{a}_{ij}=a_{ij}/\hat{a}$. One can implement a $(\sqrt{Ns_r}\hat{a}, n+3, \epsilon)$-encoding of $A$ with two uses of $\mathcal{O}_A$ with $\mathcal{O}(\mathrm{polylog}(\frac{N s_r\hat{a}}{\epsilon}))$ one- and two-qubit gates, and ancilla qubits.
\end{lemma}

This lemma can be directly shown by taking $s_c=N$ in \cref{be.sparse}.
The second method is for general matrices, yet we need some further assumptions which may not be easy to achieve.
\begin{lemma}[Block-encodings of matrices stored in quantum data structures \cite{kerenidisQuantumRecommendationSystems2016, gilyen2019quantum}]
Let \(A \in \mathbb{C}^{2^{n} \times 2^{n}}\).
For \(q \in[0,2]\), let us define \(\mu_{q}(A)=\sqrt{w_{q}(A) w_{(2-q)}\left(A^{T}\right)}\), where \(w_{q}(A):=\max _{i}\left\|a_{i} .\right\|_{q}^{q}\) is the \(q\)-th
power of the maximum \(q\)-norm of the rows of \(A\). Let \(A^{(q)}\) denote the matrix of the same dimensions
as \(A\), with \(A_{i j}^{(q)}=\sqrt{a_{i j}^{q}}\). 

If \(A^{(q)}\) and \(\left(A^{(2-q)}\right)^{\dagger}\) are both stored in quantum accessible data structures,
then there exist unitaries \(U_{R}\) and \(U_{L}\) that can be implemented in time \(\mathcal{O}(\operatorname{poly}(n \log (1 / \varepsilon)))\)
such that \(U_{R}^{\dagger} U_{L}\) is a \(\left(\mu_{q}(A), n+2, \varepsilon\right)\)-block-encoding of \(A\).
On the other hand, if \(A\) is stored in a quantum accessible data structure, then there exist
unitaries \(U_{R}\) and \(U_{L}\) that can be implemented in time \(\mathcal{O}(\operatorname{poly}(n \log (1 / \varepsilon)))\) such that \(U_{R}^{\dagger} U_{L}\) is
an \(\left(\|A\|_{F}, n+2, \varepsilon\right)\)-block-encoding of \(A\).
\end{lemma}

Another method that may be useful, especially for the transformer architecture is for the Gram matrix whose entries are given by the inner products.

\begin{lemma}[Block-encoding of Gram matrices by state preparation unitaries]
Let  $U_{L}$ and $U_{R}$ be state preparation unitaries acting on  $a+n$ qubits preparing the vectors $\left\{\left|\psi_{i}\right\rangle: i \in\left[2^{n}\right]-1\right\}$ and $\left\{\left|\phi_{j}\right\rangle: j \in\left[2^{n}\right]-1\right\}$ such that 
\begin{align}
    U_{L}&:|0\rangle|i\rangle \rightarrow\left|\psi_{i}\right\rangle\\
    U_{R}&:|0\rangle|j\rangle \rightarrow\left|\phi_{j}\right\rangle,
\end{align}
Then $U=U_{L}^{\dagger} U_{R}$is an $(1, a, 0)$-block-encoding of the Gram matrix A such that  $A_{i j}=\left\langle\psi_{i} | \phi_{j}\right\rangle$.    
\end{lemma}

\subsection{Robust nonlinear amplitude transformation\label{appen.nonlinearamplitude}}

\begin{theorem}[Robust amplitude encoding]
    Given an $(\alpha,a,\epsilon)$-state-encoding $U_{\psi}$ of an $n$-qubit state $\ket{\psi}=\sum_{j=1}^{N} \psi_j \ket{j}$, where $\{\psi_j\}$ are real and $\norm{\psi}_2=1$,
    one can construct an $(\alpha,2a+n+2,\epsilon)$-encoding of the diagonal matrix $A=\diag(\psi_1, \dots, \psi_{N})$ with $\mathcal{O}(n)$ circuit depth and $\mathcal{O}(1)$ queries to controlled-$U$ and controlled-$U^\dagger$.
    One can also construct an $(\alpha^2, 3a+2n+2, 3\epsilon)$-encoding of diagonal matrix $A_{abs}=\diag(\psi_1^2,\dots, \psi_{N}^2)$.
\end{theorem}

\begin{proof}

The construction is the same as Ref.~\cite{guo2021nonlinear, rattew2023nonlinear} and our focus is on the error analysis. 
The $(\alpha,a,\epsilon)$-state-encoding $U_{\psi}$ approximately prepares the state
\begin{align}
    U\ket{0}\ket{0}=\frac{1}{\alpha} \ket{0}\ket{\psi}+\sqrt{1-\alpha^2}\ket{1}\ket{\mathrm{bad}},
\end{align}
where $\ket{\mathrm{bad}}$ is a quantum state we are not interested.
By the diagonal amplitude block-encoding introduced in Ref.~\cite{guo2021nonlinear, rattew2023nonlinear},
one can approximately construct a block-encoding of $A=\diag(\psi_1,\dots,\psi_N)$.
By direct computation, one can see it is an $(\alpha,2a+n+2,\epsilon)$-encoding, where $\alpha$ is directly from the state-encoding, and the error can be obtained from the $L_{\infty}$-norm.
Let the exact block-encoded diagonal matrix be $A'$.
Note that $\|A-A'\|=\max_j |\psi_j-\psi'_j|=\|\ket{\psi}-\ket{\psi'}\|_{\infty}\leq \epsilon$.
Block-encoding of $A_{abs}$ can be constructed following Theorem $2$ in Ref.~\cite{mitarai2019quantum} and Ref.~\cite{guo2021nonlinear, rattew2023nonlinear}.
The error analysis follows $\max_{j}|\psi_j^2-\psi'_j|^2\leq \max_j|\psi_j^2-(\psi_j+\epsilon)^2|\leq 3\epsilon$.
Query complexity analysis follows the previous results.
\end{proof}

\subsection{Matrix maximum entry norm}
The standard block encoding assumption directly tells us about the matrix norm of the block-encoded matrix, i.e., $\|A\|\leq \alpha$.
With the following lemma, the condition also tells us that $\max_{i,j} |A_{ij}|\leq \alpha$, i.e., the absolute value of each element is also bounded by $\alpha$.

\begin{lemma}\label{lem:bound.maxnorm}
For a complex matrix $A\in \mathbb{C}^{n\times m}$, $\max_{i,j}|A_{ij}|\leq \|A\|$.
\end{lemma}
\begin{proof}
Let $\sigma_{\max}(A)$ be the largest singular value of $A$.
By definition, we have $\|A\|=\sigma_{\max}(A)$.
Consider the singular value decomposition $A=U\Sigma V^{\dagger}$, where $U$ and $V$ are unitaries and $\Sigma$ is a diagonal matrix.
Let $\{f_i\}_i$ and $\{g_j\}_j$ be the basis of $\mathbb{C}^{n}$ and $\mathbb{C}^m$ respectively.
Since $U$ and $V$ are unitaries, we have 
\begin{align}
\|U^\dagger f_i\|=\|V^\dagger g_j\|=1.
\end{align}
Write $v=V^\dagger g_j$. We have 
\begin{align}
\|A_{ij}\|&=|\langle  f_i, Ag_j\rangle|=|\langle  f_i, U\Sigma V^\dagger g_j \rangle| 
= |\langle U^\dagger f_i,  \Sigma V^\dagger g_j\rangle|
\leq \|U^\dagger f_i\| \|\Sigma V^\dagger g_j\|\notag\\
&=\|\ \Sigma V^\dagger g_j\|=\left(\sum_k \left(\Sigma v\right)_k^2\right)^{\frac{1}{2}}
=\left(\sum_k \left(\sum_j \Sigma_{kj} v_j\right)^2\right)^{\frac{1}{2}}
= \left(\sum_k \Sigma_{kk}^2v_k^2\right)^{\frac{1}{2}}
\notag \\
&\leq \left(\sum_k \sigma_{\max}^2 v_k^2\right)^{\frac{1}{2}}
=\sigma_{\max} \|v\|^2=\|A\|.
\end{align}
\end{proof}

\subsection{Normalized error bound}\label{appen.ineqnormalizeerror}

Here, we show some results that are useful when considering the conversion from matrix block encoding to state preparation encoding.
\begin{lemma}\label{ineq.normalizederror}
    For two $d$-dimensional vectors $\psi=(\psi_1,\dots,\psi_d)$ and $\psi'=(\psi'_1,\dots,\psi'_d)$, if $|\psi_j-\psi'_j|\leq \epsilon$ for each $j\in [d]$, we have 
    \begin{align}
        \norm[\Big]{\frac{1}{C}\psi -\frac{1}{C'}\psi'}_2 \leq \frac{2\sqrt{d}\epsilon}{C}+\sqrt{\frac{2\epsilon\sqrt{d}}{C}},
    \end{align}
    where $C=\norm{\psi}_2$ and $C'=\norm{\psi'}_2$.
\end{lemma}
\begin{proof}
    By direct computation, we have 
\begin{align}
    \norm[\Big]{\frac{1}{C}\sum_{j\in \mathcal S} \psi_j\ket{j} - \frac{1}{C'}\sum_{j\in \mathcal S} \psi_j'\ket{j}}_2
    &= \frac{1}{CC'} \norm[\Big]{C' \sum_{j\in \mathcal S} \psi_j\ket{j} - C \sum_{j\in \mathcal S} \psi'_j\ket{j}}_2 \notag\\
    &= \frac{1}{CC'} \norm[\Big]{C' \Bigl(\sum_{j\in \mathcal S} \psi_j\ket{j} - \sum_{j\in \mathcal S} \psi_j'\ket{j}\Bigr) + (C'-C) \sum_{j\in \mathcal S} \psi'_{j}\ket{j}}_2\notag\\
    &\leq \frac{1}{CC'}\left(\norm[\Big]{C' \Bigl(\sum_{j\in \mathcal S} \psi_{j}\ket{j} - \sum_{j\in \mathcal S} \psi'_{j}\ket{j}\Bigr)}_2 + \norm[\Big]{(C'-C) \sum_{j\in \mathcal S} \psi'_{j}\ket{j}}_2 \right),
\end{align}
where the inequality comes from the triangle inequality.
The first term can be easily bounded by $\sqrt{d}\epsilon/C$ since for each $j\in [d]$, we have $|\psi_j-\psi_j'|\leq \epsilon $.
Note that for nonnegative real numbers $a$ and $b$, we have $|a-b|=|(\sqrt{a}-\sqrt{b})(\sqrt{a}+\sqrt{b})|=|\sqrt{a}-\sqrt{b}||\sqrt{a}+\sqrt{b}|\geq |\sqrt{a}-\sqrt{b}|^2$, hence $|\sqrt{a}-\sqrt{b}|\leq \sqrt{|a-b|}$.
The second term can be bounded with the following computation:
\begin{align}
    \frac{1}{C}\abs{C-C'}
    &\leq \frac{\sqrt{\abs{C^2-C'^2}}}{C}\notag\\
    &\leq \frac{\sqrt{\abs{\sum_{j\in \cS} (\psi_{j}^2-\psi'^2_{j})}}}{C}\notag\\
    &\leq \frac{\sqrt{\sum_{j\in \cS} \abs{(\psi_j-\psi'_{j})(\psi_{j}+\psi'_{j})}}}{C}\notag\\
    &\leq \frac{\sqrt{\epsilon\sum_{j\in \cS} \abs{\psi_{j}+\psi'_{j}}}}{C}\notag\\
    &\leq \frac{\sqrt{\epsilon\sum_{j\in \cS} (2\abs{\psi_{j}}+\epsilon)}}{C}\notag\\
    &\leq \frac{\sqrt{d\epsilon^2 +2\epsilon\sum_{j\in \cS} \abs{\psi_{j}}}}{C}\notag\\
    &\leq \frac{\sqrt{d}\epsilon}{C}+\frac{\sqrt{2\epsilon\sum_{j\in \cS} \abs{\psi_{j}}}}{C}\notag\\
    &\leq \frac{\sqrt{d}\epsilon}{C}+\sqrt{\frac{2\epsilon\sqrt{d}}{C}},
\end{align}
where the last inequality is from the inequality between $L_1$ and $L_2$ norm.
Combining these two terms together, we achieve our final result.
\end{proof}

\begin{lemma}
    For two $d$-dimensional vectors $\psi=(\psi_1,\dots,\psi_d)$ and $\psi'=(\psi'_1,\dots,\psi'_d)$, if $|\psi_j-\psi'_j|\leq \epsilon$ for each $j\in [d]$, we have 
    \begin{align}
        \norm[\Big]{\frac{1}{C}\psi - \frac{1}{C'}\psi'}_\infty \leq \frac{(\sqrt{d}+1)\epsilon}{C} +\sqrt{\frac{2\epsilon\sqrt{d}}{C}},
    \end{align}
    where $C=\norm{\psi}_2$ and $C'=\norm{\psi'}_2$.
\end{lemma}
\begin{proof}
    Note that the $L_\infty$ distance can be written as
    \begin{equation}
        \norm[\Big]{\frac{1}{C}\psi -\frac{1}{C'}\psi'}_\infty = \max_{j\in [d]} \abs[\Big]{\frac{\psi_j}{C} - \frac{\psi'_j}{C'}}.
    \end{equation}
    We consider each element individually as
    \begin{equation}
        \abs[\Big]{\frac{\psi_j}{C} - \frac{\psi'_j}{C'}} = \frac{1}{CC'}\abs{C'\psi_j - C\psi'_j}.
    \end{equation}
    Having $\max_{j\in [d]} \abs{\psi_j-\psi'_j} \leq \epsilon$, we can write $\psi_j = \psi'_j + \Delta_j$ where $\abs{\Delta_j}\leq \epsilon$. Substituting $\psi_j$ in $\abs{C'\psi_j - C\psi'_j}$ we have
    \begin{align}
        \abs{C'\psi_j - C\psi'_j} &= \abs{C'\psi'_j + C'\Delta_j - C\psi'_j}\\
        &= \abs{(C'-C)\psi'_j + C'\Delta_j}\\
        &\leq \abs{(C'-C)\psi'_j} + C'\epsilon.
    \end{align}
    Then we can write
    \begin{align}
        \max_{j\in [d]} \abs[\Big]{\frac{\psi_j}{C} - \frac{\psi'_j}{C'}} &= \max_{j\in [d]} \frac{1}{CC'}\abs{C'\psi_j - C\psi'_j}\\
        &\leq \frac{C'\epsilon + \max_{j\in[d]} \abs{(C'-C)\psi'_j}}{CC'} \\
        &\leq \frac{\epsilon}{C} + \frac{\abs{C'-C}C'}{CC'}\\
        &= \frac{\epsilon}{C} + \frac{\abs{C'-C}}{C}\\
        &\leq \frac{(\sqrt{d}+1)\epsilon}{C} +\sqrt{\frac{2\epsilon\sqrt{d}}{C}}.
    \end{align}
    The bound of $\frac{\abs{C'-C}}{C}$ directly follows from the proof of \cref{ineq.normalizederror}.
\end{proof}

\begin{lemma}
    For two $d$-dimensional vectors $\psi=(\psi_1,\dots,\psi_d)$ and $\psi'=(\psi'_1,\dots,\psi'_d)$, if $|\psi_j-\psi'_j|\leq \epsilon$ and $\psi_j, \psi'_j \leq \Gamma \in \cO(1)$ for each $j\in [d]$, we have 
    \begin{align}
        \norm[\Big]{\frac{1}{C}\psi - \frac{1}{C'}\psi'}_\infty \leq \frac{\epsilon}{C} + \frac{\Gamma\sqrt{d}\epsilon}{CC'} + \frac{\Gamma}{C'}\sqrt{\frac{2\epsilon\sqrt{d}}{C}},
    \end{align}
    where $C=\norm{\psi}_2$ and $C'=\norm{\psi'}_2$.
\end{lemma}
\begin{proof}
    Note that the $L_\infty$ distance can be written as
    \begin{equation}
        \norm[\Big]{\frac{1}{C}\psi -\frac{1}{C'}\psi'}_\infty = \max_{j\in [d]} \abs[\Big]{\frac{\psi_j}{C} - \frac{\psi'_j}{C'}}.
    \end{equation}
    We consider each element individually as
    \begin{equation}
        \abs[\Big]{\frac{\psi_j}{C} - \frac{\psi'_j}{C'}} = \frac{1}{CC'}\abs{C'\psi_j - C\psi'_j}.
    \end{equation}
    Having $\max_{j\in [d]} \abs{\psi_j-\psi'_j} \leq \epsilon$, we can write $\psi_j = \psi'_j + \Delta_j$ where $\abs{\Delta_j}\leq \epsilon$. Substituting $\psi_j$ in $\abs{C'\psi_j - C\psi'_j}$ we have
    \begin{align}
        \abs{C'\psi_j - C\psi'_j} &= \abs{C'\psi'_j + C'\Delta_j - C\psi'_j}\\
        &= \abs{(C'-C)\psi'_j + C'\Delta_j}\\
        &\leq \abs{(C'-C)\psi'_j} + C'\epsilon.
    \end{align}
    Then we can write
    \begin{align}
        \max_{j\in [d]} \abs[\Big]{\frac{\psi_j}{C} - \frac{\psi'_j}{C'}} &= \max_{j\in [d]} \frac{1}{CC'}\abs{C'\psi_j - C\psi'_j}\\
        &\leq \frac{C'\epsilon + \max_{j\in[d]} \abs{(C'-C)\psi'_j}}{CC'} \\
        &\leq \frac{\epsilon}{C} + \frac{\Gamma\abs{C'-C}}{CC'}\\
        &= \frac{\epsilon}{C} + \frac{\Gamma\sqrt{d}\epsilon}{CC'} + \frac{\Gamma}{C'}\sqrt{\frac{2\epsilon\sqrt{d}}{C}}.
    \end{align}
\end{proof}

\subsection{Polynomial approximation of exponential function}

Here we describe how to approximate the exponential function efficiently by a polynomial for $x\in [-1,1]$.
\begin{lemma}\label{approximation.exp}
    For $x\in [-1,1]$, the function $f(x)\coloneqq  e^{x}$ can be approximated with error bound $\epsilon$ with an $\mathcal{O}(\log(1/\epsilon))$-degree polynomial function.
\end{lemma}

\begin{proof}
    Consider the Taylor expansion of  $f(x)=\sum_{j=0}^{\infty}\frac{x^j}{j!}$.
    Let $f_k(x)\coloneqq \sum_{j=0}^{k}\frac{x^j}{j!}$.
    To achieve $|f_k(x)-f(x)|\leq \epsilon$ for $|x|\leq 1$,
    \begin{align}
        |f_k(x)-f(x)|&= \left |\sum_{j=k+1}^{\infty} \frac{x^j}{j!}\right |\notag \leq |\sum_{j=k+1}^{\infty} \frac{1}{j!}|\notag = \left |\sum_{j=1}^{\infty} \frac{1}{(j+k)!}\right|\notag\\
        (\text{Assume } k>2)\quad &\leq \frac{1}{k!} \left |\sum_{j=1}^{\infty} \frac{1}{2^j}\right| \notag\leq \frac{1}{k!}\leq \epsilon.
    \end{align}
    It suffices to set $k=\mathcal{O}(\log(\frac{1}{\epsilon}))$, which can be seen by the Stirling's approximation.
\end{proof}

\subsection{Quantum softmax via nonlinear amplitude transformation\label{attention.nonlinear}}

In the following, we provide how to achieve the quantum softmax via the nonlinear amplitude transformation method, introduced in \cite{guo2021nonlinear, rattew2023nonlinear}.
Note that this is possible if we focus on the $j$-th token.

\begin{theorem}[Quantum softmax via nonlinear amplitude transformation\label{softmax.nonlinear}]
    Given an $(\alpha,a,\epsilon)$-encoding $U_A$ of a matrix $A\in \mathbb{R}^{N\times N}$, a positive integer $d\in \NN^+$, and an index $j\in [N]$,
    one can prepare a $\ab\big(1, \mathcal{O}(a+n),\mathcal{O}\left(\sqrt[4]{\frac{N\epsilon}{Z\alpha}} \right))$-state-encoding of the state
    \[
    \ket{A_j}\coloneqq \sum_{k=1}^N \sqrt{\mathrm{softmax}\left(A/\alpha\right)_{jk}} \ket{k}=
    \frac{1}{\sqrt{Z_j}}\sum_{k=1}^N \exp\circ\ab\Big(\frac{A}{2\alpha})_{jk}\ket{k},\] 
    by using $U_A$ for $\cO\ab\big(\sqrt{\frac{N}{Z_j}}\ell)$ times, where $Z_j=\sum_{k=1}^N \exp\circ(A/\alpha)_{jk}$, and $\ell=\mathcal{O}\ab\big(\log(\frac{\alpha}{\epsilon}))$.
\end{theorem}
\begin{proof}
Note that the block encoding of a matrix can be considered as a state encoding of its columns.
 We have 
 \begin{align}
     U_A^\dagger(I\otimes U_j)\ket{0}\ket{0}\approx\frac{1}{\alpha}\sum_{k} A_{jk}\ket{0}\ket{k}+\sqrt{1-\frac{1}{\alpha^2}\sum_{k} A_{jk}^2}\ket{1}\ket{\perp},
 \end{align}
 where $U_j:\ket{0}\rightarrow \ket{j}$ and $\ket{\perp}$ is some arbitrary state.
 By using \cref{block encoding.amplitudes}, one can construct a $(\alpha, 2a+n+2,\epsilon)$-encoding of matrix $\mathrm{diag}(A_{j1},\dots, A_{jN})$ by using $\mathcal{O}(1)$ times of $U_A^\dagger (I\otimes U_j)$.
 With \cref{theorem.qsvt}, one can prepare a $(1, 2a+n+4, 4\log(1/\delta)\sqrt{\epsilon/\alpha}+2\delta)$-encoding of $\frac{1}{e}\mathrm{diag}(\exp(A_{j1}/2\alpha),\dots, \exp(A_{jN}/2\alpha))$, where $\delta$ is error bound for both approximating $\frac{1}{e}e^{x/2}$ and computing circuit description.
 Here, we take $\delta=\mathcal{O}(\sqrt{\epsilon/\alpha})$ such that block encoding error can be bounded by $\mathcal{O}(\sqrt{\epsilon/\alpha})$.
 This implies that we take $\ell=\mathcal{O}(\log(\alpha/\epsilon))$-degree polynomial to approximate the function.
 Let this constructed circuit be $U_{\exp(A)}$.
 We have 
 \begin{align}
     U_{\exp(A)}(I\otimes H^{\otimes n})\ket{0}\ket{0}\approx\frac{1}{e\sqrt{N}}\ket{0}\sum_{k} \exp\left(\frac{A_{jk}}{2\alpha}\right)\ket{k}+ \ket{\widetilde{\perp}},
 \end{align}
where $\ket{\widetilde{\perp}}$ is a arbitrary unnormalized state.
One can see that it is a $(\mathcal{O}(\sqrt{N/Z}),\mathcal{O}(a+n), err )$-state encoding of the final state,
where by \cref{lem:inf_norm_bound_for_normalized_vectors} $err=\mathcal{O}\left(\sqrt[4]{\frac{N\epsilon}{Z\alpha}} \right)$.
One can further use amplitude amplitude $\mathcal{O}(\sqrt{N/Z})$ times to achieve a $(1,\mathcal{O}(a+n), err )$-state-encoding.
\end{proof}

To achieve the masked self-attention with the nonlinear amplitude transformation method follows similarly to the element-wise function case.

Here we make a comparison between \cref{attention.softmax} and \cref{softmax.nonlinear}.
Note that for a $N\times N$ matrix, in most cases the block encoding factor $\alpha$ is bounded by $\mathcal{O}({\rm poly} (N))$. This means that $\mathcal{O}(\log(\frac{\alpha}{\epsilon}))=\mathcal{O}(n\log(\frac{1}{\epsilon}))$.
One can see that the element-wise function method has the same query complexity with the nonlinear amplitude transformation method and has a better dependency for the initial error, yet it requires more ancilla qubits.
Regardless of the ancilla qubits, the element-wise function is a stronger method than the nonlinear amplitude transformation, since it can implement functions onto each element of a matrix, while the nonlinear amplitude transformation can only implement functions onto each element of a state.

\subsection{General case of quantum residual connection\label{appen.residual}}

We first provide the theorem for only quantum residual connection, which might be an additional interest.

\begin{prob}[Quantum residual connection\label{prob.resnet}]
Let $c>0$ and $g(x)$ be a real $k$-degree polynomial function.
Given an $(\alpha,a,\epsilon)$-state-encoding $U$ of a quantum state $\sum_{j=1}^{d} x_j\ket{j}$, where $\{x_j\}$ are real and $\norm{x}_2=1$, prepare a state-encoding of the state
\begin{align}
\frac{1}{\sqrt{\sum_{j=1}^d (c\cdot g(x)_j+x_j)^2}}\sum_{j=1}^{d} (c\cdot g(x)_j+x_j) \ket{j}.\label{eq.residual.target}
\end{align}
\end{prob}

\begin{theorem}[Quantum residual connection]

Consider the setting of \cref{prob.resnet}.
For the polynomial $g(x)$, let $g_{\max} \coloneqq \max_{x\in [-1,1]} \ab|g(\alpha x)|$, one can prepare an $\bigl(\mathcal{O}\ab\big(\sqrt{N}(\alpha+2cg_{\max})/C), a+n+4, \mathcal{O}\ab\big((cg_{\max}(4\ell\sqrt{\epsilon}+\delta)+\alpha \epsilon)/C)\bigr)$-state-encoding of the state $\frac{1}{C}\sum_{k=1}^N(c\cdot g(x_k)+x_k)\ket{k}$, where $C^2\coloneqq \sum_{k=1}^N(c\cdot g(x_k)+x_k)^2$.
Further, if $g(x)/x$ is bounded with $\eta\coloneqq \max_{x\in [-1,1]}\ab|g(\alpha x)/x|$, one can prepare an $\bigl(\mathcal{O}\ab(\alpha(1+2c\eta)/C), a+n+4, \mathcal{O}\ab(c\eta(4\ell\sqrt{\epsilon}+\delta)/C)\bigr)$-state-encoding instead. The preparation uses $\mathcal{O}(\ell)$ times of $U_x$ and $U_x^\dagger$.
\end{theorem}
\begin{proof}
We first discuss the general case.
Given the state-encoding $U_x$, by \cref{block encoding.amplitudes}, one can construct an $(\alpha, a+n+2,\epsilon)$-encoding of $A=\diag(x_1,\dots,x_N)$.
Let $g_{\max}\coloneqq \max_{x\in [-1,1]} \ab|g(\alpha x)|$, then by \cref{theorem.qsvt} with function $g(x)/(2g_{\max})$, one can construct a $(2g_{\max}, a+n+4, 2g_{\max}(4 \ell\sqrt{\epsilon}+\delta))$-encoding of the matrix $\diag(g(x_1),\dots, g(x_N))$. Note that the normalization factor $2g_{\max}$ is to satisfy the requirements of \cref{theorem.qsvt}.

By using the linear combination of block-encoded matrices as \cref{LCU.blockencoding} with state preparation pair $(P,P)$, where $P:\ket{0}\to 1/\sqrt{\alpha+2cg_{\max}}(\sqrt{\alpha}\ket{0}+\sqrt{2cg_{\max}}\ket{1})$,
one can construct an $(\alpha+2cg_{\max}, a+n+5, 2cg_{\max}(4\ell\sqrt{\epsilon}+\delta)+\alpha \epsilon)$-encoding $U_g$ of the matrix $\diag(c\cdot g(x_1)+x_1,\dots , c\cdot g(x_N)+x_N)$. One can easily verify that $U_g(I\otimes H_n)$ is a state-encoding of the target state $\frac{1}{C}\sum_{k=1}^N(c\cdot g(x_k)+x_k)\ket{k}$.
We have 
\begin{align}
U_g(I\otimes H_n)\ket{0}\ket{0}&=\frac{1}{\sqrt{N}(\alpha+2cg_{\max})}\ket{0}\sum_{k=1}^N \psi_k \ket{k}+\ket{\widetilde{\perp}}\notag\\
&= \frac{C'}{\sqrt{N}(\alpha+2cg_{\max})}\ket{0}\frac{1}{C'}\sum_{k=1}^N \psi_k \ket{k}+\ket{\widetilde{\perp}},
\end{align}
where $C'=\norm{\psi}_{2}$, $\norm{\psi-(c\cdot g(x)+x)}_{\infty}\leq 2cg_{\max}(4\ell\sqrt{\epsilon}+\delta)+\alpha \epsilon$, and $\ket{\widetilde{\perp}}$ is a unnormalized orthogonal state.
For simplicity, let $\epsilon_g\coloneqq 2cg_{\max}(4\ell\sqrt{\epsilon}+\delta)+\alpha \epsilon$.
By \cref{lem:inf_norm_bound_for_normalized_vectors}, the final error bound is 
\[\frac{\epsilon_g}{C}+\frac{(cg_{\max}+1)}{C'} \ab\Bigg(\frac{\sqrt{N}\epsilon_g}{C}+\sqrt{\frac{2\sqrt{N}\epsilon_g}{C}})=\mathcal{O}\ab((cg_{\max}(4\ell\sqrt{\epsilon}+\delta)+\alpha \epsilon)/C).\]

Now we consider the specific case, i.e., when the polynomial $g(x)$ has no constant term.
Note that for a polynomial $g(x)$, if $g(x)/x$ is bounded on the interval across $x=0$, it cannot have the constant term.
Instead of implementing function $g(x)/(2g_{\max})$ with quantum singular value transformation, here we implement $g'(A)/2\eta$ instead, where $g'(x)\coloneqq g(\alpha x)/x$ and $\eta\coloneqq \max_{x\in [-1,1]}|g'(x)|$.
By \cref{LCU.blockencoding} with state preparation pair $(P',P')$, where $P':\ket{0}\rightarrow 1/(\sqrt{1+2c\eta})(\ket{0}+\sqrt{2c\eta}\ket{1})$ to construct a $(1+2c\eta,a+n+4, 2c\eta(4\ell\sqrt{\epsilon}+\delta))$-encoding of diagonal matrix $I+c\cdot g'(A)$.
Let this block-encoding unitary be $U_{g'}$ and $\epsilon_{g'}\coloneqq 2c\eta(4\ell\sqrt{\epsilon}+\delta)$.
We have $U_{g'}(I\otimes U_x)$ is the $\left(\frac{\alpha(1+2c\eta)}{C''}, a+n+4, \frac{\epsilon_{g'}}{C}+\frac{(c\eta+1)}{C''}\left(\frac{\sqrt{N}\epsilon_{g'}}{C}+\sqrt{\frac{2\sqrt{N}\epsilon_{g'}}{C}}\right)\right)$-state-encoding of the target state, where $C''$ is the $L_2$ norm for the exact prepared state.
\end{proof}

For the quantum residual connection and layer normalization, in the main paper, we only mention a specific case, i.e., when $\gamma=1/\sqrt{d}$ and $\beta=0$.
If we consider the general layer normalization, the quantum state mentioned in \cref{prob.residualwithlayer} should be 
\begin{align}
    \frac{1}{C}\sum_{k=1}^d \mathrm{LN}_{\gamma,\beta}(G^{\mathrm{soft}}_{j},S_{j})_k\ket{k}, 
\end{align}
where $C$ is the normalization factor.
Since vector $\beta$ can be implemented on quantum computers via \cref{lemma.statepreparation}, and taking sum via the linear combination of unitaries, here we omit $\beta$.
Then the representation of the quantum state can be simplified as 
\begin{align}
    \frac{\gamma}{\sqrt{d}}\sum_{k=1}^d \mathrm{LN}_{\gamma,0}(G^{\mathrm{soft}}_{j},S_{j})_k\ket{k}.
\end{align}
Note that compared to the case which we consider in the main paper, there is an additional factor $\sqrt{\gamma}/\sqrt{d}\eqqcolon \gamma'$, since now the $L^2$-norm is $\gamma'$.
Now we describe how this factor will affect our analysis.
If we continue to implement the feedforward network, we need to implement the function $\mathrm{GELU}(\frac{1}{\gamma'}\cdot)$ instead of $\mathrm{GELU}(\cdot)$.
By \cref{thm.gelu}, the degree of the polynomial for approximating the $\mathrm{GELU}$ function will increase $\mathcal{O}(\frac{1}{\gamma'})$.
For the second residual connection and layer normalization which is after the feedforward network, this factor does not affect the scaling for implementing this block, but the output state will become
\begin{align}
    \gamma'\sum_{k=1}^d \mathrm{Transformer}(S,j)\ket{k}.
\end{align}
If one wants to obtain the information via quantum state tomography using \cref{thm.tomography} with final precision $\mathcal{O}(\epsilon)$, one needs to set $\delta=\mathcal{O}(\epsilon\gamma')$ in  \cref{thm.tomography}.
An specific case is when $\gamma'=1/\sqrt{d}$, i.e., $\gamma=1$.
Under such case, our results in \cref{thmTransformer} will have another factor $\sqrt{d}$.
Note that this does not affect our result as $N$ is the dominant factor rather than $d$.

\end{document}